\batchmode
\documentclass[12pt,english]{article}

\usepackage[]{graphicx}\usepackage[]{color}
\makeatletter
\def\maxwidth{ %
  \ifdim\Gin@nat@width>\linewidth
    \linewidth
  \else
    \Gin@nat@width
  \fi
}
\makeatother

\definecolor{fgcolor}{rgb}{0.345, 0.345, 0.345}

\usepackage{mathtools}

\usepackage{framed}				
\makeatletter
 {\par\unskip\endMakeFramed%
 \at@end@of@kframe}
\makeatother

\definecolor{shadecolor}{rgb}{.97, .97, .97}
\definecolor{messagecolor}{rgb}{0, 0, 0}
\definecolor{warningcolor}{rgb}{1, 0, 1}
\definecolor{errorcolor}{rgb}{1, 0, 0}
\usepackage{float}
\usepackage{alltt}
\usepackage{lmodern}

\usepackage[T1]{fontenc}
\usepackage[latin9]{inputenc}
\usepackage[letterpaper]{geometry}
\usepackage{color}
\usepackage{babel}
\usepackage{amsmath}
\usepackage{amsthm}
\usepackage{amssymb}
\usepackage{tabularx}
\usepackage{bm}
\usepackage{pifont}

\usepackage{graphicx}
\usepackage{array}
\usepackage{tikz}
\usepackage[title]{appendix}
\newcolumntype{C}[1]{>{\centering\arraybackslash}m{#1}}
\usepackage{setspace}
\usepackage[ruled,vlined]{algorithm2e}
\usepackage{breakcites}
\usepackage[unicode=true,pdfusetitle,
 bookmarks=true,bookmarksnumbered=false,bookmarksopen=false,
 breaklinks=true,pdfborder={0 0 1},backref=false,colorlinks=true]
 {hyperref}
 \hypersetup{urlcolor=blue}
\hypersetup{citecolor=blue}
 \usepackage{xr-hyper}
\usepackage{multirow}
\makeatletter
\theoremstyle{plain}
\newtheorem{assumption}{\protect\assumptionname}
\theoremstyle{plain}
\newtheorem{thm}{\protect\theoremname}
  \theoremstyle{remark}
  \newtheorem{rem}{\protect\remarkname}
  \theoremstyle{plain}
  \newtheorem{lem}{\protect\lemmaname}
  \theoremstyle{plain}
  \newtheorem{cor}{\protect\corollaryname}
    \theoremstyle{plain}
    \newtheorem{example}{\protect\examplename}
     \theoremstyle{plain}

\newtheorem{innercustomgeneric}{\customgenericname}
\providecommand{\customgenericname}{}
\newcommand{\newcustomtheorem}[2]{%
  \newenvironment{#1}[1]
  {%
   \renewcommand\customgenericname{#2}%
   \renewcommand\theinnercustomgeneric{##1}%
   \innercustomgeneric
  }
  {\endinnercustomgeneric}
}
\numberwithin{equation}{section}
\newcustomtheorem{customthm}{Theorem}
\newcustomtheorem{customlem}{Lemma}
\newcustomtheorem{customassumption}{Assumption}
\newcustomtheorem{customprop}{Proposition}
\newcustomtheorem{customexample}{Example}
\newcustomtheorem{customdef}{Definition}
\newcustomtheorem{customcor}{Corollary}
\newcustomtheorem{customrem}{Remark}

\@ifundefined{date}{}{\date{}}
\newtheorem{appxlem}{Lemma}[section]
\newtheorem{appxthm}{Theorem}[section]

\newtheorem{appxrem}{Remark}[section]

\newcommand\independent{\protect\mathpalette{\protect\independenT}{\perp}}
\def\independenT#1#2{\mathrel{\rlap{$#1#2$}\mkern2mu{#1#2}}}

\usepackage{enumitem}
\newlist{steps}{enumerate}{1}
\setlist[steps, 1]{label = Step \arabic*}

\renewcommand{\qed}{\hfill
\small$\mathscr{Q.E.D.}$\normalsize}

\allowdisplaybreaks

\usepackage{bbm}
\usepackage{chngcntr}

\usepackage{indentfirst}
\usepackage{breakcites}
\makeatother
\usepackage{mathrsfs}
\SetKwInput{KwInput}{initialize}
\SetKwInput{KwOutput}{output}
  \providecommand{\lemmaname}{Lemma}
  \providecommand{\remarkname}{Remark}
\providecommand{\corollaryname}{Corollary}
\providecommand{\theoremname}{Theorem}
\providecommand{\examplename}{Example}
\providecommand{\assumptionname}{Assumption}
\providecommand{\propositionname}{Proposition}

\IfFileExists{upquote.sty}{\usepackage{upquote}}{}

\usepackage{natbib}
\usepackage{caption}
\makeatletter
\newcommand*{\addFileDependency}[1]{
  \typeout{(#1)}
  \@addtofilelist{#1}
  \IfFileExists{#1}{}{\typeout{No file #1.}}
}
\makeatother
 
\begin{document}
\begin{titlepage}
\title{Regularized Quantile Regression with Interactive Fixed Effects
\thanks{This paper is based on the third chapter of my doctoral dissertation at Columbia. I thank Jushan Bai, Sokbae (Simon) Lee and Bernard Salani\'e, who were gracious with their advice, support and feedback. I have also greatly benefited from comments and discussions with Songnian Chen, Roger Koenker, Jos\'e Luis Montiel Olea, Roger Moon, Serena Ng,  J\"{o}rg Stoye, Peng Wang and participants at the Columbia Econometrics Colloquium. All errors are my own.}
}
\author{Junlong Feng\thanks{Department of Economics, Hong Kong University of Science and Technology, Hong Kong SAR. E-mail address: jlfeng@ust.hk.}
}
\date{\today}\maketitle
\begin{abstract}
\noindent  This paper studies large $N$ and large $T$ conditional quantile panel data models with interactive fixed effects. We propose a nuclear norm penalized estimator of the coefficients on the covariates and the low-rank matrix formed by the fixed effects. The estimator solves a convex minimization problem, not requiring pre-estimation of the (number of the) fixed effects. It also allows the number of covariates to grow slowly with $N$ and $T$. We derive an error bound on the estimator that holds uniformly in quantile level. The order of the bound implies uniform consistency of the estimator and is nearly optimal for the low-rank component. Given the error bound, we also propose a consistent estimator of the number of fixed effects at any quantile level. To derive the error bound, we develop new theoretical arguments under primitive assumptions and new results on random matrices that may be of independent interest. We demonstrate the performance of the estimator via Monte Carlo simulations.\\
\vspace{0in}
\vspace{0in}
\bigskip
\end{abstract}
\setcounter{page}{0}
\thispagestyle{empty}
\end{titlepage}

\section{Introduction}

Panel data models are widely applied in economics and finance. Allowing for rich heterogeneity, interactive fixed effects are important components in such models in many applications. In applications such as asset pricing, it could be desirable to explain or forecast an outcome variable at certain quantile levels. However, the well studied mean regression with interactive fixed effects (e.g. \cite{pesaran2006estimation}, \cite{bai2009panel} and \cite{moon2015linear}) misses such distributional heterogeneity.

In this paper, we consider a panel data model where the conditional quantile of an outcome variable is linear in the covariates and in the product of time- and individual-fixed effects. These fixed effects are unobservables that may be correlated with the covariates. The number of covariates is allowed to grow slowly to infinity with $N$ and $T$. Meanwhile, we allow the coefficients, the set of the effective fixed effects and the realization of each of them to all be quantile level dependent, generating large modeling flexibility.

To estimate the model, this paper proposes a \textit{nuclear norm penalized estimator}. By deriving its theoretical error bound, we show that the estimator can consistently estimate the coefficients and the (realizations of) the interactive fixed effects uniformly in quantile level. The estimator solves a convex problem and computes fast in practice even in large panel data sets based on our proposed augmented Lagrangian multiplier algorithm. Implementing the estimator does not require pre-estimation of the number of fixed effects or the fixed effects themselves. 

To illustrate the estimator, let us consider a simple example where only the coefficients are quantile level dependent: For a panel data set $\{(Y_{it},X_{it}):i=1,...,N;t=1,...,T\}$, suppose the $u$-th conditional quantile of $Y_{it}$ given $p$ covariates $X_{it}$ and $r$ time- and individual-fixed effects $(F_{t},\Lambda_{i})$ is $q_{Y_{it}|X_{it},F_{t},\Lambda_{i}}(u)\coloneqq X_{it}'\beta_{0}(u)+F_{t}'\Lambda_{i}$
where both $F_{t}$ and $\Lambda_{i}$ are $r\times 1$ vectors. The fixed effects form an $N\times T$ matrix $L_{0}\coloneqq (\Lambda_{1},...,\Lambda_{N})'(F_{1},...,F_{T})$ whose rank is at most $r$. Thus, the dense matrix $L_{0}$ is low-rank when $r$ is small relative to $N$ and $T$. Exploiting such low-rankness, our estimator, inspired by the seminal work by \cite{candes2009exact}, jointly estimates $(\beta_{0}(u),L_{0})$ for a given quantile level $u\in (0,1)$ by solving
\begin{equation}\label{eq.estimator.intro}
\min_{\beta\in \mathcal{B},L\in\mathcal{L}} \frac{1}{NT}\sum_{i,t}\rho_{u}(Y_{it}-X_{it}'\beta-L_{it})+\lambda ||L||_{*}
\end{equation}
where $\rho_{u}$ is the standard check function in the quantile regression literature, $\lambda$ is the positive penalty coefficient, $||L||_{*}$ is the nuclear norm of the $N\times T$ matrix $L$, and $\mathcal{B}$ and $\mathcal{L}$ are convex parameter spaces about which we will be specific later. 

The key component of the estimator is the convex nuclear norm penalty. Summing the singular values of a matrix, the nuclear norm is to the rank, counting the nonzero singular values, what the convex $\ell_{1}$-norm is to the nonconvex $\ell_{0}$-norm of the vector of the singular values. Hence, the nuclear norm penalty can be viewed as the matrix counterpart of the LASSO penalty in regression with high-dimensional regressors. Being a convex surrogate of the rank functional, we show that this penalty is effective to deliver consistent estimates under low-rankness of $L_{0}$.  

The main benefits of setting up the minimization problem as in \eqref{eq.estimator.intro} are that the objective function is convex in $(\beta,L)$ and that one does not need to know $r$ before implementation. To highlight these benefits, let us consider a natural alternative estimator
\begin{equation}\label{eq.estimator.iter}
\min_{\beta\in \mathcal{B},\Lambda\in\Xi,F\in\mathcal{F}} \frac{1}{NT}\sum_{i,t}\rho_{u}(Y_{it}-X_{it}'\beta-\Lambda_{i}'F_{t})
\end{equation}
\cite{ando2019quantile} study a similar estimator where the coefficients are $i$-specific. The objective function in \eqref{eq.estimator.iter} is nonconvex in the parameters $(\beta,\Lambda,F)$. Due to nonsmoothness of the check function $\rho_{u}$, there does not exist a closed form solution to the minimization problem \eqref{eq.estimator.iter}. Implementation needs to be carried out iteratively. Then, nonconvexity leads to two potential issues. First, one may obtain a local minimum which can be arbitrarily far from the global one. Second, solving the minimization problem may be computationally intensive, especially in large panel data sets. In the simulation experiments in the paper, we will see that the penalized estimator we propose indeed outperforms this alternative estimator in the computation aspect as it saves computation time by much. Meanwhile, to make the alternative estimator \eqref{eq.estimator.iter} feasible, $r$ needs to be known or pre-estimated. This step results in additional computation burden and a mis-specified $r$ may lead to inconsistent estimates. 

Although as we show in this paper, the nuclear norm penalty slows down the rate of convergence of the coefficient estimator when the number of regressors is small, we find from the simulation experiments that its finite sample bias is actually comparable with the alternative estimator \eqref{eq.estimator.iter} even though we use the true $r$ for the latter. Moreover, we propose a consistent estimator of $r$ based on our penalized estimator. With this rank estimator, treating our consistent penalized estimator as an initial value for the iterative estimator \eqref{eq.estimator.iter} may potentially avoid the local-minimum problem, reduce computation time, and remove the bias in the penalized estimator. Finally, although the coefficient estimator may have a slower rate of convergence, we show that the error bound on the estimator of $L_{0}$ can be nearly optimal in squared Frobenius norm, not affected by penalization.

With the dense latent component and the nonsmooth objective function involved, deriving the estimator's uniform error bound is challenging. We prove new results on random matrices for this purpose. Moreover, we develop novel theoretical arguments which relax some usually made assumptions or replace some high-level technical conditions in the panel data quantile regression literature with primitive ones that are easier to interpret. These results will be introduced later in related sections and may be of independent interest. 

This paper adds to the literature of panel data quantile regression. Since \cite{koenker2004quantile}, panel data quantile regression began to draw increasing attention. \cite{abrevaya2008effects}, \cite{lamarche2010robust}, \cite{canay2011simple}, \cite{kato2012asymptotics}, \cite{galvao2013estimation} and \cite{galvao2016smoothed} study quantile regression with one-way or two-way fixed effects. \cite{harding2014estimating} consider interactive fixed effects with endogenous regressors. They require the factors to be pre-estimated or known. \cite{chen2019fixed} considers quantile regression with interactive fixed effects. They need to first estimate the time fixed effects, or, the factors, that are assumed to be quantile-level-independent. They then estimate the coefficients and the individual fixed effects via smoothed quantile regression. \cite{chen2019quantile} propose a quantile factor model without regressors. They estimate the factors and the factor loadings via nonconvex minimization similar to \eqref{eq.estimator.iter}. Pre-estimation of the number of factors is needed. \cite{ando2019quantile} consider quantile regression with heterogeneous coefficients and a factor structure. They propose both a frequentist and a Bayesian estimation procedure. The number of factors also needs to be estimated first, and the minimization problem is noncovex. Both \cite{chen2019quantile} and \cite{ando2019quantile} establish consistency pointwise in quantile level, while we focus on uniform consistency. On the technical side, both impose stronger assumptions on the conditonal density of the outcome variable than our paper\footnote{We will discuss these differences in detail in Appendix \ref{app.technical}.}. In our simulation study in Section \ref{3.sec5}, we compare our estimator with \cite{ando2019quantile} and find the estimates are comparable while our procedure is computationally more  efficient. 

Another literature this paper speaks to is on nuclear norm penalized estimation. This literature was initially motivated by low-rank matrix completion or recovery problems in computer science and statistics (e.g. \cite{candes2009exact}, \cite{ganesh2010dense}, \cite{zhou2010stable}, \cite{candes2011robust}, \cite{hsu2011robust}, \cite{negahban2011estimation}, \cite{agarwal2012noisy} and \cite{negahban2009unified} among others). In this literature, the outcome matrix is usually modeled as the sum of a low-rank matrix and some other matrices that are for instance, sparse or Gaussian. The primary goal is to estimate the low-rank or the sparse matrix. This setup is different from our paper. Nuclear norm penalized estimation and matrix completion related topics have also gained interest in econometrics recently. \cite{athey2018matrix}, \cite{moon2019nuclear}\footnote{\cite{moon2019nuclear} also briefly discuss nuclear norm penalized quantile regression with a single regressor as an extension. Using a different approach than this paper, they focus on pointwise (in quantile level) convergence rate of the coefficient estimator. In this paper, we obtain uniform rates for both the coefficients and the low-rank component. Also, the number of covariates can be more than one and growing to infinity slowly.} and \cite{chernozhukov2018inference} investigate nuclear norm penalized mean regression with interactive fixed effects. \cite{beyhum2019square} also consider mean regression with interactive fixed effects but they use a square-root nuclear norm penalty. \cite{bai2019robust} propose a nuclear norm regularized median regression for robust principal component anaylsis for fat tailed data. \cite{bai2019matrix} consider imputation of missing data and counterfactuals. \cite{bai2019rank} study penalized estimation for approximate factor models with singular values thresholding. \cite{chao2015factorisable} consider penalized multi-task quantile regression where there are multiple outcome variables and the coefficient matrix is low-rank. \cite{ma2020detecting} apply nuclear norm penalized logistic regression to the study of an undirected network formation model.

A recent paper by \cite{belloni2019high} studies quantile regression with both interactive fixed effects and high-dimensional regressors. This work was done in parallel and our paper is independent of it. In that paper, they use a nuclear norm constraint for the low-rank matrix and an additional $\ell_{1}$-norm constraint on the coefficients to deal with the high-dimensional regressors. In contrast, we focus on low-dimensional regressors, although we do allow the number of regressors to slowly grow to infinity. On the other hand, unlike our paper that derives a uniform error bound, they focus on convergence rate pointwise in quantile level. As aforementioned, achieving uniformity is challenging and requires the new results on random matrices developed in this paper. Moreover, some of our assumptions are weaker or more primitive. We defer a detailed discussion on this aspect until Appendix \ref{app.technical}. On the computation side, our algorithm differs from theirs and works fast in our simulation experiments. We view these two papers as complementary. 

The rest of the paper is organized as follows. Section \ref{3.sec2} introduces the model and the estimator. Section \ref{sec.preview} previews the main results and provides a proof sketch to highlight the challenges and the key theoretical arguments we develop. Section \ref{3.sec4} discusses the main results and their consequences. Section \ref{3.sec5} demonstrates a Monte Carlo simulation study by comparing our estimator with two alternative approaches. Section \ref{3.sec6} concludes. The algorithm and implementation details are in Appendix \ref{app.algorithm}. An alternative approach to proving consistency and a comparison of the assumptions in this paper with the most related literature are in Appendix \ref{app.lowerbound}. Appendix \ref{app.proof} collects all the proofs. Appendix \ref{app.simulation} presents some additional simulation results.

\subsection*{Notation}
Besides the nuclear norm $||\cdot||_{*}$, four additional matrix norms are used in the paper: Let $||\cdot||$, $||\cdot||_{F}$, $||\cdot||_{1}$, and $||\cdot||_{\infty}$ denote the spectral norm, the Frobenius norm, the $\ell_{1}$-norm and the maximum norm. When applied to a vector, the Frobenius norm is equal to the Euclidean norm. For two generic $N\times T$ matrices $A$ and $B$, $\left\langle A, B\right\rangle\coloneqq \sum_{i,t}A_{it}B_{it}$ denotes the inner product of $A$ and $B$. For two generic real numbers, $a\lor b$ and $a\land b$ return the maximum and the minimum of $a$ and $b$, respectively.

\section{The Model and the Estimator}\label{3.sec2}
We consider a panel data set $\{(Y_{it},X_{it}):i=1,...,N;t=1,...,T\}$ where $Y_{it}$ is a scalar outcome and $X_{it}$ is a $p\times 1$ vector of covariates. Let $Y=(Y_{it})_{i,t}$ and $X_{j}=(X_{j,it})_{i,t}$ ($j=1,...,p$) be $N\times T$ matrices of the outcome and the $j$-th covariate. Let $\mathcal{U}$ be a compact subset of $(0,1)$. For any $u\in\mathcal{U}$, there are $\bar{r}$ possibly $u$-dependent time- and individual-fixed effects. For $k=1,...,\bar{r}$, let $F_{k}(u)=(F_{1k}(u),...,F_{Tk}(u))'$ be the $k$-th time fixed effect. Let $\Lambda_{k}(u)=(\Lambda_{1k}(u),...,\Lambda_{Nk}(u))'$ be the $k$-th individual fixed effects.  Let $W=(X_{1},...,X_{p},\{F_{k}(u)\}_{k=1,...,\bar{r},u\in\mathcal{U}},\{\Lambda_{k}(u)\}_{k=1,...,\bar{r},u\in\mathcal{U}})$. Assume for all $u\in\mathcal{U}$, the conditional quantile of outcome $Y_{it}$ in matrix notation satisfies the following model with probability one:
\begin{align}
q_{Y|W}(u)=&\sum_{j=1}^{p}X_{j}\beta_{0,j}(u)+\sum_{k=1}^{\bar{r}}\mathbbm{1}_{k}(u)\Lambda_{k}(u)F_{k}(u)'\label{3.eq1}\\
\eqqcolon &\sum_{j=1}^{p}X_{j}\beta_{0,j}(u)+L_{0}(u)\label{3.eq2}
\end{align}
where $\mathbbm{1}_{k}(u)\in\{0,1\}$ determines whether the $k$-th fixed effect $F_{k}(u)$ or $\Lambda_{k}(u)$ affects the $u$-th conditional quantile of $Y$ at all. Model \eqref{3.eq1} allows both the set of the effective fixed effects and the realizations of them to depend on $u$. Throughout, we allow the fixed effects to be either random or deterministic. The covariates can be correlated with them when they are random. Similar setups of the fixed effects or factor structures in panel data quantile regression can be found in \cite{ando2019quantile} and \cite{chen2019quantile}. 

The fixed effects in equation \eqref{3.eq1} form an $N\times T$ matrix $L_{0}(u)\coloneqq \sum_{k=1}^{\bar{r}}\mathbbm{1}_{k}(u)\Lambda_{k}(u)F_{k}(u)'$. The rank of the matrix $L_{0}(u)$ is at most $r(u)\coloneqq\sum_{k=1}^{\bar{r}}\mathbbm{1}_{k}(u)\leq \bar{r}$ by construction. In this paper, we assume $\bar{r}$ is fixed. Thus, $L_{0}(u)$ is low-rank when $N$ and $T$ are large. 

Let $\beta_{0}(u)=(\beta_{0,j}(u))_{j=1,...,p}$. This paper focuses on consistently estimating $(\beta_{0}(u),L_{0}(u))$ uniformly in $u\in\mathcal{U}$. When $L_{0}(u)$ is random, consistency is in terms of its realization.

\begin{rem}
When $\mathcal{U}$ is a singleton containing $u$, the conditioning variables $W$ only contains the covariates and the fixed effects at $u$. The model is then in line with the models in the literature on panel data quantile regression that focuses on a fixed $u\in (0,1)$, for example \cite{harding2014estimating}, \cite{ando2019quantile} and \cite{chen2019quantile}. 
\end{rem}

Now let us present a few models which admit the conditional quantile function \eqref{3.eq1}.
\begin{example}[A location shift model with one-way fixed effects only]\label{eg1}
Suppose the outcome matrix $Y$ is determined by the following linear model with only individual fixed effects $\Lambda^{o}=(\Lambda_{1}^{o},...,\Lambda_{N}^{o})'$ (similarly, one can also consider a model with time fixed effects only):
\[
Y=\beta^{o}\cdot\bm{1}_{N\times T}+\sum_{j=1}^{p}X_{j}\beta_{0,j}+\Lambda^{o}\cdot\bm{1}_{1\times T}+\epsilon.
\]
where $\epsilon$ is an $N\times T$ error matrix, $\bm{1}_{N\times T}$ and $\bm{1}_{1\times T}$ are $N\times T$ and $1\times T$ matrices of ones. Assume $(\{X_{j}\}_{j},\Lambda^{o})\independent \epsilon$ while $\{X_{j}\}_{j}$ and $\Lambda^{o}$ can be arbitrarily correlated. Assume the $\epsilon_{it}$s are identically distributed on $\mathbb{R}$ and let $q_{\epsilon}(\cdot)$ denote their quantile function. Then the $u$-th conditional quantile of $Y$ is $q_{Y|W}(u)=q_{Y|\{X_{j}\}_{j},\Lambda(u)}(u)=\sum_{j=1}^{p}X_{j}\beta_{0,j}+\Lambda(u)F'$ with probability one for all $u\in (0,1)$ where $\Lambda(u)=\Lambda^{o}+(\beta^{o}+q_{\epsilon}(u))\cdot\bm{1}_{N\times 1}$ and $F=\bm{1}_{T\times 1}$. The fixed effects form an $N\times T$ rank one matrix for all $u\in (0,1)$ with identical columns. 
\end{example}

\begin{example}[A location-scale model with interactive fixed effects]\label{eg2}
Suppose 
\[Y=\sum_{j=1}^{p}X_{j}\beta^{a}_{0,j}+\sum_{k=1}^{\bar{r}_{1}}\Lambda_{k}^{a}F^{a'}_{k}+\left(\sum_{j=1}^{p}X_{j}\beta^{b}_{0,j}+\sum_{m=1}^{\bar{r}_{2}}\Lambda_{m}^{b}F^{b'}_{m}\right)\circ\epsilon.
\]
where $\circ$ denotes the Hadamard product of matrices. Assume that $(\{X_{j}\}_{j},\{F^{a}_{k},\Lambda_{k}^{a}\}_{k},\{F^{b}_{m},\Lambda_{m}^{b}\}_{m})\independent \epsilon$ and let $q_{\epsilon}(\cdot)$ be the quantile function of the identically distributed $\epsilon_{it}$s. If for all $i,t$ and $m$, and all $x,f_{m}^{b}$ and $\lambda^{b}_{m}$ in the support sets of $X_{it}$, $F_{tm}^{b}$ and $\Lambda_{im}^{b}$, the inequality $x'\beta^{b}_{0}+\sum_{m=1}^{\bar{r}_{2}}f_{m}^{b}\lambda_{m}^{b}>0$ holds, then for any $u\in (0,1)$, by letting $\beta_{0}(u)=\beta^{a}_{0}+\beta^{b}_{0}q_{\epsilon}(u)$, $\Lambda(u)=(\Lambda^{a}_{1},...,\Lambda^{a}_{\bar{r}_{1}},\Lambda^{b}_{1} q_{\epsilon}(u),...,\Lambda^{b}_{\bar{r}_{2}} q_{\epsilon}(u))$ and $F=(F^{a}_{1},...,F^{a}_{\bar{r}_{1}},F^{b}_{1},...,F^{b}_{\bar{r}_{2}})$, the $u$-th conditional quantile of $Y$ is
 $q_{Y|W}(u)=q_{Y|\{X_{j}\}_{j},\{F_{l},\Lambda_{l}(u)\}_{l}}(u)=\sum_{j=1}^{p}X_{j}\beta_{0,j}(u)+\sum_{l=1}^{\bar{r}_{1}+\bar{r}_{2}}\Lambda_{l}(u)F_{l}'$ with probability one. In this example, the coefficients and the individual fixed effects are $u$-dependent. 
\end{example}
\begin{example}[A random coefficient model with quantile-dependent fixed effects]\label{eg3}
Let $U_{it}\sim \text{Unif}[0,1]$ and
\[
Y_{it}=X_{it}'\beta_{0}(U_{it})+\sum_{k=1}^{\bar{r}-1}\mathbbm{1}_{k}(U_{it})F_{tk}^{o}(U_{it})\Lambda_{ik}^{o}(U_{it})+\epsilon_{it}.
\]
Assume that for every $k=1,...,\bar{r}-1$, the indicator function $\mathbbm{1}_{k}(u)$ is non-decreasing in $u$. For instance, $\mathbbm{1}_{1}(u)=1$, $\mathbbm{1}_{2}(u)=\mathbbm{1}(u>0.3)$, etc. Moreover, suppose for all $i,t$ and $k$ and for all $x,f_{k}(u)$ and $\lambda_{k}(u)$ in the support of $X_{it}$, $F_{tk}^{o}(u)$ and $\Lambda_{ik}^{o}(u)$, functions $x'\beta(u)$ and $f_{k}(u)\lambda_{k}(u)$ are all strictly increasing in $u$. Suppose  $\epsilon_{it}$ is generated by some strictly increasing function of $U_{it}$, $G^{-1}(U_{it})$. Finally, assume that $\{U_{it}\}_{i,t}$ are independent of $(\{X_{it},\{F_{tk}^{o}(u)\}_{k,u},\{\Lambda_{ik}^{o}(u)\}_{k,u}\}_{i,t})$. Then by strict monotonicity and independence, the $u$-th conditional quantile of $Y$ is $q_{Y|W}(u)=\sum_{j=1}^{p}X_{j}\beta_{0,j}(u)+\sum_{k=1}^{\bar{r}}\mathbbm{1}_{k}(u)\Lambda_{k}(u)F_{k}(u)'$ with probability one for all $u\in (0,1)$ where for $k<\bar{r}$, $F_{k}(u)=F^{o}_{k}(u)$ and $\Lambda_{k}(u)=\Lambda^{o}_{k}(u)$. For $k=\bar{r}$, $\Lambda_{\bar{r}}(u)=G^{-1}(u)\bm{1}_{N\times 1}$ while $\mathbbm{1}_{\bar{r}}(u)=1$ and $F_{\bar{r}}(u)=\bm{1}_{T\times 1}$ for all $u\in (0,1)$. In this example, the coefficients, fixed effects, and the set of the effective fixed effects all depend on $u$. We will revisit this example in our Monte Carlo experiment in Section \ref{3.sec5}.
\end{example}
Now we introduce our estimator of $(\beta_{0}(u),L_{0}(u))$. For a generic $N\times T$ matrix $Z$, define $\bm{\rho}_{u}(Z)\coloneqq \sum_{i,t}\rho_{u}(Z_{it})\equiv \sum_{i,t} Z_{it}(u-\mathbbm{1}(Z_{it}\leq 0))$. By exploiting the linearity of the conditional quantile function \eqref{3.eq2} in $(\beta_{0}(u),L_{0}(u))$ and the low-rankness of $L_{0}(u)$, this paper proposes the following \textit{nuclear norm penalized quantile regression estimator} to jointly estimate $\beta_{0}(u)$ and $L_{0}(u)$ for any $u\in\mathcal{U}$: 
\begin{equation}\label{3.eq3}
(\hat{\beta}(u),\hat{L}(u))\coloneqq\arg\min_{\beta\in\mathbb{R}^{p},L\in\mathcal{L}} \ \ \frac{1}{NT}\bm{\rho}_{u} \left(Y-\sum_{j=1}^{p} X_{j} \beta_{j}-L\right)+\lambda ||L||_{*}
\end{equation}
where $\lambda>0$. The parameter space $\mathcal{L}\coloneqq\{L\in\mathbb{R}^{N\times T}:||L||_{\infty}\leq \alpha_{NT}\}$ is convex and compact and $\alpha_{NT}\geq 1$ can be $(N,T)$-dependent. In particular, we allow $\alpha_{NT}$ to \textit{grow to infinity} with $N$ and $T$. We need $\alpha_{NT}$ for technical reasons to be discussed in Section \ref{secStrategy}. In Appendix \ref{app.second}, we show that we can drop $\alpha_{NT}$ to make $\mathcal{L}=\mathbb{R}^{N\times T}$ under a different set of assumptions. 

Two remarks on the estimator are in order .
\begin{rem}\label{nuclear norm}
The estimator does not directly penalize or constrain the rank of the estimated interactive fixed effect matrix to avoid nonconvexity. Instead, it seeks an $\hat{L}(u)$ that has a small nuclear norm. Intuitively, this is reasonable because the rank-$r(u)$ matrix $L_{0}(u)$ itself typically has a small nuclear norm by low-rankness\footnote{Suppose all the nonzero singular values of $L_{0}(u)$ are of the order of $\sqrt{NT}$, a valid assumption if elements in $L_{0}(u)$ are $O(1)$, then the order of $||L_{0}(u)||_{*}$ is only $\sqrt{NT}r(u)$, much smaller than that of a full-rank matrix, which can be as large as $\sqrt{NT}(N\land T)$}. To achieve a small nuclear norm, the penalty would shrink some of $\hat{L}(u)$'s singular values even though the rank of $\hat{L}(u)$ may still remain high since rank may increase dramatically even by a very small perturbation to $L_{0}(u)$. For instance, the $(r(u)+1)$-th to the $(N\land T)$-th singular values in $\hat{L}(u)$ can be nonzero, but they may have a smaller order than the first $r(u)$ singular values. This is shown in Section \ref{3.sec4} and is helpful to develop the estimator of $r(u)$ we propose.
\end{rem}
\begin{rem}\label{lambda}
The penalty coefficient $\lambda$ balances how small $||\hat{L}(u)||_{*}$ is and how well the estimator fits the data. When $\lambda=0$, the trivial solution to \eqref{3.eq3} is $(\hat{\beta}(u),\hat{L}(u))=(0,Y)$, provided that $Y\in\mathcal{L}$. This estimator fits the data perfectly but is inconsistent as long as the true coefficients $\beta_{0}(u)\neq 0$. On the other hand, when $\lambda$ is infinity, $\hat{L}(u)$ would be $0$ to set the nuclear norm penalty equal to 0. This could again be implausible because it is equivalent to ignoring $L_{0}(u)$ and estimating $\beta_{0}(u)$ simply by pooled quantile regression. As a result, the estimator $\hat{\beta}(u)$ would be inconsistent when the covariates are correlated with the fixed effects. In the next section, we will be precise about the appropriate order of $\lambda$ that guarantees uniform consistency of the estimator. 
\end{rem}

\section{Preview of the Main Results and Proof Strategy}\label{sec.preview}
In this section, we first briefly summarize the main theoretical results of this paper, then outline the proof strategy to highlight the challenges arising from the nonsmooth objective function and the dense low-rank common component $L_{0}(u)$. We also introduce and discuss the assumptions we make motivated by these challenges.

Under the assumptions to be introduced in this section, this paper shows that for some universal constant $C_{error}>0$, the estimator $(\hat{\beta}(u),\hat{L}(u))$ defined in equation \eqref{3.eq3} satisfies the following inequality with probability approaching one (w.p.a.1):
\begin{equation}\label{preview}
\sup_{u\in\mathcal{U}}\left(||\hat{\beta}(u)-\beta_{0}(u)||_{F}^{2}+\frac{1}{NT}||\hat{L}(u)-L_{0}(u)||_{F}^{2}\right)\leq \gamma^{2}
\end{equation}
where 
\begin{equation}\label{eq.t}
\gamma=C_{error}\alpha_{NT}^{2}\sqrt{\log(NT)}\left(\sqrt{\frac{p\log(pNT)}{NT}}\lor \sqrt{\frac{\bar{r}}{N\land T}}\right).
\end{equation}
The error bound implies uniform consistency of the estimator given a fixed $\bar{r}$ and a fixed or slowly growing $p$ and $\alpha_{NT}$. Based on the order of the error bound, we also propose a consistent estimator of the number of the effective fixed effects $r(u)$ for each $u\in\mathcal{U}$. We will discuss these result with more details in Section \ref{3.sec4}, but now let us first sketch how the error bound is derived.

To derive the error bound, it is helpful to first exploit some simple implications from the definition of the estimator to sharpen the space where the estimation errors $(\hat{\Delta}_{\beta}(u),\hat{\Delta}_{L}(u))\coloneqq (\hat{\beta}(u)-\beta_{0}(u),\hat{L}(u)-L_{0}(u))$ lie so that the analysis can be conducted in this smaller space instead of $\mathbb{R}^{p}\times \mathbb{R}^{N\times T}$. For this purpose, we make the following assumption.
\begin{assumption}\label{3.ass1}
i) Let $V(u)\coloneqq Y-\sum_{j=1}^{p}X_{j}\beta_{0,j}(u)-L_{0}(u)$. For all $u\in\mathcal{U}$, elements in matrix $V(u)$ are independent conditional on $W$. ii) There exists a universal constant $C_{X}>0$ such that $\max_{1\leq j\leq p}||X_{j}||_{F}^{2}\leq C_{X}NT$ w.p.a.1.
\end{assumption}
The independence requirement in part i) is for simplicity so that some inequalities for random matrices can be easily applied in the proof. The same assumption when $\mathcal{U}$ is a singleton so that $W$ only contains the covariates and the fixed effects at $u$ can be found in \cite{ando2019quantile} and \cite{chen2019quantile} as well. Moderate serial correlation in $V_{it}(u)$ can be allowed at a cost of more technical conditions. On the other hand, serial correlation in the covariates is indeed allowed, and $p$ is allowed to be growing in $N$ and $T$; part ii) in Assumption \ref{3.ass1} holds as long as for instance, for every $j\leq p$, Chebyshev's inequality holds for $\sum_{i,t}X_{j,it}^{2}/NT$ and its variance multiplied by $p$ is $o(1)$. 

It turns out that under Assumption \ref{3.ass1} alone, the estimation error $(\hat{\Delta}_{\beta}(u),\hat{\Delta}_{L}(u))$ lies in a cone uniformly in $u\in\mathcal{U}$ w.p.a.1. The cone has nice properties for deriving the error bound. To characterize the cone, let us introduce some notation. Let $R(u)\Sigma(u) S(u)'$ be a singular value decomposition of $L_{0}(u)$. Following \cite{candes2009exact}, let $\Phi(u)\coloneqq\{M\in \mathbb{R}^{N\times T}:\exists A\in\mathbb{R}^{r(u)\times T}\ \text{and}\ B\in\mathbb{R}^{N\times r(u)}\ s.t.\ M=R(u)A+BS(u)'\}$. Denote the orthogonal projection of a generic $N\times T$ matrix $W$ onto this space by $\mathcal{P}_{\Phi(u)}W$, then
\begin{equation}\label{projection}
\mathcal{P}_{\Phi(u)}W=R(u)R(u)'W+WS(u)S(u)'-R(u)R(u)'WS(u)S(u)'.
\end{equation}
We then have the following lemma.

\begin{lem}\label{3.lem1}
Under Assumption \ref{3.ass1}, if $L_{0}(u)\in\mathcal{L}$ for all $u\in\mathcal{U}$, then for $C_{X}$ defined in Assumption \ref{3.ass1}, there exists universal constants $C_{Cone}>0$ such that for $\lambda=2\sqrt{2C_{X}(N\lor T)}/(C_{Cone}NT)$, we have
\begin{equation}\label{coneineq}
\sup_{u\in\mathcal{U}}\left(||\hat{\Delta}_{L}(u)||_{*}-4||\mathcal{P}_{\Phi(u)}\hat{\Delta}_{L}(u)||_{*}-C_{Cone}\sqrt{p(N\land T)\log(pNT)}||\hat{\Delta}_{\beta}(u)||_{F}\right)\leq 0,w.p.a.1.
\end{equation}

\end{lem} 
\begin{proof}
See Appendix \ref{app.sec3}.
\end{proof}
By Lemma \ref{3.lem1},  define cone $\mathcal{R}_{u}$ by \small
\begin{equation*}\label{cone}
\mathcal{R}_{u}\coloneqq \left\lbrace(\Delta_{\beta},\Delta_{L})\in\mathbb{R}^{p}\times\mathbb{R}^{N\times T}: ||\Delta_{L}||_{*}\leq 4||\mathcal{P}_{\Phi(u)}\Delta_{L}||_{*}+C_{Cone}\sqrt{p(N\land T)\log(pNT)}||\Delta_{\beta}||_{F}\right\rbrace.
\end{equation*}\normalsize
We thus have $(\hat{\Delta}_{\beta}(u),\hat{\Delta}_{L}(u))\in\mathcal{R}_{u}$ uniformly in $u\in\mathcal{U}$ w.p.a.1 under Assumption \ref{3.ass1} and under $||L_{0}(u)||_{\infty}\leq \alpha_{NT}$ for all $u\in\mathcal{U}$.

The key property of $\mathcal{R}_{u}$ is that for any element $(\Delta_{\beta},\Delta_{L})\in\mathcal{R}_{u}$ and for any $u\in\mathcal{U}$, $||\Delta_{L}||_{*}$ and $||\Delta_{L}||_{F}$ can be of the same order, a property that low-rank matrices also share\footnote{Note that this property is non-trivial because in general the nuclear norm $||\Delta_{L}||_{*}$ can be as large as $\sqrt{N\land T}||\Delta_{L}||_{F}$. }. To see why this is true, note that by equation \eqref{projection}, the rank of $\mathcal{P}_{\Phi(u)}\Delta_{L}$ is at most $3r(u)$ for all $u$. Hence, for all $u\in\mathcal{U}$,
\[||\mathcal{P}_{\Phi(u)}\Delta_{L}||_{*}\leq \sqrt{3r(u)}||\mathcal{P}_{\Phi(u)}\Delta_{L}||_{F}\leq \sqrt{3r(u)} ||\Delta_{L}||_{F}\leq \sqrt{3r(u)}||\Delta_{L}||_{*},
\]
where the first and the last inequalities are by the relationship between the nuclear norm and the Frobenius norm. The second inequality is due to $\left\langle\mathcal{P}_{\Phi(u)}\Delta_{L},\Delta_{L}-\mathcal{P}_{\Phi(u)}\Delta_{L}\right\rangle=0$ and by the Pythagoras formula. As a consequence, elements in $\mathcal{R}_{u}$ satisfy
\begin{equation}\label{equiorder}
\frac{1}{4\sqrt{3\bar{r}}}\left(||\Delta_{L}||_{*}-C_{Cone}\sqrt{p(N\land T)\log(pNT)}||\Delta_{\beta}||_{F}\right)\leq ||\Delta_{L}||_{F}\leq  ||\Delta_{L}||_{*}
\end{equation}
which implies that if $C_{Cone}\sqrt{p(N\land T)\log(pNT)}||\Delta_{\beta}||_{F}/||\Delta_{L}||_{*}=o(1)$, $||\Delta_{L}||_{*}$ and $||\Delta_{L}||_{F}$ are of the same order. Note that since the estimation error $(\hat{\Delta}_{\beta}(u),\hat{\Delta}_{L}(u))$ is in $\mathcal{R}_{u}$ w.p.a.1, $||\hat{\Delta}_{L}(u)||_{F}$ and $||\hat{\Delta}_{L}(u)||_{*}$ are also of the same order w.p.a.1. This implies that, although it is unclear whether the nuclear norm penalty makes some of the singular values of $\hat{\Delta}_{L}(u)$ to be exact zero so that $\hat{\Delta}_{L}(u)$, and in turn, $\hat{L}(u)$, is low-rank, at least the singular values of $\hat{\Delta}_{L}(u)$ must not be of the same order w.p.a.1.  

In the sense that the nuclear norm and the Frobenius norm of the matrix elements in $\mathcal{R}_{u}$ can be of the same order, the cone $\mathcal{R}_{u}$ in Lemma \ref{3.lem1} matches those obtained in the broad literature of nuclear norm penalized estimation under different objective functions (see e.g. \cite{agarwal2012noisy}, \cite{negahban2012restricted}, \cite{athey2018matrix}, \cite{chernozhukov2018inference}). Different from the mentioned literature, to establish uniformity in Lemma \ref{3.lem1} under the nonsmooth objective function, we derive new uniform bounds on some norms of random matrices whose entries are jump processes (see Lemma \ref{3.lemA1} in Appendix \ref{app.sec3} for details). These results may be of independent interest.

Under Lemma \ref{3.lem1}, we can conduct all the subsequent analysis conditional on the event that $(\hat{\Delta}_{\beta}(u),\hat{\Delta}_{L}(u))\in\mathcal{R}_{u}$ for all $u\in\mathcal{U}$ to exploit property \eqref{equiorder} of $\mathcal{R}_{u}$. Now let us sketch the derivation of the error bound \eqref{preview} to motivate the other assumptions we make, introduce the theoretical challenges, and discuss how we overcome them.

\subsection{Proof Strategy}\label{secStrategy}

 Let $\mathcal{D}\coloneqq \mathbb{R}^{p}\times \{\Delta_{L}\in\mathbb{R}^{N\times T}:||\Delta_{L}||_{\infty}\leq 2\alpha_{NT} \}$. Recall that $\mathcal{L}=\{L\in\mathbb{R}^{N\times T}:||L||_{\infty}\leq\alpha_{NT}\}$. If $||L_{0}(u)||_{\infty}\leq \alpha_{NT}$ for all $u\in\mathcal{U}$ w.p.a.1, then by $\hat{L}(u)\in\mathcal{L}$ and by Lemma \ref{3.lem1}, we have $(\hat{\Delta}_{\beta}(u),\hat{\Delta}_{L}(u))\in \mathcal{R}_{u}\cap\mathcal{D}$ for all $u\in\mathcal{U}$ w.p.a.1. For a generic sample $(Z_{it})_{i,t}$ and a function $f$, let $\mathbb{G}_{u}(f(Z_{it}))\coloneqq \sum_{i,t}[f(Z_{it})-\mathbb{E}(f(Z_{it})|W)]/\sqrt{NT}$ where recall $W\coloneqq (X_{1},...,X_{p},\{F_{k}(u)\}_{k=1,...,\bar{r},u\in\mathcal{U}},\{\Lambda_{k}(u)\}_{k=1,...,\bar{r},u\in\mathcal{U}})$. Since $\mathcal{R}_{u}$ is a cone and $\mathcal{D}$ is convex and contains zero, by convexity of the objective function in equation \eqref{3.eq3}, it is sufficient to show that it is zero probability to have a $u\in\mathcal{U}$ such that:

\begin{align}
0> \inf_{\substack{(\Delta_{\beta},\Delta_{L})\in\mathcal{R}_{u}\cap\mathcal{D}\\||\Delta_{\beta}||_{F}^{2}+\frac{1}{NT}||\Delta_{L}||_{F}^{2}=\gamma^{2}}}&\frac{1}{NT}\left[\bm{\rho}_{u} \big(V(u)-\sum_{j=1}^{p} X_{j} \Delta_{\beta,j}-\Delta_{L}\big)-\bm{\rho}_{u} \left(V(u)\right)\right]\notag\\
& +\lambda \left[||L_{0}(u)+\Delta_{L}||_{*}-||L_{0}(u)||_{*}\right]\notag\\
=\inf_{\substack{(\Delta_{\beta},\Delta_{L})\in\mathcal{R}_{u}\cap\mathcal{D}\\||\Delta_{\beta}||_{F}^{2}+\frac{1}{NT}||\Delta_{L}||_{F}^{2}=\gamma^{2}}}& \frac{1}{NT} \mathbb{E}\left[\bm{\rho}_{u} \big(V(u)-\sum_{j=1}^{p} X_{j} \Delta_{\beta,j}-\Delta_{L}\big)-\bm{\rho}_{u} (V(u))\Big|W(u)\right]\notag\\
&+\frac{1}{\sqrt{NT}}\mathbb{G}_{u}\left(\rho_{u} \big(V_{it}(u)-X_{it}' \Delta_{\beta}-\Delta_{L,it}\big)-\rho_{u} (V_{it}(u))\right)\notag\\
&+\lambda \left[||L_{0}(u)+\Delta_{L}||_{*}-||L_{0}(u)||_{*}\right]\label{eq.base}
\end{align}
where $(\Delta_{\beta},\Delta_{L})$ is a generic element in $\mathcal{R}_{u}\cap\mathcal{D}$ and recall that $V(u)\coloneqq Y-\sum_{j=1}^{p}X_{j}\beta_{0,j}(u)-L_{0}(u)$. 

We prove inequality \eqref{eq.base} can not hold for any $u\in\mathcal{U}$ in two steps:
\begin{enumerate}[leftmargin=1in, label=Step \arabic*]
\item\label{step1}: Find a positive quadratic lower bound on the conditional expectation.
\item\label{step2}: Find an upper bound on the absolute process $|\mathbb{G}_{u}|$ and the absolute penalty difference. 
\end{enumerate}
After we find these bounds, we conclude that inequality \eqref{eq.base} is impossible uniformly in $u\in \mathcal{U}$ if the lower bound in \ref{step1} is always greater in than the upper bound in \ref{step2}.
\subsubsection{Step 1}

In \ref{step1}, we are to lower bound the conditional expectation by $(||\Delta_{\beta}||_{F}^{2}+||\Delta_{L}||_{F}^{2}/NT)$ up to some multiplicative factor. The major theoretical difficulty arises from high-dimensionality of $\Delta_{L}$. Note that \ref{step1} is independent of the nuclear norm penalty and a similar difficulty exists even if we maintain the factor structure (e.g. \cite{ando2019quantile} and \cite{chen2019quantile}) instead of treating the interactive fixed effects as a single low-rank matrix. Therefore, the theoretical arguments developed in the paper to overcome the difficulty can be applied to other estimators for panel data quantile regression with interactive fixed effects or factor structures. 

To illustrate the difficulty, consider a simple case without covariates. The conditional expectation under consideration can then be simplified as $\mathbb{E}\left[\bm{\rho}_{u} \left(V(u)-\Delta_{L}\right)-\bm{\rho}_{u}(V(u))|W\right]$. By Knight's identity \citep{knight1998limiting} and by the definition of $V_{it}(u)$, it can be rewritten as 
\begin{equation}\label{eq.integral}
\sum_{i,t}\int_{0}^{\Delta_{L,it}}\left(F_{V_{it}(u)|W}(s)-F_{V_{it}(u)|W}(0)\right)ds
\end{equation}
where $F_{V_{it}(u)|W}$ is the conditional cumulative distribution function of $V_{it}(u)$. The key problem is as follows. The magnitude of some $\Delta_{L,it}$s can be arbitrarily large or grow to infinity with $N$ and $T$ even when $||\Delta_{L}||_{F}^{2}/NT\to 0$. Hence, if one adopts the standard argument in quantile regression to first-order Taylor expand $F_{V_{it}(u)|W}(s)$ around $0$, it is \textit{insufficient} to obtain a strictly positive lower bound on quantity \eqref{eq.integral} if one only assumes that the conditional density of $V_{it}(u)$ at zero, $f_{V_{it}(u)|W}(0)$, is continuous and positive uniformly in $i$, $t$ and $u$ almost surely. 

To overcome the difficulty, the literature on panel data quantile regression with factor structures or interactive fixed effects often assumes that elements in $L_{0}(u)$ (or the fixed effects) lie in a \textit{fixed} compact space and the conditional density $f_{V_{it}(u)|W}(s)$ is bounded away from 0 on \textit{any compact interval} of $s$ (e.g. \cite{ando2019quantile} and \cite{chen2019quantile}). Alternatively, \cite{belloni2019high} adopt higher order Taylor expansion by assuming differentiability of the conditional density with bounded derivatives. They also impose a high-level condition on the estimation error matrix 
which is not straightforward to interpret\footnote{A more detailed comparison between these approaches and ours is in Appendix \ref{app.technical}}.

In this paper, we develop two novel sets of theoretical arguments, or, approaches, to overcome the difficulty to achieve a positive quadratic lower bound. Under both approaches, $f_{V_{it}(u)|W}$ bounded away from zero in an arbitrarily small neighborhood of $0$ is sufficient. Meanwhile, we do \textit{not} need the conditional density function to be differentiable. Formally, we make the following assumption:
\begin{assumption}\label{3.ass2}
There exists a $\delta>0$ such that the conditional densities $f_{V_{it}(u)|W}$ satisfies
\begin{equation*}
\underline{f}\coloneqq \inf _{\substack{s\in[-\delta,\delta], u\in\mathcal{U}\\1\leq i\leq N,1\leq t\leq T}}f_{V_{it}(u)|W}(s)>0\ a.s.
\end{equation*}
\end{assumption}
Note that $\delta$ can be arbitrarily small. A sufficient condition for Assumption \ref{3.ass2} to hold is that $f_{V_{it}(u)|W}(0)>0$ uniformly in $i,t$ and $u\in\mathcal{U}$ and the functions $\{f_{V_{it}(u)|W}\}_{i,t,u}$ are equicontinuous at $0$ for all realizations of $W$. This assumption can be shown to be weaker than the assumptions on the conditional density in \cite{belloni2019high}, \cite{ando2019quantile} and \cite{chen2019quantile}. See Appendix \ref{app.technical} for a detailed discussion.

Under Assumption \ref{3.ass2}, our first approach requires a compact parameter space $\mathcal{L}$ for the matrix component as in equation \eqref{3.eq3}, but unlike \cite{ando2019quantile} and \cite{chen2019quantile}, the boundary $\alpha_{NT}$ is not fixed and is allowed to grow to infinity. Our second approach relaxes $\mathcal{L}$ to be $\mathbb{R}^{N\times T}$ at the cost of more technical conditions. 
 
We present the main idea behind our first approach now and discuss our second approach in Appendix \ref{app.second}. Since we are considering a case without the covariates, let us redefine $\mathcal{D}=\{\Delta_{L}\in\mathbb{R}^{N\times T}:||\Delta_{L}||_{\infty}\leq 2\alpha_{NT}\}$. We show in Appendix \ref{app.techlem} that each integral in the summation \eqref{eq.integral} is decreasing in the absolute upper limit. So each integral satisfies
\begin{align*}
\int_{0}^{\Delta_{L,it}}\left(F_{V_{it}(u)|Wu)}(s)-F_{V_{it}(u)|W}(0)\right)ds
\geq &\int_{0}^{(1\land \delta)\Delta_{L,it}/2\alpha_{NT}}\left(F_{V_{it}(u)|Wu)}(s)-F_{V_{it}(u)|W}(0)\right)ds
\end{align*}
where $\delta$ is the constant in Assumption \ref{3.ass2}. For the integral on the right side, even if $\Delta_{L,it}$ is diverging, $|\Delta_{L,it}|/2\alpha_{NT}\leq 1$ so $(1\land \delta)\Delta_{L,it}/2\alpha_{NT}$ must lie in the region where the conditional density is positive. Then first-order Taylor expanding $F_{V_{it}(u)|W}(s)$ around $0$ yields a strictly positive quadratic lower bound. The issue is thus resolved.


When there are covariates, the upper limits in the integrals in quantity \eqref{eq.integral} become $\Delta_{L,it}+X_{it}'\Delta_{\beta}$. To make this approach still valid at least in the ball $||\Delta_{\beta}||_{F}^{2}+||\Delta_{L}||_{F}^{2}/NT\leq \gamma^{2}$, we make the following assumption so that $X_{it}'\Delta_{\beta}$ is also bounded by some function of $\alpha_{NT}$ w.p.a.1. 
\begin{assumption}\label{3.ass3}
The $\alpha_{NT}$ in the parameter space $\mathcal{L}$ in equation \eqref{3.eq3} is no smaller than $1$ and can grow to infinity with $N$ and $T$. Meanwhile, i) $||L_{0}(u)||_{\infty}\leq \alpha_{NT}$ for all $u\in\mathcal{U}$ w.p.a.1 and ii) $\sum_{j=1}^{p}||X_{j}||_{\infty}=o_{p}(\alpha_{NT}/\gamma)$ for $\gamma$ defined in equation \eqref{eq.t}. 
\end{assumption}
Assumption \ref{3.ass3} is reasonably mild because we allow $\alpha_{NT}$ to be $(N,T)$-dependent and to grow to infinity. For instance, if $\alpha_{NT}$ has order $\log(NT)$, part i) is satisfied in all the models in Examples \ref{eg1} to \ref{eg3} if all the individual- and time-fixed effects are sub-Gaussian with $\bar{r}$ fixed. Note that this \textit{does not} rule out the case where $Y_{it}$ itself has a heavy tail. The restriction on the covariates' tails is even milder. For instance, if $\alpha_{NT}=\log(NT)$ and $p$ is fixed, given the order of $\gamma$ defined in equation \eqref{eq.t}, some heavy-tailed distributions are allowed for the covariates. 

With the covariates, recall that $\mathcal{D}\coloneqq \mathbb{R}^{p}\times \{\Delta_{L}\in\mathbb{R}^{N\times T}:||\Delta_{L}||_{\infty}\leq 2\alpha_{NT} \}$. We have the following lemma.

\begin{lem}\label{lowerbound}
 Under Assumptions \ref{3.ass2} and the event that $\sum_{j=1}^{p}||X_{j}||_{\infty}\gamma\leq \alpha_{NT}$ with $\alpha_{NT}\geq 1$, the following inequality holds almost surely:
\begin{align}
&\inf_{\substack{u\in\mathcal{U}\\||\Delta_{\beta}||_{F}\leq \gamma\\(\Delta_{\beta},\Delta_{L})\in\mathcal{D}}}\left(\mathbb{E}\big[\bm{\rho}_{u} \big(V(u)-\sum_{j=1}^{p} X_{j} \Delta_{\beta,j}-\Delta_{L}\big)-\bm{\rho}_{u} (V(u))\big|W\big]- \frac{C_{min}}{\alpha_{NT}^{2}}\big|\big|\sum_{j=1}^{p}X_{j}\Delta_{\beta,j}+\Delta_{L}\big|\big|_{F}^{2}\right)\notag\\
&\geq 0\label{eq.lowerbound}
\end{align}
where $C_{min}\coloneqq (1\land \delta)^{2}\underline{f}/18$, $\underline{f}$ and $\delta$ are defined in Assumption \ref{3.ass2}, and $\gamma$ is defined in equation \eqref{eq.t}.
\end{lem}
\begin{proof}
See Appendix \ref{app.sec3}.
\end{proof}
\begin{rem}
Lemma \ref{lowerbound} does not rely on Assumption \ref{3.ass1}, so independence among the $V_{it}(u)$s is not required.
\end{rem}
Finally, in order to obtain the error bounds on $\hat{\beta}(u)$ and $\hat{L}(u)$ separately, we impose the following identification assumption. 

\begin{assumption}\label{ass.id}
 i) There exists a universal constant $\sigma_{min}^{2}>0$ such that the smallest eigenvalue of $\sum_{i,t}(X_{it}X_{it}')/NT$ converges to $\sigma_{min}^{2}$ in probability. \newline
ii) There exists a universal constant $C_{RSC}>0$ such that the following holds w.p.a.1:
\begin{equation}\label{eq.rsc}
\inf_{\substack{(\Delta_{\beta},\Delta_{L})\in\mathcal{R}_{u}\\u\in\mathcal{U}}}\left(\big|\big|\sum_{j=1}^{p}X_{j}\Delta_{\beta,j}+\Delta_{L}\big|\big|_{F}^{2}- C_{RSC}\left(\big|\big|\sum_{j=1}^{p}X_{j}\Delta_{\beta,j}\big|\big|_{F}^{2}+||\Delta_{L}||_{F}^{2}\right)\right)\geq 0
\end{equation} 

\end{assumption}

Part i) in Assumption \ref{ass.id} guarantees that the individual coefficients $\beta_{0,j}(u)$s can be separately consistently estimated. Part ii) is a version of the widely adopted \textit{restricted strong convexity} condition in the literature on low-rank matrix recovery or nuclear norm penalized estimation (e.g. \cite{agarwal2012noisy}, \cite{negahban2011estimation,negahban2012restricted},  \cite{negahban2009unified}, \cite{belloni2019high}, \cite{chernozhukov2018inference}, etc.). It says that any linear combinations of the covariate matrices $X_{j}$s must lie sufficiently far from the matrix elements in the cone $\mathcal{R}_{u}$. The constant $C_{RSC}$ is determined by the joint distribution of the covariates and $L_{0}(u)$. When the covariates are more correlated with $L_{0}(u)$, this constant tends to become smaller. Consequently, the error bound on the estimation error would be larger. We can observe this pattern in the Monte Carlo simulations in Section \ref{3.sec5} and Appendix \ref{app.simulation} later. Restricted strong convexity can be interpreted as an identification condition. To see this, note that $(\Delta_{\beta},L_{0}(u))\in\mathcal{R}_{u}$ for all $\Delta_{\beta}\in\mathbb{R}^{p}$ because $\mathcal{P}_{\Phi(u)}L_{0}(u)=L_{0}(u)$ by construction. Thus by letting $\Delta_{L}=L_{0}(u)$, condition \eqref{eq.rsc} implies that $L_{0}(u)$ is not equal to any linear combination of the $X_{j}$s, a necessary condition in order to identify $L_{0}(u)$.  

Under Assumption \ref{ass.id}, the bound obtained in Lemma \ref{lowerbound}  can be shown to be further lower bounded by $(C_{\sigma}\sigma_{min}^{2}\land 1)C_{min}C_{RSC}(||\Delta_{\beta}||_{F}^{2}+||\Delta_{L}||_{F}^{2}/NT)/\alpha_{NT}^{2}$ for some $C_{\sigma}>0$ w.p.a.1 (see the proof of Theorem \ref{3.thm1} for details). \ref{step1} is thus completed.


\subsubsection{Step 2}
The key property we use to obtain a tight enough upper bound on the absolute empirical process and the penalty difference in equation \eqref{eq.base} is equation \eqref{equiorder} implied by Lemma \ref{3.lem1}. 

For illustrative purpose, let us again assume there are no covariates, so equation \eqref{equiorder} implies that $||\Delta_{L}||_{*}$ has the same order as $||\Delta_{L}||_{F}$, which, in the ball $||\Delta_{L}||_{F}^{2}/NT\leq\gamma^{2}$ we consider, is at most $\sqrt{NT}\gamma$. For the absolute penalty difference in equation \eqref{eq.base}, by triangle inequality, $|||L_{0}(u)+\Delta_{L}||_{*}-||L_{0}(u)||_{*}|$ is upper bounded by $||\Delta_{L}||_{*}\leq 4\sqrt{3r(u)}||\Delta_{L}||_{F}\leq 4\sqrt{3r(u)}\sqrt{NT}\gamma$ by property \eqref{equiorder}. Note that without the property, the upper bound would be $\sqrt{N\land T}\sqrt{NT}\gamma$ instead, much greater than the current one as $r(u)$ is fixed.

Property \eqref{equiorder} also helps to obtain a tight upper bound on the absolute process $|\mathbb{G}_{u}|$ in equation \eqref{eq.base}. Let $A$ denote a generic $N\times T$ matrix. When deriving an upper bound on $|\mathbb{G}_{u}|$, we often need to upper bound terms in the form of $|\sum_{i,t}A_{it}\Delta_{L,it}|$, or equivalently, $|\langle A,\Delta_{L}\rangle|$. There are at least two possible upper bounds for this inner product:
\begin{align}
|\langle A,\Delta_{L}\rangle|&\leq ||A||_{F}\cdot ||\Delta_{L}||_{F}\ \text{or}\label{eq.wrongrelax}\\
|\langle A,\Delta_{L}\rangle|&\leq ||A||\cdot ||\Delta_{L}||_{*}\label{eq.rightrelax}
\end{align}
where inequality \eqref{eq.wrongrelax} is by Cauchy-Schwartz and inequality \eqref{eq.rightrelax} is by Lemma 3.2 in \cite{candes2009exact}. In general, these two upper bounds can be of the same order; although $||A||$ can be as small as $||A||_{F}/\sqrt{N\land T}$, $||\Delta_{L}||_{*}$ can be as large as $\sqrt{N\land T}||\Delta_{L}||_{F}$ without any restrictions. However, now in the ball $||\Delta_{L}||_{F}\leq \sqrt{NT}\gamma$, under equation \eqref{equiorder} without covariates, $||\Delta_{L}||_{*}$ can be bounded by $\sqrt{4\bar{r}}\sqrt{NT}\gamma$. As $\bar{r}$ is fixed, the order of the right side of inequality \eqref{eq.rightrelax} can then be only $1/(\sqrt{N\land T})$ of that of inequality \eqref{eq.wrongrelax}, providing a tight enough upper bound on the inner product.

Finally, in order to bound the absolute process and the penalty difference uniformly in $u\in\mathcal{U}$, we make the following assumption on smoothness in our estimands. This assumption is trivially satisfied if $\mathcal{U}$ is a singleton.
\begin{assumption}\label{3.ass4}
There exist $\zeta_{1},\zeta_{2}>0$ such that 
\begin{align}
||\beta_{0}(u')-\beta_{0}(u)||_{F}\leq &\zeta_{1}|u'-u|,\forall u,u'\in\mathcal{U},\label{smoothb}\\
\frac{1}{\sqrt{NT}}||L_{0}(u')-L_{0}(u)||_{F}\leq &\zeta_{2}|u'-u|,\forall u,u'\in\mathcal{U}\ w.p.a.1.\label{smoothL}
\end{align}
\end{assumption}
The smoothness assumption for $\beta_{0}(\cdot)$, equation \eqref{smoothb}, is the same as in \cite{belloni2011l1}. The smoothness condition on $L_{0}(\cdot)$, equation \eqref{smoothL}, is a matrix counterpart. One can verify that $L_{0}(\cdot)$ in Examples \ref{eg1} satisfies condition \eqref{smoothL} if the quantile function of $\epsilon$, $q_{\epsilon}(\cdot)$, is Lipschitz continuous. For Example \ref{eg2}, $L_{0}(\cdot)$ satisfies condition \eqref{smoothL} if  $q_{\epsilon}(\cdot)$ is Lipschitz continuous and if $\sum_{i,t}(\sum_{m=1}^{\bar{r}_{2}}F_{tk}^{b}\Lambda^{b}_{ik})/NT$ converges in probability to a constant. Note that condition \eqref{smoothL} rules out \textit{some} cases where the set of the effective fixed effects changes on $\mathcal{U}$. To see this, suppose in Example \ref{eg3}, there exists a jumping point $u_{0}\in \mathcal{U}$ such that  $r(u)<r(u')$ for any $u<u_{0}\leq u'$. Suppose the individual- and the time-fixed effects do not depend on $u$ and the first $r(u)$ of them at $u'$ are the same as those at $u$. Then 
\[\frac{1}{\sqrt{NT}}||L_{0}(u')-L_{0}(u)||_{F}=\sqrt{\frac{1}{NT}\sum_{i,t}\left(\sum_{k=r(u)+1}^{r(u')}F_{tk}\Lambda_{ik}\right)^{2}}
\]
which may converge in probability to a positive constant if the law of large number holds for it.
Nevertheless, this assumption is not restrictive even in this situation when there are only a finite number of such jumping points in $(0,1)$: Let $\{u_{0,k}:k=1,2,...,K\}$ $(K<\infty)$ be the set of such jumping points with $u_{0,k}<u_{0,k+1}$ for all $k=1,...,K-1$. If Assumption \ref{3.ass4} holds for compact interval $\mathcal{U}_{k}\subset (u_{0,k},u_{0,k+1})$ for each $k$, we can then establish uniform error bound over each $\mathcal{U}_{k}$ and the uniform bound over $\bigcup_{k=1}^{K}\mathcal{U}_{k}$ is immediately obtained. 

We have the following lemma bounding the absolute process in equation \eqref{eq.base}. The bound on the penalty difference is straightforward and is thus omitted here.

\begin{lem}\label{3.lem3}
Under Assumption \ref{3.ass1} and Assumption \ref{3.ass4}, if $\alpha_{NT}$ and $p$ are such that $\gamma=o(1)$ where $\gamma$ is defined in equation \eqref{eq.t}, then there exists a universal constant $C_{sup}>0$
such that
\begin{align*}
\mathbb{P}\Bigg(\sup_{\substack{u\in\mathcal{U}\\(\Delta_{\beta},\Delta_{L})\in\mathcal{R}_{u}\\||\Delta_{\beta}||_{F}^{2}+\frac{1}{NT}||\Delta_{L}||_{F}^{2}\leq\gamma^{2}}} &\big|\mathbb{G}_{u}\big(\rho_{u} \big(V_{it}(u)-X_{it}' \Delta_{\beta}-\Delta_{L,it}\big)-\rho_{u} (V_{it}(u))\big)\big|\\
\leq &C_{sup}\log(NT)\left(\sqrt{p\log(pNT)}\lor\sqrt{\bar{r}(N\lor T)}\right)\gamma\Bigg)\to 1
\end{align*}
The formula of $C_{sup}$ is in the proof.
\end{lem}
\begin{proof}
See Appendix \ref{app.sec3}.
\end{proof}
\section{Main Result and Its Consequences}\label{3.sec4}
In this section, we formally present the main result on the error bound on $(\hat{\beta}(u),\hat{L}(u))$. We also derive a consistent estimator of the number of fixed effects $r(u)$, i.e., the rank of $L_{0}(u)$. 

Under the assumptions proposed in Section \ref{preview}, we have the following theorem.


\begin{thm}\label{3.thm1}
Under Assumptions \ref{3.ass1}-\ref{3.ass4} and the condition in Lemma \ref{3.lem3}, if $\lambda$ is the same as in Lemma \ref{3.lem1}, then there exists a universal constant $C_{error}>0$ such that the following holds w.p.a.1. 
\small\begin{align}
\sup_{u\in\mathcal{U}}\ \ ||\hat{\beta}(u)-\beta_{0}(u)||_{F}^{2}+\frac{1}{NT}||\hat{L}(u)-L_{0}(u)||_{F}^{2}
\leq& \gamma^{2}\coloneqq C_{error}^{2} \alpha_{NT}^{4}\log(NT)\left(\frac{p\log(pNT)}{NT}\lor \frac{\bar{r}}{N\land T}\right)\label{3.eq3.1}
\end{align}
\normalsize The constant $C_{error}$ is positively associated with $1/ \sigma_{min}^{2}$, $1/\underline{f}$ and $1/C_{RSC}$. Its exact formula is in the proof.
\end{thm}
\begin{proof}
See Appendix \ref{app.sec4}.
\end{proof}

Theorem \ref{3.thm1} implies uniform consistency of our estimator of $\beta_{0}(u)$ and $L_{0}(u)$ over $\mathcal{U}$ for a fixed $\bar{r}$ and slowly growing $\alpha_{NT}$ and $p$. The key determinants of the rate of convergenge are $p\log(pNT)/NT$ and $1/(N\land T)$. This part matches the nuclear norm penalized mean regression literature (\cite{athey2018matrix}, \cite{moon2019nuclear} and \cite{chernozhukov2018inference}). We will discuss this part with more details in this section. Other determinants of the error bound include the number of fixed effects ($\bar{r}$), the quality of lower-bounding the conditional expectation in \ref{step1} ($\alpha_{NT}^{2}$ and $\underline{f}$), and the strength of identification ($\sigma_{min}^{2}$ and $C_{RSC}$). The way that these parameters affect the error bound is expected: A larger $\bar{r}$ results in a relatively higher-rank common component, making the estimation problem more difficult. A greater quadratic lower bound (a greater $\underline{f}$ or a smaller $\alpha_{NT}$) implies that the objective function tends to be more sensitive to perturbations around the global minimum. Finally, stronger identification (a larger $C_{RSC}$ or a larger $\sigma_{min}^{2}$) makes it easier to separate the estimation errors $\hat{\Delta}_{\beta,j}(u)$s and $\hat{\Delta}_{L}(u)$. 

\subsection{On the Rate of $\hat{L}(u)$ and A Rank Estimator}\label{sec4.1}
\cite{agarwal2012noisy} provide a minimax result that implies the (near) optimality of the error bound on $\hat{L}(u)$ obtained in Theorem \ref{3.thm1} when $\alpha_{NT}=\log(NT)$ and $p\log(pNT)/NT=O(\bar{r}/(N\land T))$. In \cite{agarwal2012noisy}, they study a model where an observable $N\times T$ matrix $Y$ is the sum of a low-rank matrix, a sparse matrix and a noise matrix of i.i.d. Gaussian entries. To apply their result to our model, consider a special case where $u=0.5$ and $Y=L_{0}+V$ where $V$ is an $N\times T$ matrix of i.i.d. $N(0,v^{2})$ entries. This model both satisfies the conditional quantile model \eqref{3.eq2} studied in this paper at $u=0.5$ with $\beta(0.5)=0$ and $q_{Y|L_{0}}(0.5)=L_{0}$, and also satisfies their setup with the sparse component being exactly zero. Their Theorem 2 (p.1195) shows that the lower bound on the minimax risk in the squared Frobenius norm over the family $\{L_{0}:\text{rank}(L_{0})\leq \bar{r},||L_{0}||_{\infty}\leq\alpha_{NT}\}$ has the order $\bar{r}/(N\land T)$. Comparing the order of the lower bound and our upper bound on the estimation error under the aforementioned order of $p$, they are equal up to a factor of $\alpha_{NT}^{4}\log(NT)$.

From the error bound on $\hat{L}(u)$, we can obtain the order of the singular values of the estimation error $\hat{\Delta}_{L}(u)$ by Weyl's theorem. Let $\sigma_{1}(u)\geq\cdots\geq \sigma_{r(u)}(u)>0$ be the nonzero singular values of $L_{0}(u)$, and $\hat{\sigma}_{1}(u)\geq\cdots\geq \hat{\sigma}_{N\land T}(u)$ be the singular values of $\hat{L}(u)$. We have the following corollary.
\begin{cor}\label{cor.sv}
Under the conditions in Theorem \ref{3.thm1}, the following holds w.p.a.1:
\begin{align}
&\sup_{u\in\mathcal{U}}\left\lbrace\max\left\lbrace|\hat{\sigma}_{1}(u)-\sigma_{1}(u)|,...,|\hat{\sigma}_{r(u)}(u)-\sigma_{r(u)}(u)|,\hat{\sigma}_{r(u)+1}(u),...,\hat{\sigma}_{N\land T}(u)\right\rbrace\right\rbrace\notag\\
\leq &\sqrt{NT}\gamma\coloneqq C_{error} \alpha_{NT}^{2}\sqrt{\log(NT)\left(p\log(pNT)\lor \bar{r}(N\lor T)\right)}.\label{eq.cor}
\end{align}
where $C_{error}$ is the same as in Theorem \ref{3.thm1}.
\end{cor}
\begin{proof}
See Appendix \ref{app.sec4}.
\end{proof}
Note that the nonzero singular values of $L_{0}(u)$ has order $\sqrt{NT}$ if $L_{0}(u)$ is formed by strong factors and factor loadings or if elements in $L_{0}(u)$ are $O(1)$. Although Theorem \ref{3.thm1} and Corollary \ref{cor.sv} are silent about whether the estimated low-rank component $\hat{L}(u)$ is low-rank or not, Corollary \ref{cor.sv} says that for large enough $N$ and $T$, as long as $p$ and $\alpha_{NT}$ are such that $\gamma=o(1)$, there does exist an arbitrarily large gap between the largest $r(u)$ and the remaining $((N\land T)-r(u))$ singular values of $\hat{L}(u)$. Specifically, the first $r(u)$ singular values of $\hat{L}(u)$ are of the order of $\sqrt{NT}$ while the other singular values are of the order of $\sqrt{NT}\gamma$. This confirms the intuition in Remark \ref{nuclear norm}.

This implication naturally leads to an estimator of $r(u)$. Let $\hat{r}(u)=\sum_{k}\mathbbm{1}(\hat{\sigma}_{k}(u)\geq C_{r})$ for an $(N,T)$-dependent $C_{r}$ such that $\sqrt{NT}\gamma=o(C_{r})$ and $C_{r}=o(\sqrt{NT})$. The following corollary establishes consistency of this estimator.
\begin{cor}\label{cor.rank}
Under the conditions in Theorem \ref{3.thm1}, for any $u\in\mathcal{U}$, suppose all the nonzero singular values of $L_{0}(u)$ are of the order of $\sqrt{NT}$ w.p.a.1, then $\mathbb{P}\left(\hat{r}(u)=r(u)\right)\to 1$.
\end{cor}
\begin{proof}
See Appendix \ref{app.sec4}.
\end{proof}
\begin{rem}
When the rank is stable in a range of quantile levels, one can estimate the rank at multiple quantile levels to improve efficiency.
\end{rem}
\begin{rem}
Consistency can be shown to hold uniformly in $u\in\mathcal{U}$, i.e., $\mathbb{P}(\sup_{u\in\mathcal{U}}|\hat{r}(u)-r(u)|=0)\to 1$, if the required order of the nonzero singular values of $L_{0}(u)$ holds uniformly in $u$ as well.
\end{rem}
\subsection{On the Rate of $\hat{\beta}(u)$}
Theorem \ref{3.thm1} implies uniform consistency of $\hat{\beta}(u)$ if $\alpha_{NT}$ and $p$ are not too large such that $\gamma=o(1)$. When $\bar{r}/(N\land T)=O(p\log(pNT)/NT)$, the rate of convergence of $\hat{\beta}(u)$ is nearly optimal under $\alpha_{NT}=O(\log(NT))$. However, the rate of convergence of $\hat{\beta}(u)$ is slower than optimal when the number of covariates is small such that $p\log(pNT)/(NT)=o(\bar{r}/(N\land T))$. This is because of penalization and because the estimation error of the coefficient estimator and that of the low-rank matrix estimator are not orthogonal under the check function-involved objective function. In this case, the penalized estimator $\hat{\beta}(u)$ can be used as a first-step estimator; using it as an initialization with the rank estimator proposed in Section \ref{sec4.1}, one can adopt the iterative estimator \eqref{eq.estimator.iter} described in the Introduction without penalization. Although this is back to a nonconvex problem, the initial value is already lying in a small neighborhood of the true parameter by uniform consistency. Some rounds of iterations even before convergence is reached may correct the penalization bias and achieve the optimal rate (see \cite{moon2019nuclear} and \cite{chernozhukov2018inference} for mean regressions). The benefits of adopting such a two-step procedure instead of a fully iterative approach are three-fold. First, it provides a consistent initial guess of $\beta_{0}(u)$, potentially avoiding the issue of convergence to a local minimum. Second, as a by-product, the penalized estimation step also provides a rank estimator which is needed for the iterative approach. Third, from the Monte Carlos in Section \ref{3.sec5} and Appendix \ref{app.simulation}, we can see that the penalized estimator saves a significant amount of computation time than the fully iterative one.

\section{Monte Carlo Simulations}\label{3.sec5}
In this section, we illustrate the finite sample performance of our estimator using Monte Carlo simulations. 
\subsection{Data Generating Process}
We consider the following data generating process, which is a special case of Example \ref{eg3} and is adapted from \cite{ando2019quantile}:
\begin{equation*}
Y_{it}=\sum_{j=1}^{3}X_{j,it}\beta_{j}(U_{it})+\sum_{k=1}^{3}\mathbbm{1}_{k}(U_{it})F_{kt}\Lambda_{ki}(U_{it})+\varepsilon_{it}
\end{equation*}
where the $U_{it}$s are independently drawn from $\text{Unif}[0,1]$. The coefficients satisfy
\begin{align*}
\beta_{1}(U_{it})=\beta_{3}(U_{it})=-1+0.1U_{it},\ \ \ \beta_{2}(U_{it})=1+0.1U_{it}
\end{align*}

The indicator functions $\mathbbm{1}_{k}(\cdot):(0,1)\mapsto\{0,1\},k=1,2,3,$ satisfy
\begin{align*}
\mathbbm{1}_{1}(u)=1, \forall u\in (0,1),\ \ \ \ \mathbbm{1}_{2}(u)=\mathbbm{1}(u>0.3),\ \ \ \mathbbm{1}_{3}(u)=\mathbbm{1}(u>0.7).
\end{align*}

We draw the time fixed effects $F_{1t},F_{2t}$ and $F_{3t}$ independently from $\text{Unif}[0,2]$. We generate the individual fixed effects as $\Lambda_{ki}(U_{it})=\chi_{ki}+0.1U_{it}$ where the $\chi_{ki}$s are independently drawn from $\text{Unif}[0,1]$ for $k=1,2,3$. The covariates are generated by:
\begin{align*}
X_{j,it}=\eta_{j,it}+\phi\cdot \left( F_{jt}^{2}+\chi_{ji}^{2}\right),j=1,2,3
\end{align*}
where $\eta_{j,it}$ are independent $\text{Unif}[0,2]$ for all $j=1,2,3$. The parameter $\phi\in\{0.1,0.2,0.3\}$ governs the correlation between the covariates and the fixed effects. Finally, $\varepsilon_{it}$ is generated by $G^{-1}(U_{it})$ where $G$ is the cumulative distribution function of either the standard normal distribution or student's $t$-distribution with degree of freedom 2. 

Since $X_{j,it}$ and $F_{kt}$ are positive for all $i,t,j$ and $k$ almost surely and $\beta_{j}(\cdot)$, $\Lambda_{ki}(\cdot)$ and $G^{-1}(\cdot)$ are strictly increasing  (almost surely) for all $i,j$ and $k$, $Y_{it}$ is strictly increasing in $U_{it}$ almost surely. Let $\textbf{1}_{N\times T}$ be an $N\times T$ matrix of all ones. The $u$-th conditional quantile of $Y$ is thus $q_{Y|W}(u)=\sum_{j=1}^{3}X_{j}\beta_{j}(u)+L_{0}(u)$ where
\begin{align*}
L_{0}(u)=\begin{cases}
G^{-1}(u)\textbf{1}_{N\times T}+\Lambda_{1}(u)F_{1}',\text{ if }u\leq 0.3\\
G^{-1}(u)\textbf{1}_{N\times T}+\Lambda_{1}(u)F_{1}'+\Lambda_{2}(u)F_{2}',\text{ if }0.3< u\leq 0.7\\
G^{-1}(u)\textbf{1}_{N\times T}+\Lambda_{1}(u)F_{1}'+\Lambda_{2}(u)F_{2}'+\Lambda_{3}(u)F_{3}',\text{ if }u> 0.7
\end{cases}
\end{align*}

From the model, one can see that the rank of $L_{0}(u)$ and the set of effective fixed effects vary in $u$. Also, the model allows the covariates to be correlated with $L_{0}(u)$. Higher correlation (greater $\phi$) would make it more difficult to separately estimate $\beta(u)$s and the common component $L_{0}(u)$ since it tends to yield a smaller restricted strong convexity constant $C_{RSC}$. Finally, $Y_{it}$ is allowed to have heavy tails because $\varepsilon_{it}$ can be student's $t$-distributed.

\subsection{Evaluate the Performance}
To illustrate the performance of the estimator, we conduct Monte Carlo simulations with various sample sizes at $\phi\in\{0.1,0.2,0.3\}$ and $u\in\{0.2,0.5,0.8\}$ with $G$ equal to the cumulative distribution function of either the standard normal or student's \textit{t}-distribution with degree of freedom 2. As for the sample size, $(N,T)\in \{(200,200),(300,300),(400,400),(500,500),\newline(200,300),(300,200), (200,400),(400,200),(200,500),(500,200)\}$. The first four sample sizes with $N=T$ allow us to see the convergence of the estimator. The other six sample sizes with $N\neq T$ and $N\land T=200$ allow us to test the theory which suggests that the rate of convergence does not depend on $N\lor T$ under a fixed $p$.

We use multiple measures to evaluate the estimator's performance. For the three coefficients $\beta_{j}(u)$s, we compute their average squared bias $\text{Bias}_{\beta}^{2}$ and variance $\text{Var}_{\beta}$ over 100 simulation replications\footnote{Specifically, $\text{Bias}_{\beta}^{2}\coloneqq\sum_{j=1}^{3}\left(\sum_{b=1}^{100}(\hat{\beta}_{j,b}(u)-\beta_{j}(u))/100\right)^{2}/3$ where $\hat{\beta}_{j,b}(u)$ is the estimator of $\beta_{j}(u)$ in the $b$-th simulation. The variance $\text{Var}_{\beta}\coloneqq\sum_{j=1}^{3}\left[\sum_{b=1}^{100}\hat{\beta}_{j,b}(u)^{2}/100-\left(\sum_{b=1}^{100}\hat{\beta}_{j,b}(u)/100\right)^{2}\right]/3$.}. For the low-rank component, we compute the average squared Frobenius norm of the estimation error:
\begin{align*}
\text{MSE}_{L}\coloneqq \frac{1}{100}\sum_{b=1}^{100}\left(\frac{\sum_{i,t}\left(\hat{L}_{it,b}(u)-L_{0,it}(u)\right)^{2}}{NT}\right)
\end{align*}
where $\hat{L}_{b}(u)$ is the estimated $L_{0}(u)$ in the $b$-th simulation ($b=1,...,100$). Besides these measures, we also look at the mean squared error of the estimated conditional quantile function $\hat{q}_{Y|W}(u)$, defined as follows \citep{ando2019quantile}:
\begin{align*}
\text{MSE}_{q}\coloneqq \frac{1}{100}\sum_{b=1}^{100}\left(\frac{\big|\big|\sum_{j=1}^{3}X_{j}\left(\hat{\beta}_{j,b}(u)-\beta_{j}(u)\right)+\left(\hat{L}_{b}(u)-L_{0}(u)\right)\big|\big|_{F}^{2}}{NT}\right)
\end{align*}

Note that $\text{MSE}_{q}$ can be small even if $\text{MSE}_{L}$ is large; the latter is also affected by $C_{RSC}$, which in turn, is affected by the correlation between the covariates and $L_{0}(u)$ (determined by parameter $\phi$). Hence, by comparing $\text{MSE}_{q}$ with $\text{MSE}_{L}$ across different values of $\phi$, we can see how restricted strong convexity affects the results.

Finally, we also record the average computation time of the estimator.

To make comparison between the nuclear norm penalized estimator proposed in this paper and alternative estimators, we compute the following iterative estimator adapted from \cite{ando2019quantile} as introduced in the Introduction, and the pooled estimator:
\begin{align}
(\hat{\beta}^{It}(u),\hat{\Lambda}(u),\hat{F})=&\arg\min_{\beta,F,\Lambda}\frac{1}{NT}\sum_{i,t}\rho_{u}\left(Y_{it}-X_{it}'\beta-\Lambda_{i}'F_{t}\right),\hat{L}^{It}(u)=\hat{\Lambda}(u)\hat{F}'\label{estimatorIt}\\
\hat{\beta}^{Po}(u)=&\arg\min_{\beta}\frac{1}{NT}\sum_{i,t}\rho_{u}\left(Y_{it}-X_{it}'\beta\right)\label{estimatorPo}
\end{align}
The finite sample bias of these two estimators provide us with two benchmarks for bias analysis; the bias of any reasonably estimator should be close to the iterative estimator and at least smaller than the pooled estimator. As discussed in the Introduction, $(\hat{\beta}^{It}(u),\hat{\Lambda}(u),\hat{F})$ is obtained by a nonconvex problem which needs the number of fixed effects $r(u)$ to be known. Nonconvexity may result in convergence to local minima, but when the global minimum is indeed achieved and $r(u)$ is correctly set, its finite sample bias is presumably small as there is no penalization. On the other hand, the pooled estimator $\hat{\beta}^{Po}$ completely ignores the fixed effects. Since the covariates are correlated with the fixed effects by construction, $\hat{\beta}^{Po}$ is inconsistent. By comparing the bias of our penalized estimator with these two estimators, we can see what regularization buys and costs.

\subsection{Implementation}
For our penalized estimator, we present the details about our algorithm in Appendix \ref{app.algorithm}. The algorithm is adapt from the Augmented Lagrangian Multiplier method proposed in \cite{lin2010augmented}, \cite{candes2011robust} and \cite{yuan2013sparse}. The method was originally designed for $u=0.5$ with no covariates. We extend it to accommodate any $u\in(0,1)$ with covariates. From the simulation results, the new algorithm works fast and well. For $\lambda$, here we simply set it to be $\lambda=\log(NT)\sqrt{N\lor T}/(3.6NT)$. This choice of $\lambda$ slightly slows down the theoretical rate of convergence because of the additional $\log(NT)$, but it works no matter what $C_{Cone}$ is. Alternatively, one may use cross-validation or adopt the BIC criterion proposed in \cite{belloni2019high} to select $\lambda$. 

For the iterative estimator, we compute $(\hat{\beta}^{It}(u),\hat{\Lambda}(u),\hat{F})$ by iteratively running quantile regression to update $\beta(u),\Lambda(u)$ and $F$ using the \textit{true} number of the fixed effects $r(u)$. This algorithm is adapted from \cite{ando2019quantile} which was for $i$-specific $\beta(u)$. Specicically, for each $i$, obtain $\Lambda_{i}(u)$ by quantile regression using the time series variation by fixing $\beta$ and $F$. For each $t$, update $F_{t}$ by quantile regression using the cross-sectional variation by fixing $\Lambda(u)$ and $\beta$. Then fixing $\Lambda(u)$ and $F$, update $\beta$ by pooled quantile regression. Iterate until converged. For initialization, for each $u$, use $\hat{\beta}^{Po}(u)$ as the initial value of $\beta(u)$, obtain the residual matrix $R$, and initialze the $F$ matrix by the eigenvector matrix of $R'R$ multiplied by $\sqrt{T}$. The termination criterion for the iteration, adopted from \cite{ando2019quantile}, is set to be the same as that for the nuclear norm penalized estimator (see Appendix \ref{app.algorithm} for details) to make the results, especially the computation time, comparable. 
\subsection{Results}
Tables \ref{tab.n.0.2} and \ref{tab.t.0.2} present the results for $\phi=0.2$ with standard normal and student $t$-distributed $\varepsilon_{it}$, respectively. The results for $\phi=0.1$ and $0.3$ are provided in Appendix \ref{app.simulation}.
The Columns Nu, It and Po report results of our penalized estimator, the iterative estimator \eqref{estimatorIt} and the pooled estimator \eqref{estimatorPo}. The column Time reports on average how many seconds each estimator takes. All experiments were performed in MATLAB 2019b under Windows 10 on a desktop computer with 10-core 2.8GHz Intel i9 processor and 16GB RAM using parallel computing. From the results, we have the following key observations.

First, for the penalized estimator (Columns Nu), we can see that in all specifications, the bias and variance of $\hat{\beta}(u)$ shrink as $N$ and $T$ \textit{both} increase for all $u$ considered (see the $\text{Bias}^{2}_{\beta}\times 100$ and $\text{Var}_{\beta}\times 10^{4}$ columns). The same holds for the estimation errors of the estimated low-rank common component $\hat{L}(u)$ and the conditional quantile function $\hat{q}_{Y|W}(u)$. In particular, the results are robust to heavy-tailed data (Table \ref{tab.t.0.2}). On the other hand, when $N\land T$ is fixed, increasing $N\lor T$ does not lead to big improvement in the penalized estimator's performance. This echoes Theorem \ref{3.thm1} which says the rate of convergence is dominated by $1/(N\land T)$ for a fixed $p$. Meanwhile, we can see that the average squared estimation error of the conditional quantile function, $\text{MSE}_{q}$, is smaller than that of the low-rank common component, $\text{MSE}_{L}$. This is due to the separation issue and restricted strong convexity, arisen from the correlation between the covariate matrices and the low-rank component. Together with the results in Appendix \ref{app.simulation}, we can see that $\text{MSE}_{L}$ becomes smaller and closer to $\text{MSE}_{q}$ when $\phi$, and thus the correlation between the covariates and the low-rank component, gets smaller. This is because a smaller $\phi$ tends to yield a larger restricted strong convexity constant $C_{RSC}$.

Next, let us compare the performance of the different estimators we consider. Comparing the pooled estimator (Column Po) and the penalized estimator, we can see that even though the latter is biased due to regularization, the bias is still much smaller than that of the pooled estimator, whose bias does not shrink at all as the sample size increases due to endogeneity. Now let us compare our penalized estimator with the iterative estimator (Columns It), which seems not to suffer from the local minima issue in these experiments. The penalized estimator has similar performance to the iterative one with relatively larger bias and smaller variance
    \begin{table}[H]
\centering
\caption{Standard Normal Error. Parameter $\phi=0.2$}\label{tab.n.0.2}
\begin{tabular}{ccccccccccccc}
\hline
\hline
&&\multicolumn{3}{c}{$\text{Bias}^{2}_{\beta}\times 10^{2}$}&  \multicolumn{2}{c}{$\text{Var}_{\beta}\times 10^{4}$} &\multicolumn{2}{c}{$\text{MSE}_{L}$}&\multicolumn{2}{c}{$\text{MSE}_{q}$}&\multicolumn{2}{c}{\begin{tabular}{@{}c@{}}Time\\(second)
\end{tabular}}\\
\cline{2-13}
$u$& $(N,T)$ & Nu & It & Po  & Nu & It  &Nu & It   & Nu & It & Nu & It \\
\hline
\multirow{10}{*}{0.2}&$(200,200)$&$2.67$&$0.08$&$11.5$&$4.57$&$2.07$&$0.32$&$0.05$&$0.24$&$0.04$&$1$&$55$\tabularnewline
&$(300,300)$&$1.97$&$0.04$&$11.9$&$2.03$&$0.86$&$0.22$&$0.03$&$0.14$&$0.02$&$3$&$153$\tabularnewline
&$(400,400)$&$1.48$&$0.03$&$11.7$&$1.04$&$0.48$&$0.18$&$0.02$&$0.09$&$0.02$&$9$&$409$\tabularnewline
&$(500,500)$&$1.19$&$0.02$&$11.8$&$0.55$&$0.27$&$0.15$&$0.01$&$0.07$&$0.01$&$18$&$742$\tabularnewline
\cline{2-13}
&$(200,300)$&$2.62$&$0.07$&$11.9$&$3.36$&$1.19$&$0.29$&$0.04$&$0.18$&$0.03$&$2$&$87$\tabularnewline
&$(300,200)$&$2.59$&$0.07$&$11.8$&$3.48$&$1.27$&$0.29$&$0.04$&$0.18$&$0.03$&$2$&$87$\tabularnewline
\cline{2-13}
&$(200,400)$&$2.50$&$0.06$&$11.7$&$2.68$&$0.99$&$0.28$&$0.03$&$0.16$&$0.03$&$3$&$133$\tabularnewline
&$(400,200)$&$2.49$&$0.05$&$11.8$&$2.51$&$0.97$&$0.28$&$0.03$&$0.16$&$0.03$&$3$&$133$\tabularnewline
\cline{2-13}
&$(200,500)$&$2.43$&$0.05$&$11.8$&$2.15$&$0.83$&$0.28$&$0.03$&$0.15$&$0.03$&$4$&$189$\tabularnewline
&$(500,200)$&$2.45$&$0.05$&$11.9$&$$2.35$$&$0.76$&$0.28$&$0.03$&$0.15$&$0.03$&$4$&$188$\tabularnewline
\hline
\multirow{10}{*}{0.5}&$(200,200)$&$0.31$&$0.12$&$22.0$&$3.14$&$3.49$&$0.36$&$0.16$&$0.26$&$0.14$&$1$&$46$\tabularnewline
&$(300,300)$&$0.18$&$0.05$&$22.0$&$1.10$&$1.88$&$0.19$&$0.13$&$0.13$&$0.12$&$3$&$123$\tabularnewline
&$(400,400)$&$0.13$&$0.03$&$21.9$&$0.58$&$1.12$&$0.13$&$0.12$&$0.09$&$0.11$&$7$&$303$\tabularnewline
&$(500,500)$&$0.11$&$0.02$&$22.0$&$0.32$&$0.75$&$0.10$&$0.12$&$0.07$&$0.11$&$13$&$467$\tabularnewline
\cline{2-13}
&$(200,300)$&$0.24$&$0.08$&$21.9$&$1.65$&$2.56$&$0.21$&$0.14$&$0.14$&$0.13$&$1$&$73$\tabularnewline
&$(300,200)$&$0.24$&$0.07$&$22.1$&$1.66$&$2.44$&$0.21$&$0.15$&$0.14$&$0.14$&$1$&$67$\tabularnewline
\cline{2-13}
&$(200,400)$&$0.22$&$0.06$&$22.6$&$1.28$&$2.13$&$0.17$&$0.14$&$0.10$&$0.13$&$2$&$109$\tabularnewline
&$(400,200)$&$0.22$&$0.06$&$22.2$&$1.26$&$1.81$&$0.18$&$0.14$&$0.10$&$0.13$&$2$&$112$\tabularnewline
\cline{2-13}
&$(200,500)$&$0.20$&$0.06$&$21.7$&$1.07$&$1.93$&$0.15$&$0.13$&$0.08$&$0.12$&$3$&$148$\tabularnewline
&$(500,200)$&$0.20$&$0.05$&$22.1$&$0.98$&$1.70$&$0.15$&$0.13$&$0.08$&$0.12$&$3$&$151$\tabularnewline
\hline
\multirow{10}{*}{0.8}&$(200,200)$&$0.22$&$0.31$&$35.6$&$6.36$&$5.19$&$1.01$&$0.60$&$0.73$&$0.48$&$2$&$75$\tabularnewline
&$(300,300)$&$0.11$&$0.11$&$34.9$&$2.43$&$1.95$&$0.48$&$0.29$&$0.35$&$0.25$&$7$&$170$\tabularnewline
&$(400,400)$&$0.06$&$0.05$&$34.8$&$1.01$&$1.06$&$0.29$&$0.21$&$0.22$&$0.19$&$18$&$455$\tabularnewline
&$(500,500)$&$0.04$&$0.03$&$35.3$&$0.61$&$0.72$&$0.20$&$0.17$&$0.16$&$0.17$&$34$&$783$\tabularnewline
\cline{2-13}
&$(200,300)$&$0.17$&$0.20$&$34.5$&$3.59$&$3.27$&$0.66$&$0.41$&$0.47$&$0.34$&$4$&$114$\tabularnewline
&$(300,200)$&$0.17$&$0.17$&$34.9$&$3.41$&$2.71$&$0.67$&$0.41$&$0.47$&$0.34$&$4$&$105$\tabularnewline
\cline{2-13}
&$(200,400)$&$0.15$&$0.16$&$35.6$&$2.74$&$2.59$&$0.57$&$0.35$&$0.40$&$0.29$&$5$&$159$\tabularnewline
&$(400,200)$&$0.16$&$0.12$&$34.7$&$2.68$&$2.41$&$0.57$&$0.33$&$0.40$&$0.29$&$5$&$147$\tabularnewline
\cline{2-13}
&$(200,500)$&$0.14$&$0.15$&$34.8$&$2.19$&$2.07$&$0.53$&$0.32$&$0.37$&$0.27$&$8$&$200$\tabularnewline
&$(500,200)$&$0.14$&$0.09$&$34.9$&$1.96$&$1.82$&$0.53$&$0.30$&$0.38$&$0.26$&$8$&$199$\tabularnewline
\hline
\multicolumn{13}{l}{\footnotesize \textit{Note}: Columns Nu, It and Po report the results of the nuclear norm penalized estimator proposed in}\tabularnewline
\multicolumn{13}{l}{\footnotesize this paper, the iterative estimator \eqref{estimatorIt} and the pooled estimator \eqref{estimatorPo}, averaged over 100 simulations.}
     \end{tabular}
    \end{table}
\noindent  in many cases  (e.g. $u=0.5$ in all cases). Admittedly, the penalized estimates of the coefficients tend to have a larger mean squared error since the theoretical convergence rate is slow. However, this is not always the case especially when the sample size is relatively small. For instance, when $u=0.8$, $(N,T)=(200,200)$ for standard normal error, and $(N,T)=(200,200),(300,300)$, etc. for student $t$-distributed error, the penalized estimator

     \begin{table}[H]
\centering
\caption{Student $t$-Distributed Error (Degree of Freedom$=2$). Parameter $\phi=0.2$}\label{tab.t.0.2}
\begin{tabular}{ccccccccccccc}
\hline
\hline
&&\multicolumn{3}{c}{$\text{Bias}^{2}_{\beta}\times 10^{2}$}&  \multicolumn{2}{c}{$\text{Var}_{\beta}\times 10^{4}$} &\multicolumn{2}{c}{$\text{MSE}_{L}$}&\multicolumn{2}{c}{$\text{MSE}_{q}$}&\multicolumn{2}{c}{\begin{tabular}{@{}c@{}}Time\\(second)
\end{tabular}}\\
\cline{2-13}
$u$& $(N,T)$ & Nu & It & Po  & Nu & It &Nu & It   & Nu & It & Nu & It \\
\hline
\multirow{10}{*}{0.2}&$(200,200)$&$3.61$&$0.17$&$10.0$&$7.98$&$4.52$&$0.48$&$0.11$&$0.42$&$0.10$&$1$&$48$\tabularnewline
&$(300,300)$&$2.84$&$0.08$&$10.4$&3.67&1.68&0.34&0.06&0.25&0.05&4&151\tabularnewline
&$(400,400)$&2.18&0.05&10.2&1.92&1.02&0.27&0.04&0.17&0.04&13&404\tabularnewline
&$(500,500)$&1.79&0.03&10.3&1.14&0.65&0.23&0.03&0.14&0.03&26&730\tabularnewline
\cline{2-13}
&$(200,300)$&3.58&0.14&10.5&5.58&3.10&0.42&0.08&0.30&0.07&2&85\tabularnewline
&$(300,200)$&3.48&0.14&10.3&6.42&3.18&0.41&0.08&0.30&0.07&2&81\tabularnewline
\cline{2-13}
&$(200,400)$&3.33&0.12&10.3&4.60&2.48&0.39&0.07&0.26&0.06&3&128\tabularnewline
&$(400,200)$&3.30&0.11&10.2&4.40&2.21&0.39&0.07&0.26&0.06&3&124\tabularnewline
\cline{2-13}
&$(200,500)$&3.29&0.09&10.5&3.95&1.79&0.39&0.06&0.25&0.06&6&185\tabularnewline
&$(500,200)$&3.31&0.10&10.4&4.39&1.79&0.39&0.06&0.25&0.06&5&184\tabularnewline
\hline
\multirow{10}{*}{0.5}&$(200,200)$&0.38&0.14&22.6&4.49&4.50&0.45&0.18&0.32&0.16&1&45\tabularnewline
&$(300,300)$&0.22&0.05&22.9&1.23&2.10&0.24&0.14&0.16&0.13&4&116\tabularnewline
&$(400,400)$&0.15&0.02&22.5&0.73&1.28&0.16&0.13&0.11&0.13&10&267\tabularnewline
&$(500,500)$&0.13&0.02&22.5&0.47&0.85&0.13&0.12&0.08&0.12&19&418\tabularnewline
\cline{2-13}
&$(200,300)$&0.31&0.10&23.3&2.21&2.51&0.27&0.16&0.16&0.14&2&72\tabularnewline
&$(300,200)$&0.31&0.10&22.8&2.24&2.71&0.27&0.16&0.17&0.15&2&70\tabularnewline
\cline{2-13}
&$(200,400)$&0.26&0.09&22.5&1.55&2.12&0.21&0.15&0.12&0.13&3&114\tabularnewline
&$(400,200)$&0.27&0.06&22.7&1.87&2.18&0.21&0.16&0.12&0.15&3&92\tabularnewline
\cline{2-13}
&$(200,500)$&0.27&0.06&23.0&1.33&2.15&0.19&0.15&0.10&0.14&4&139\tabularnewline
&$(500,200)$&0.24&0.06&22.6&1.16&2.23&0.18&0.15&0.10&0.14&5&133\tabularnewline
\hline
\multirow{10}{*}{0.8}&$(200,200)$&0.31&0.82&35.0&11&11&1.30&1.06&0.92&0.83&2&74\tabularnewline
&$(300,300)$&0.17&0.20&36.3&3.73&3.53&0.66&0.51&0.47&0.45&6&257\tabularnewline
&$(400,400)$&0.10&0.09&35.8&1.84&1.95&0.41&0.30&0.30&0.28&18&528\tabularnewline
&$(500,500)$&0.08&0.06&36.3&1.17&1.34&0.30&0.24&0.22&0.23&34&882\tabularnewline
\cline{2-13}
&$(200,300)$&0.26&0.50&35.6&6.20&6.68&0.87&0.79&0.59&0.65&3&152\tabularnewline
&$(300,200)$&0.27&0.43&35.9&5.63&6.49&0.87&0.78&0.59&0.66&3&135\tabularnewline
\cline{2-13}
&$(200,400)$&0.25&0.35&36.5&4.19&5.23&0.75&0.64&0.50&0.55&5&219\tabularnewline
&$(400,200)$&0.27&0.30&35.9&4.38&5.02&0.75&0.64&0.50&0.56&5&212\tabularnewline
\cline{2-13}
&$(200,500)$&0.24&0.31&36.0&3.19&4.05&0.69&0.56&0.47&0.48&8&302\tabularnewline
&$(500,200)$&0.22&0.19&36.2&3.19&3.52&0.67&0.51&0.46&0.46&8&283\tabularnewline
\hline
\multicolumn{13}{l}{\footnotesize \textit{Note}: Columns Nu, It and Po report the results of the nuclear norm penalized estimator proposed in}\tabularnewline
\multicolumn{13}{l}{\footnotesize this paper, the iterative estimator \eqref{estimatorIt} and the pooled estimator \eqref{estimatorPo}, averaged over 100 simulations.}
     \end{tabular}
    \end{table}
\noindent has smaller bias and smaller mean squared error in the coefficients than the iterative estimator.

In terms of computation speed, the penalized estimator in these experiments is much faster than the iterative one. The former rarely takes more than half a minute and in most cases only a few seconds. But the iterative estimator takes about 20-40 times longer, and when the sample size is large ($N=T=500$), it can take almost 15 minutes. Note that when we compute the iterative estimator, we treat the number of fixed effects as known. In practice, this additional parameter also needs to be estimated, so the actual computation time can be even longer.


\section{Concluding Remarks}\label{3.sec6}
  
In this paper, we study a conditional quantile panel data model with interactive fixed effects. By exploiting the low-rankness of the matrix formed by the interactive fixed effects, we propose a nuclear norm penalized estimator. The estimator jointly estimates the coefficients and the low-rank matrix by solving a convex problem. We derive a uniform error bound on the estimator and in turn, establish uniform consistency. Based on the error bound, we also construct a consistent estimator of the number of fixed effects at any given quantile level. From the Monte Carlo simulations, our estimator performs well and computation is much more efficient than a related iterative estimator. 

We conjecture that after the penalized estimator and the rank estimator are obtained, by using them as the initial value, a few rounds of iterations based on the iterative estimator's minimization problem would remove the bias and restore the rate of convergence of the coefficient estimator that is slowed down by penalization. Inference may also be available based on such post-penalization procedures. We leave these for future work.

\newpage
\begin{appendices}
\section{Implementation of the Estimator}\label{app.algorithm}
We adapt an augmented Lagrange multiplier (ALM) algorithm introduced in \cite{lin2010augmented}, \cite{candes2011robust} and \cite{yuan2013sparse}. Rewrite the minimization problem \eqref{3.eq3} as\footnote{In theory, the estimator \eqref{3.eq3} proposed in Section \ref{3.sec2} solves a constrained minimization problem with $||L||_{\infty}\leq \alpha_{NT}$. However, since $\alpha_{NT}$ is allowed to grow to infinity with $N$ and $T$, in practice we can set $\alpha_{NT}$ as a large positive number, solve an unconstrained problem using the algorithm in this appendix, and check whether the obtained $\hat{L}(u)$ satisfies the constraint. Moreover, under the alternative assumptions in Appendix \ref{app.second}, the constraint in the minimization problem can be dropped.}
\begin{align*}
\min_{L,V,\beta}\ \ &\frac{1}{\lambda NT}\bm{\rho}_{u}(V)+||L||_{*}\\
s.t.\ \ & \sum_{j=1}^{p}X_{j}\beta_{j}+L+V=Y
\end{align*}
The ALM method is based on the augmented Lagrangian
\begin{equation}
l(L,\beta,V,H)=\frac{1}{\lambda NT}\bm{\rho}_{u}(V)+||L||_{*}+\left\langle H,Y-\sum_{j=1}^{p}X_{j}\beta_{j}-L-V\right\rangle+\frac{\mu}{2}||Y-\sum_{j=1}^{p}X_{j}\beta_{j}-L-V||_{F}^{2}
\end{equation}
where $H\in\mathbb{R}^{N\times T}$ is the Lagrangian multiplier of the linear constraint $\sum_{j=1}^{p}X_{j}\beta_{j}+L+V=Y$ and $\mu>0$ is the penalty parameter for the violation of the constraint. By separability of the parameters in $l$, the ALM method iteratively updates $L,\beta,V$ and $H$ one at a time until converged. Given the $k$-th step $L^{(k)},\beta^{(k)}$, $V^{(k)}$ and $H^{(k)}$, ALM updates $L$, $\beta$ and $V$ by the first order condition and $H$ as follows:
\begin{align}
&L\text{-minimization: }0\in \nabla||L^{(k+1)}||_{*}-\left(H^{(k)}-\mu\left(V^{(k)}+\sum_{j=1}^{p}X_{j}\beta_{j}^{(k)}+L^{(k+1)}-Y\right)\right)\label{Lstep}\\
&V\text{-minimization: }0\in \frac{1}{\lambda NT}\nabla \bm{\rho}_{u}(V^{(k+1)})-\left(H^{(k)}-\mu\left(V^{(k+1)}+\sum_{j=1}^{p}X_{j}\beta_{j}^{(k)}+L^{(k+1)}-Y\right)\right)\label{Vstep}\\
&\beta\text{-minimization: }0=\left\langle H^{(k)},X_{j} \right\rangle+\mu \left\langle Y-\sum_{j=1}^{p}X_{j}\beta_{j}^{(k+1)}-L^{(k+1)}-V^{(k+1)}, X_{j}\right\rangle ,\forall j=1,...,p\label{bstep}\\
& H\text{-minimization: } 	H^{(k+1)}=H^{(k)}-\mu\left(V^{(k+1)}+\sum_{j=1}^{p}X_{j}\beta_{j}^{(k+1)}+L^{(k+1)}-Y\right)\label{Wstep}
\end{align}
where $\nabla$ denotes the subgradient operator. It can be verified that the three first order conditions \eqref{Lstep} to \eqref{bstep}  have explicit solutions:

For equation \eqref{Lstep}, let $R^{(k)}\text{diag}(\{\sigma^{(k)}_{j}\}_{j})S^{(k)'}$ be a singular value decomposition of the matrix $(Y-V^{(k)}-\sum_{j=1}^{p}X_{j}\beta_{j}^{(k)}+H^{(k)}/\mu)$. According to \cite{yuan2013sparse}, the solution to equation \eqref{Lstep} is
\begin{equation}\label{Lupdate}
L^{(k+1)}=R^{(k)}\text{diag}\left(\max\left\lbrace\sigma^{(k)}_{j}-\frac{1}{\mu},0\right\rbrace\right)S^{(k)'}
\end{equation}

For equation \eqref{Vstep}, let $\Gamma^{(k+1)}_{V}=H^{(k)}/\mu-\sum_{j=1}^{p}X_{j}\beta_{j}^{(k)}-L^{(k+1)}+Y$. For every $i=1,...,N$ and $t=1,...,T$, $\left(\nabla\bm{\rho}_{u}\left(V^{(k+1)}\right)\right)_{it}=u\mathbbm{1}(V_{it}^{(k+1)}>0)+(u-1)\mathbbm{1}(V_{it}^{(k+1)}<0)$. It can be verified that the following is a solution:
\begin{equation}\label{Vupdate}
V^{(k+1)}_{it}=\begin{cases}
\max\left\lbrace\Gamma^{(k+1)}_{V,it}-\frac{u}{\mu\lambda NT},0\right\rbrace,& \text{if }\Gamma^{(k+1)}_{V,it}\geq 0\\
-\max\left\lbrace -\Gamma^{(k+1)}_{V,it}-\frac{1-u}{\mu\lambda NT},0\right\rbrace,& \text{if }\Gamma^{(k+1)}_{V,it}<0
\end{cases}
\end{equation}

For equation \eqref{bstep}, it is the first order condition of a least square problem. Let $\Gamma^{(k+1)}_{\beta}=Y-L^{(k+1)}-V^{(k+1)}+H^{(k)}/\mu$. Define the $NT\times p$ matrix $\bm{X}=(\text{vec}(X_{1}),...,\text{vec}(X_{p}))$. Then,
\begin{equation}\label{bupdate}
\beta^{(k+1)}=\left(\bm{X}'\bm{X}\right)^{-1}\left(\bm{X}'\text{vec}\left(\Gamma^{(k+1)}_{\beta}\right)\right)
\end{equation}

Finally, following \cite{yuan2013sparse}, we set $\mu=0.25NT/||Y||_{1}$. The termination criterion\footnote{We also experimented with different termination criteria for instance by including $||V^{(k+1)}-V^{(k)}||_{F}^{2}/NT$ and/or $||H^{(k+1)}-H^{(k)}||_{F}^{2}/NT$ (as in \cite{lin2010augmented} and \cite{candes2011robust}). The results (including computation time) are almost the same, so we adopt the current criterion so that we can compare the penalized estimator with the iterative estimator \eqref{estimatorIt} under the same termination criterion as the latter does not involve $V$ and $H$.} is set as $||\beta^{(k+1)}-\beta^{(k)}||_{F}^{2}/p+||L^{(k+1)}-L^{(k)}||_{F}^{2}/NT\leq 10^{-6}$. The following algorithm summarizes these steps.
\begin{algorithm}
\SetAlgoLined
\KwInput{$\beta^{0}=\bm{0},V^{0}=H^{0}=\bm{0},\mu=0.25NT/||Y||_{1}, \lambda=\log(NT)\sqrt{N\lor T}/(3.6NT)$.}
\While{not converged}{
compute $L^{(k+1)}$ as \eqref{Lupdate}\;
compute $V^{(k+1)}$ as \eqref{Vupdate}\;
compute $\beta^{(k+1)}$ as \eqref{bupdate}\;
compute $H^{(k+1)}$ as \eqref{Wstep}\;
}
\KwOutput{$\beta,L$.}
\caption{Nuclear Norm Penalized Quantile Regression by ALM}
\end{algorithm}

\section{More on the Lower Bound in \ref{step1}}\label{app.lowerbound}
In this appendix, we first introduce a second set of assumptions under which a similar quadratic lower bound as in Lemma \ref{lowerbound} can be obtained for \ref{step1}. Under these assumptions, we can drop the constraint in the minimization problem \eqref{3.eq3} that defines our estimator. We then compare the assumptions in this paper with those in \cite{ando2019quantile}, \cite{belloni2019high} and \cite{chen2019quantile}. 

\subsection{Dropping the Constraint in Equation \eqref{3.eq3} }\label{app.second}
In this section, we maintain Assumption \ref{3.ass2} on the conditional density and add a new assumption so that the requirement $||\hat{L}(u)||_{\infty}\leq \alpha_{NT}$ can be dropped while a similar lower bound as in Lemma \ref{lowerbound} can still be obtained. To illustrate the intuition, let us consider the case without covariates. 

Recall in Section \ref{secStrategy}, our goal is to lower bound the following quantity by $||\Delta_{L}||_{F}^{2}$ multiplied by some constants for all $\Delta_{L}\in\mathcal{D}\coloneqq \{\Delta_{L}\in\mathbb{R}^{N\times T}:||\Delta_{L}||_{\infty}\leq 2\alpha_{NT}\}$:
\begin{equation}\label{eq.display}
\sum_{i,t}\int_{0}^{\Delta_{L,it}}\left(F_{V_{it}(u)|W}(s)-F_{V_{it}(u)|W}(0)\right)ds
\end{equation}
We only focus on $\mathcal{D}$ because we can show that the estimation error $\hat{\Delta}_{L}(u)$ lies in $\mathcal{D}$ uniformly in $u\in\mathcal{U}$ w.p.a.1 under the constrained parameter space $\mathcal{L}$ and by $||L_{0}(u)||_{\infty}\leq \alpha_{NT}$ w.p.a.1 by Assumption \ref{3.ass3}. 

Now that we are to replace the constrained parameter space $\mathcal{L}$ with $\mathbb{R}^{N\times T}$, $\hat{\Delta}_{L}(u)$ may no longer lie in $\mathcal{D}$ w.p.a.1. We need to constrain $\hat{\Delta}_{L}$ in a  different set. For any $\Delta_{L}\in\mathbb{R}^{N\times T}$, let $\mathcal{P}_{\Omega}\Delta_{L}$ be an $N\times T$ matrix whose $(i,t)$-th element is $\mathbbm{1}(|\Delta_{L,it}|\leq 2\alpha_{NT})\cdot\Delta_{L,it}$. Let $\mathcal{P}_{\Omega^{\perp}}\Delta_{L}\coloneqq \Delta-\mathcal{P}_{\Omega}\Delta_{L}$. By construction, $||\Delta_{L}||_{F}^{2}=||\mathcal{P}_{\Omega}\Delta_{L}||_{F}^{2}+||\mathcal{P}_{\Omega^{\perp}}\Delta_{L}||_{F}^{2}$. Let $||\mathcal{P}_{\Omega}\Delta_{L}||_{F}^{2}=C_{sm}||\Delta_{L}||_{F}^{2}$ where $C_{sm}$ is in $[0,1]$ and may depend on $N$ and $T$. Note that $\mathcal{D}$ is equivalent to the set of matrices whose $C_{sm}$ equals 1. Quantity \eqref{eq.display} is equal to
\begin{align}
&\sum_{i,t}\int_{0}^{\Delta_{L,it}}\left(F_{V_{it}(u)|W}(s)-F_{V_{it}(u)|W}(0)\right)ds\notag\\
=&\sum_{\{i,t:|\Delta_{L,it}|\leq 2\alpha_{NT}\}}\int_{0}^{\Delta_{L,it}}\left(F_{V_{it}(u)|W}(s)-F_{V_{it}(u)|W}(0)\right)ds\notag\\
&+\sum_{\{i,t:|\Delta_{L,it}|> 2\alpha_{NT}\}}\int_{0}^{\Delta_{L,it}}\left(F_{V_{it}(u)|W}(s)-F_{V_{it}(u)|W}(0)\right)ds\notag
\end{align}
with probability one. For the first sum on the right side, we can lower bound it by $C_{min}||\mathcal{P}_{\Omega}\Delta_{L}||_{F}^{2}/\alpha_{NT}^{2}$ using the same argument as in Section \ref{secStrategy} (the proof is similar to that of Lemma \ref{lowerbound} and is thus omitted). For the second term, now that $\Delta_{L,it}$ can be unbounded, the conditional density $f_{V_{it}(u)|W}$ may be arbitrarily close to zero. Hence, it can only be lower bounded by $0$. 
Yet as long as $||\mathcal{P}_{\Omega}\Delta_{L}||_{F}^{2}$ is of a nonnegligible proportion of $||\Delta_{L}||_{F}^{2}$, we can still lower bound \eqref{eq.display} be $||\Delta||_{F}^{2}$ multiplied by some constants.

Formally, assume there exists a universal constant $C_{sm}>0$ such that for all $u\in\mathcal{U}$, we have $\hat{\Delta}_{L}(u)\in\mathcal{D}^{(2)}$ where $\mathcal{D}^{(2)}$ is the following cone:
\begin{equation}\label{eq.secondkey}
\mathcal{D}^{(2)}\coloneqq\left\lbrace\Delta_{L}\in\mathbb{R}^{N\times T}:||\mathcal{P}_{\Omega}\Delta_{L}||_{F}^{2}\geq C_{sm} ||\Delta_{L}||_{F}^{2}\right\rbrace,
\end{equation}
We can then restrict our analysis within $\mathcal{D}^{(2)}$ and lower bound \eqref{eq.display} for all $\Delta_{L}\in\mathcal{D}^{(2)}$ by
\[ C_{min}||\mathcal{P}_{\Omega}\Delta_{L}||_{F}^{2}/\alpha_{NT}^{2}+0\geq C_{sm}C_{min}||\Delta_{L}||_{F}^{2}/\alpha_{NT}^{2}.\] Then we can obtain an error bound on the estimator which has the same order as that in Theorem \ref{3.thm1} since the quadratic lower bound has the same order.

Both $\mathcal{D}$ (adopted in the main text) and $\mathcal{D}^{(2)}$ here limit the spikiness of the matrices $\Delta_{L}$s in them. Set $\mathcal{D}$ restricts the \textit{magnitude} of the large elements in $\Delta_{L}\in\mathcal{D}$. In contrast, by definition \eqref{eq.secondkey}, set $\mathcal{D}^{(2)}$ restricts the \textit{number} of large elements in $\Delta_{L}\in\mathcal{D}^{(2)}$. For instance, on the sphere $||\Delta_{L}||_{F}^{2}=NT\gamma^{2}$ where $\gamma$ is the same as in Section \ref{sec.preview}, definition \eqref{eq.secondkey} allows elements in $\Delta_{L}$ to be as large as $\sqrt{NT}\gamma$, greater than $2\alpha_{NT}$ if $\alpha_{NT}=O(\log(NT))$, but the number of such large elements is at most $O(1)$ .

When there are covariates, complications arise due to restricted strong convexity and $C_{sm}$ not only needs to be bounded away from zero but also needs to be sufficiently large. A sufficient condition is that $C_{sm}\to 1$ as $N$ and $T$ grow to infinity. Specifically, we have the following theorem.
\begin{appxthm}\label{appx.thmmain}  Let $\lambda$ be the same as in Lemma \ref{3.lem1}.
Under Assumptions \ref{3.ass1} to \ref{3.ass4} and the condition in Lemma \ref{3.lem3}, if w.p.a.1, $\hat{\Delta}_{L}(u)\in\mathcal{D}^{(2)}$ defined in equation \eqref{eq.secondkey} with $C_{sm}\to 1$ as $N$ and $T$ grow to infinity, then for $C_{error,2}=2.5C_{error}$ where $C_{error}$ is the constant in Theorem \ref{3.thm1}, the following estimator 
\begin{equation*}
(\hat{\beta}(u),\hat{L}(u))=\arg\min_{\beta\in\mathbb{R}^{p},L\in\mathbb{R}^{N\times T}}\frac{1}{NT}\bm{\rho}_{u}(Y-\sum_{j=1}^{p}X_{j}\beta_{j}-L)+\lambda||L||_{*}
\end{equation*}
satisfies 
\small\begin{align*}
\sup_{u\in\mathcal{U}}\ \ ||\hat{\beta}(u)-\beta_{0}(u)||_{F}^{2}+\frac{1}{NT}||\hat{L}(u)-L_{0}(u)||_{F}^{2}
\leq& C_{error,2}^{2} \alpha_{NT}^{4}\log(NT)\left(\frac{p\log(pNT)}{NT}\lor \frac{\bar{r}}{N\land T}\right)\ w.p.a.1.
\end{align*}\normalsize
\end{appxthm}
\normalsize
\begin{proof}
See Appendix \ref{app.techsecond}.
\end{proof}
\begin{appxrem} The difference between the estimator defined in Theorem \ref{appx.thmmain} and the one defined by equation \eqref{3.eq3} lies in the parameter space of $L$.
\end{appxrem}
\subsection{Comparison of Different Approaches to \ref{step1}}\label{app.technical}
In this section, we compare the assumptions needed to obtain a quadratic lower bound for the purpose of \ref{step1} in our paper with those in \cite{ando2019quantile}, \cite{belloni2019high} and \cite{chen2019quantile}.  To highlight the differences, we still consider the case where there are no covariates. Also, since these mentioned papers all focus on consistency pointwise in $u\in\mathcal{U}$, in the following discussion we also drop the requirements on uniformity in $u$ in our assumptions.

First, let us summarize the assumptions needed in our two approaches to a quadratic lower bound. Recall that without the covariates, the set $\mathcal{D}$ in the main text is defined as $\{\Delta_{L}\in\mathbb{R}^{N\times T}:||\Delta_{L}||_{\infty}\leq 2\alpha_{NT}\}$.
\begin{itemize}
\item Approach 1 (adopted in the main text).
\begin{itemize}
 \item On the conditional density of $V_{it}(u)$: Assumption \ref{3.ass2}.
\item On the \textit{magnitude} of large elements in $\hat{\Delta}_{L}(u)$: Assumption \ref{3.ass3} and the constraint in the estimator's definition \eqref{3.eq3}. They imply $\hat{\Delta}_{L}(u)\in\mathcal{D}$ for all $u\in\mathcal{U}$ w.p.a.1. 
\end{itemize}
\item Approach 2 (introduced in Appendix \ref{app.second}).
\begin{itemize}
 \item On the conditional density of $V_{it}(u)$: Assumption \ref{3.ass2}.
\item On the \textit{number} of large elements in $\hat{\Delta}_{L}(u)$: Assumption \ref{3.ass3} and $\hat{\Delta}_{L}(u)\in\mathcal{D}^{(2)}$ for all $u\in\mathcal{U}$ w.p.a.1 with $C_{sm}\to 1$. 
\end{itemize}
\end{itemize}

\cite{ando2019quantile} and \cite{chen2019quantile} impose stronger assumptions on the conditional density $f_{V_{it}(u)|W}$. They both assume that the conditional density function is continuous and for \textit{any} compact set $S$, there exists an $S$-dependent constant $\underline{f}_{S}$ such that the density $f_{V_{it}(u)|W}(s)\geq \underline{f}_{S}>0$ for all $s\in S$ and all $i$ and $t$. Note that this assumption implies our Assumption \ref{3.ass2} by choosing $S=[-\delta,\delta]$ for any $\delta>0$. This stronger assumption can help obtain a quadratic lower bound for our purpose by a simpler argument if $||L_{0}(u)||\leq \alpha_{NT}$ and $||L||_{\infty}\leq \alpha_{NT}$ are still imposed\footnote{Indeed, these two papers assume elements in $L_{0}(u)$ lie in a fixed compact space, i.e. $\alpha_{NT}$ is fixed, not $(N,T)$-dependent.}. To see this, by $|\Delta_{L,it}|\leq 2\alpha$, their assumption implies that there exists a constant $\underline{f}_{\alpha_{NT}}>0$ such that \eqref{eq.display} is lower bounded by $\underline{f}_{\alpha_{NT}}\sum_{i,t}\Delta_{L,it}^{2}/2$ by directly applying first-order Taylor expansion. Similar to our Approach 1, the lower bound also depends on $\alpha_{NT}$ via $\underline{f}_{\alpha_{NT}}$.


Now let us turn to \cite{belloni2019high}. Their approach is more similar to our Approach 2 because they also restrict the \textit{number} of large elements in $\hat{\Delta}_{L}(u)$. Again, since they only focus on pointwise consistency, we compare our related assumptions with theirs by dropping the required uniformity in $u$. Like our approaches, their assumptions to achieve \ref{step1} also consist of two parts:

First, on the conditional density of $V_{it}(u)$, their Assumption 1 (ii) requires that for all $i$ and $t$, the conditional density $f_{V_{it}(v)|W}(v)$ is bounded away from $0$ at $v=0$ by $\underline{f}$ and bounded from above uniformly in $v$ and in the realization of $W$. Meanwhile, the derivative of the conditional density function $\partial f_{V_{it}(u)|W}(v)/\partial v$ is assumed to be continuous and bounded in absolute value by $\bar{f}'$ uniformly in $v$, $i$, $t$ and in the realization of $W$ as well. These two requirements are stronger than our Assumption \ref{3.ass2}; noting that uniform boundedness of both a set of functions and of their derivatives implies equicontinuity, our Assumption \ref{3.ass2} holds under their Assumption 1(ii).

Second, on the \textit{number} of large elements in $\hat{\Delta}_{L}(u)$, their Assumption 3 and equation (25) essentially require that $\hat{\Delta}_{L}(u)\in \mathcal{D}^{(3)}$ where 
\begin{equation}\label{eq.belloni}
\mathcal{D}^{(3)}\in\left\lbrace\Delta_{L}\in\mathbb{R}^{N\times T}:\frac{\underline{f}}{2}||\Delta_{L}||_{F}^{2}-\frac{\bar{f}'}{3}\sum_{i,t}|\Delta_{L,it}|^{3}\geq 0\right\rbrace
\end{equation}
and the constants $\underline{f}$ and $\bar{f}'$ are introduced in the previous paragraph. By the inequality in \eqref{eq.belloni} and by their assumption on the conditional density, they lower bound  \eqref{eq.display} for $\Delta_{L}\in\mathcal{D}^{(3)}$ by second-order Taylor expansion:
\begin{align*}
\sum_{i,t}\int_{0}^{\Delta_{L,it}}\left(F_{V_{it}(u)|W}(s)-F_{V_{it}(u)|W}(0)\right)ds
\geq \frac{\underline{f}}{4}||\Delta||_{F}^{2}+\left (\frac{\underline{f}}{4}||\Delta||_{F}^{2}-\frac{\bar{f}'}{6}\sum_{i,t}|\Delta_{L,it}|^{3}\right )\geq \frac{\underline{f}}{4}||\Delta||_{F}^{2}
\end{align*}

The set $\mathcal{D}^{(3)}$ serves a similar purpose as $\mathcal{D}^{(2)}$ in our Approach 2. Both restrict the number of large elements in the matrices in these sets. Yet the condition on $\mathcal{D}^{(3)}$ is more restrictive in the sense that large elements allowed in $\mathcal{D}^{(3)}$ are fewer than $\mathcal{D}^{(2)}$. To see this, suppose $||\Delta_{L}||_{F}$ has order $\nu_{NT}$ and let $\alpha_{NT}\to\infty $ and $\alpha_{NT}=o(\nu_{NT})$. For large element $\Delta_{L,it}$ of order $\delta_{NT}\geq 2\alpha_{NT}$, in any matrix in $\mathcal{D}^{(2)}$, there can be as many as $o(\nu_{NT}^{2}/\delta_{NT}^{2})$ of such elements while $C_{sm}\to 1$ still holds. But in any matrix in $\mathcal{D}^{(3)}$, there can be only $O(\nu_{NT}^{2}/\delta_{NT}^{3})$ of them.

Comparing $\mathcal{D}$, $\mathcal{D}^{(2)}$ and $\mathcal{D}^{(3)}$, note that $\hat{\Delta}_{L}(u)\in\mathcal{D}$ for all $u\in\mathcal{U}$ w.p.a.1 can be guaranteed under primitive conditions: Assumption \ref{3.ass3} and the constraint in the definition of the estimator \eqref{3.eq3}. However, $\hat{\Delta}_{L}(u)\in\mathcal{D}^{(2)}$ in our Approach 2 and $\hat{\Delta}_{L}(u)\in\mathcal{D}^{(3)}$ in \cite{belloni2019high} are high level conditions. After the estimator is obtained, one can verify whether any of these conditions is met as they do not involve the true parameter $L_{0}(u)$.
  
To sum up, to obtain a quadratic lower bound in \ref{step1}, we need to i) make assumptions on the conditional density $f_{V_{it}(u)|W}$ and ii) to either restrict the \textit{magnitude} of large elements in $\hat{\Delta}_{L}(u)$  (Approach 1) or the \textit{number} of them (Approach 2). Our assumption on the conditional density seems to be the weakest in the discussed literature. For large elements in $\hat{\Delta}_{L}(u)$, our Approach 1 is under more primitive conditions while the restriction $\mathcal{D}^{(2)}$ in our Approach 2 is milder than the restriction $\mathcal{D}^{(3)}$ in \cite{belloni2019high}. On the other hand, in both of our two approaches, we need Assumption \ref{3.ass3}, while the approach in \cite{belloni2019high} is free of it. Finally, since they focus on high-dimensional regressors, some of the discussed relaxations in our approaches may not apply there. We view our weaker assumptions as the advantages gained by exploiting low dimensionality of the regressors, and all three approaches are complementary.

\section{Proofs}\label{app.proof}
\subsection{Proofs of the Lemmas in Section \ref{sec.preview} }\label{app.sec3}
\subsubsection{Proof of Lemma \ref{3.lem1}}
Recall that $V(u)\coloneqq Y-q_{Y|W}(u)$. Let $\nabla\bm{\rho}_{u}(V(u))$ be an $N\times T$ subgradient matrix of $\bm{\rho}_{u}(\cdot)$ evaluated at $V(u)$. With probability one, the $(i,t)$-th element of $\nabla\bm{\rho}_{u}(V(u))$ is
\begin{equation*}
\left(\nabla\bm{\rho}_{u}\left(V(u)\right)\right)_{it}=u\mathbbm{1}(V_{it}(u)>0)+(u-1)\mathbbm{1}(V_{it}(u)<0).
\end{equation*}
These elements are bounded and independent with mean $0$ conditional on $W$ by Assumption \ref{3.ass3} and by the definition of $V(u)$ \footnote{Conditional mean zero is obtained by noting that $\Pr(V_{it}(u)<0|W)=u$ almost surely by definition.}. We introduce the following lemma for $\nabla\bm{\rho}_{u}(V(u))$. 
The proof is in Appendix \ref{app.techlem}.
\begin{appxlem}\label{3.lemA1}
Under Assumption \ref{3.ass1}, there exists a universal constant $C_{op}>2$ such that the following inequalities hold w.p.a.1:
\begin{align}
\sup_{u\in\mathcal{U}}\max_{1\leq j\leq p}\big|\big\langle \nabla\bm{\rho}_{u}(V(u)),X_{j}\big\rangle\big|&\leq \sqrt{2C_{X}NT\log(pNT)},\label{3.eqA1}\\
\sup_{u\in\mathcal{U}}|| \nabla\bm{\rho}_{u}(V(u))||&\leq C_{op}\sqrt{N\lor T}\label{3.eqA2}
\end{align}
where $C_{X}$ is defined in Assumption \ref{3.ass1}.
\end{appxlem}

In what follows, the derivation is under the event that inequalities \eqref{3.eqA1} and \eqref{3.eqA2} hold. Since $||L_{0}(u)||_{\infty}\leq \alpha$ for all $u\in\mathcal{U}$, $L_{0}(u)$ is a feasible solution to the minimization problem \eqref{3.eq3}. Then by the definition of $(\hat{\beta}(u),\hat{L}(u))$, the following inequality holds with probability one:\small
\begin{align}
\sup_{u\in\mathcal{U}}\left(\frac{1}{NT}\left[\bm{\rho}_{u} \left(V(u)-\sum_{j=1}^{p} X_{j} \hat{\Delta}_{\beta,j}(u)-\hat{\Delta}_{L}(u)\right)-\bm{\rho}_{u} (V(u))\right]+\lambda\left( ||\hat{L}(u)||_{*}- ||L_{0}(u)||_{*}\right)\right)\leq 0\label{eq.A3}
\end{align}\normalsize
where $\hat{\Delta}_{\beta,j}(u)\coloneqq \hat{\beta}_{j}(u)-\beta_{0,j}(u)$ and $\hat{\Delta}_{L}(u)\coloneqq \hat{L}(u)-L_{0}(u)$.

Let us first consider $ \left[\bm{\rho}_{u} \big(V(u)-\sum_{j=1}^{p} X_{j} \hat{\Delta}_{\beta,j}(u)-\hat{\Delta}_{L}(u)\big)-\bm{\rho}_{u} (V(u))\right]/NT$. With probability one,
\begin{align}
&\frac{1}{NT}\left[\bm{\rho}_{u} \big(V(u)-\sum_{j=1}^{p} X_{j} \hat{\Delta}_{\beta,j}(u)-\hat{\Delta}_{L}(u)\big)-\bm{\rho}_{u} (V(u))\right]\notag\\
\geq& -\frac{1}{NT}\big|\big\langle \nabla\bm{\rho}_{u}(V(u)),\sum_{j=1}^{p} X_{j} \hat{\Delta}_{\beta,j}(u)+\hat{\Delta}_{L}(u)\big\rangle\big|\notag\\
\geq&-\frac{1}{NT}||\hat{\Delta}_{\beta}(u)||_{1}\max_{1\leq j\leq p}\big|\big\langle \nabla\bm{\rho}_{u}(V(u)), X_{j}\big\rangle\big|-\frac{1}{NT}||\nabla\bm{\rho}_{u}(V(u))||\cdot ||\hat{\Delta}_{L}(u)||_{*}\notag\\
\geq &-\sqrt{\frac{2C_{X}\log(pNT)}{NT}}||\hat{\Delta}_{\beta}(u)||_{1}-\frac{C_{op}\sqrt{N\lor T}}{NT} ||\hat{\Delta}_{L}(u)||_{*}\notag\\
\geq & -\sqrt{\frac{2C_{X}p\log(pNT)}{NT}}||\hat{\Delta}_{\beta}(u)||_{F}-\frac{C_{op}\sqrt{N\lor T}}{NT} ||\hat{\Delta}_{L}(u)||_{*}\label{eq.A4}
\end{align}
The first inequality is by the definition of subgradient. The first term in the second inequality is elementary. The second term is from Lemma 3.2 in \cite{candes2009exact} which says for any two matrices $A$ and $B$ of the same size, $|\langle A,B\rangle|\leq ||A||\cdot||B||_{*}$. The penultimate inequality is by inequalities \eqref{3.eqA1} and \eqref{3.eqA2} in Lemma \ref{3.lemA1}.

Next, consider $\lambda\big( ||\hat{L}(u)||_{*}- ||L_{0}(u)||_{*}\big)$. Let $\mathcal{P}_{\Phi(u)^{\perp}}$ be the orthogonal projection onto the orthogonal complement of $\Phi(u)$. By construction, $\mathcal{P}_{\Phi(u)^{\perp}}L_{0}(u)=0$. Moreover, for any $N\times T$ matrix $M$, $||\mathcal{P}_{\Phi(u)}M+\mathcal{P}_{\Phi(u)^{\perp}}M||_{*}=||\mathcal{P}_{\Phi(u)}M||_{*}+||\mathcal{P}_{\Phi(u)^{\perp}}M||_{*}$ since $\mathcal{P}_{\Phi(u)}M$ and $\mathcal{P}_{\Phi(u)^{\perp}}M$ have orthogonal singular vectors to each other. Hence, by $\hat{L}(u)=L_{0}(u)+\hat{\Delta}_{L}(u)$, with probability one,
\begin{align}
||\hat{L}(u)||_{*}- ||L_{0}(u)||_{*}= & ||\mathcal{P}_{\Phi(u)}L_{0}(u)+\mathcal{P}_{\Phi(u)}\hat{\Delta}_{L}(u)||_{*}+||\mathcal{P}_{\Phi(u)^{\perp}}\hat{\Delta}_{L}(u)||_{*}- ||\mathcal{P}_{\Phi(u)}L_{0}(u)||_{*}\notag\\
\geq &||\mathcal{P}_{\Phi(u)^{\perp}}\hat{\Delta}_{L}(u)||_{*}-||\mathcal{P}_{\Phi(u)}\hat{\Delta}_{L}(u)||_{*}\label{eq.A5}
\end{align}

Combining equations \eqref{eq.A3}, \eqref{eq.A4} and \eqref{eq.A5}, we have shown that\small
\begin{align}
\sup_{u\in\mathcal{U}}\Bigg(\left(\lambda-\frac{C_{op}\sqrt{N\lor T}}{NT}\right)||\mathcal{P}_{\Phi(u)^{\perp}}\hat{\Delta}_{L}(u)||_{*}
- & \sqrt{\frac{2C_{X}p\log(pNT)}{NT}}||\hat{\Delta}_{\beta}(u)||_{F}\notag\\
-&\left(\lambda+\frac{C_{op}\sqrt{N\lor T}}{NT}\right)||\mathcal{P}_{\Phi(u)}\hat{\Delta}_{L}(u)||_{*}\Bigg)
\leq 0\label{eq.A51}
\end{align}\normalsize
holds with probability one under the event that equations \eqref{3.eqA1} and \eqref{3.eqA2} hold. Let $C_{Cone}=\sqrt{2C_{X}}/C_{op}$. Then by the choice of $\lambda$ in Lemma \ref{3.lem1}, we have the following:
\begin{align*}
&\sup_{u\in\mathcal{U}}\left(||\hat{\Delta}_{L}(u)||_{*}-4||\mathcal{P}_{\Phi(u)}\hat{\Delta}_{L}(u)||_{*}-C_{Cone}\sqrt{p(N\land T)\log(pNT)}||\hat{\Delta}_{\beta}(u)||_{F}\right)\\
= & \sup_{u\in\mathcal{U}}\left(||\mathcal{P}_{\Phi(u)^{\perp}}\hat{\Delta}_{L}(u)||_{*}-3||\mathcal{P}_{\Phi(u)}\hat{\Delta}_{L}(u)||_{*}-C_{Cone}\sqrt{p(N\land T)\log(pNT)}||\hat{\Delta}_{\beta}(u)||_{F}\right)\leq  0
\end{align*}
where the inequality is by equation \eqref{eq.A51}. Hence, inequality \eqref{coneineq} holds w.p.a.1 by noting that inequalities \eqref{3.eqA1} and \eqref{3.eqA2} hold w.p.a.1 by Lemma \ref{3.lemA1}.\qed

\subsubsection{Proof of Lemma \ref{lowerbound}}

To prove the lemma, we need the following result which helps to handle the high-dimensional $\Delta_{L}$. Its proof is in Appendix \ref{app.techlem}.
\begin{appxlem}\label{3.lemA3}
For all $w_{1},w_{2}\in\mathbb{R}$ and all $ \kappa\in(0,1]$,
\begin{equation*}
\int_{0}^{w_{2}}\big(\mathbbm{1}(w_{1}\leq z)-\mathbbm{1}(w_{1}\leq 0)\big)dz\geq \int_{0}^{\kappa w_{2}}\big(\mathbbm{1}(w_{1}\leq z)-\mathbbm{1}(w_{1}\leq 0)\big)dz\geq 0
\end{equation*}
\end{appxlem}

By Knight's identity \citep{knight1998limiting}, for any two scalars $w_{1}$ and $w_{2}$,
\begin{equation*}
\rho_{u}(w_{1}-w_{2})-\rho_{u}(w_{1})=-w_{2}(u-\mathbbm{1}(w_{1}\leq 0))+\int_{0}^{w_{2}}(\mathbbm{1}(w_{1}\leq s)-\mathbbm{1}(w_{1}\leq 0))ds
\end{equation*} 
Let $w_{1}=V_{it}(u)$ and $w_{2}=X_{it}'\Delta_{\beta}+\Delta_{L,it}$ where $\Delta_{\beta}$ and $\Delta_{L}$ are arbitrary fixed $p\times 1$ vector and $N\times T$ matrix, then by construction $\mathbb{E}(-w_{2}(u-\mathbbm{1}(w_{1}\leq 0))|W)=0$. Let $\kappa=1/(3\alpha_{NT})$. By $\alpha_{NT}\geq 1$, $\kappa \in (0,1)$. By $||\Delta_{L}||_{\infty}\leq 2\alpha_{NT}$, by $||\Delta_{\beta}||_{F}\leq \gamma$ and under the event that $\sum_{j=1}^{p}||X_{j}||_{\infty}\gamma\leq \alpha_{NT}$, we have $|X_{it}'\Delta_{\beta}+\Delta_{L,it}|\leq 3\alpha_{NT}$ for all $i$ and $t$, implying that $|\kappa\cdot\left(X_{it}'\Delta_{\beta}+\Delta_{L,it}\right)|\leq 1$ for all $i$ and $t$. Therefore, for all $i$ and $t$, $\kappa\cdot\left(X_{it}'\Delta_{\beta}+\Delta_{L,it}\right)(1\land \delta)\in [-\delta,\delta]$ where $\delta>0$ is defined in Assumption \ref{3.ass2}. Then by Assumption \ref{3.ass2}, Lemma \ref{3.lemA3} and the mean value theorem, the following holds for all $i,t$, and $u\in\mathcal{U}$ almost surely:
\begin{align}
&\mathbb{E}\left(\int_{0}^{X_{it}'\Delta_{\beta}+\Delta_{L,it}}\big(\mathbbm{1}(V_{it}(u)\leq s)-\mathbbm{1}(V_{it}(u)\leq 0)\big)ds\Big|W\right)\notag\\
\geq &\mathbb{E}\left(\int_{0}^{\kappa (X_{it}'\Delta_{\beta}+\Delta_{L,it})(1\land \delta)}\big(\mathbbm{1}(V_{it}(u)\leq s)-\mathbbm{1}(V_{it}(u)\leq 0)\big)ds\Big|W\right)\notag\\
= &\int_{0}^{\kappa(X_{it}'\Delta_{\beta}+\Delta_{L,it})(1\land\delta)}\big(F_{V_{it}(u)|W}(s)-F_{V_{it}(u)|W}(0))\big)ds\notag\\
=&\int_{0}^{\kappa(X_{it}'\Delta_{\beta}+\Delta_{L,it})(1\land \delta)}sf_{V_{it}(u)|W}(\tilde{s}(s))ds\notag\\
\geq & \frac{\kappa^{2}(1\land\delta)^{2}(X_{it}'\Delta_{\beta}+\Delta_{L,it})^{2}\underline{f}}{2}\notag\\
=& \frac{(1\land \delta)^{2}\underline{f}(X_{it}'\Delta_{\beta}+\Delta_{L,it})^{2}}{18\alpha_{NT}^{2}} \notag
\end{align}
where $\underline{f}$ is defined in Assumption \ref{3.ass2} and $\tilde{s}(s)\in [0,1]$ is the mean value.
The desired result is obtained by letting $C_{min}=(1\land \delta)^{2}\underline{f}/18$.
\qed

\subsubsection{Proof of Lemma \ref{3.lem3}}
The main argument of the proof follows the proof of Lemma 5 in \cite{belloni2011l1}. The major difference is that we need to handle the matrix component $\Delta_{L}$.

Let 
\[
\mathcal{A}(\gamma)\coloneqq \sup_{\substack{u\in\mathcal{U}\\(\Delta_{\beta},\Delta_{L})\in\mathcal{R}_{u}\\||\Delta_{\beta}||_{F}^{2}+\frac{1}{NT}||\Delta_{L}||_{F}^{2}\leq\gamma^{2}}} \big|\mathbb{G}_{u}\big(\rho_{u} \big(V_{it}(u)-X_{it}' \Delta_{\beta}-\Delta_{L,it}\big)-\rho_{u} (V_{it}(u))\big)\big|.
\]
where for generic random variables $(Z_{it})_{i,t}$ and a function $f$, recall that $\mathbb{G}_{u}(f(Z_{it}))\coloneqq \sum_{i,t}[f(Z_{it})-\mathbb{E}(f(Z_{it})|W)]/\sqrt{NT}$. Denote its symmetrized version by $\mathbb{G}^{0}(f(Z_{it}))\coloneqq (\sum_{i,t}f(Z_{it})\varepsilon_{it})/\sqrt{NT}$ where $(\varepsilon_{it})_{i,t}$ is a Rademacher sequence independent of $(\{V(u)\}_{u\in (0,1)},W)$. Let $\Omega_{1}$ be the event that $\max_{1\leq j\leq p}||X_{j}||_{F}^{2}\leq C_{X}NT$ where $C_{X}$ is as in Assumption \ref{3.ass1}. Since for any $s>0$,
\begin{align}
\mathbb{P}(\mathcal{A}(\gamma)\geq s)\leq \mathbb{P}(\mathcal{A}(\gamma)\geq s|\Omega_{1})\mathbb{P}(\Omega_{1})+\mathbb{P}(\Omega_{1}^{c})\label{eq.At}
\end{align}
and $\mathbb{P}(\Omega_{1}^{c})\to 0$ under Assumption \ref{3.ass1}, we only need to show that for some $C_{sup}>0$ and $s=C_{sup}\log(NT)\left(\sqrt{p\log(pNT)}\lor\sqrt{\bar{r}(N\lor T)}\right)\gamma$, the conditional probability $\mathbb{P}(\mathcal{A}(\gamma)\geq s|\Omega_{1})$ converges to zero.

Similar to \cite{belloni2011l1} and \cite{chao2015factorisable}, for any fixed $\Delta_{\beta}$ and $\Delta_{L}$ with $||\Delta_{\beta}||_{F}^{2}+\frac{1}{NT}||\Delta_{L}||_{F}^{2}\leq\gamma^{2}$ and any $u\in\mathcal{U}$, we have the following bound on the conditional variance of the process by noting that the check function is a contraction:
\begin{align}
&\textrm{Var}\left(\mathbb{G}_{u}\left(\rho_{u} \left(V_{it}(u)-X_{it}' \Delta_{\beta}-\Delta_{L,it}\right)-\rho_{u} (V_{it}(u))\right)|W\right)\notag\\
\leq & \frac{1}{NT}||\sum_{j=1}^{p}X_{j} \Delta_{\beta,j}+\Delta_{L}||_{F}^{2}\notag\\
\leq &  \frac{2}{NT}||\sum_{j=1}^{p}X_{j} \Delta_{\beta,j}||_{F}^{2}+\frac{2}{NT}||\Delta_{L}||_{F}^{2}\notag\\
\leq & \frac{2p}{NT}\sum_{j=1}^{p}(||X_{j}||_{F}^{2}\Delta_{\beta,j}^{2})+\frac{2}{NT}||\Delta_{L}||_{F}^{2}\notag\\
\leq & 2\left(\frac{p\max_{1\leq j\leq p}||X_{j}||_{F}^{2}}{NT}\lor 1\right)\big(||\Delta_{\beta}||_{F}^{2}+\frac{1}{NT}||\Delta_{L}||_{F}^{2}\big)\notag\\
\leq &2\left(\frac{p\max_{1\leq j\leq p}||X_{j}||_{F}^{2}}{NT}\lor 1\right)\gamma^{2}\label{variance}
\end{align}
Since $\mathbb{E}(\mathbb{G}_{u}\left(\rho_{u} \left(V_{it}(u)-X_{it}' \Delta_{\beta}-\Delta_{L,it}\right)-\rho_{u} (V_{it}(u))\right)|W)=0$ by construction, with inequality \eqref{variance} we can apply the symmetrization lemma for probability, for instance Lemma 2.3.7 in \cite{van1996weak}.
\begin{align}
\mathbb{P}\left(\mathcal{A}(\gamma)>s|\Omega_{1}\right)
=&\mathbb{E}\left[\mathbb{P}\left(\mathcal{A}(\gamma)>s|W\right)|\Omega_{1}\right]\notag\\
\leq&\mathbb{E}\left[ \frac{\mathbb{P}\left(\mathcal{A}^{0}(\gamma)>\frac{s}{4}|W\right)}{1-8\left(pC_{X}\lor 1\right)\gamma^{2}/s^{2}}\Big|\Omega_{1}\right]\notag\\
\leq & 2\mathbb{P}\left(\mathcal{A}^{0}(\gamma)>\frac{s}{4}\big|\Omega_{1}\right)\label{epA}
\end{align}
where $\mathcal{A}^{0}(\gamma)$ is the symmetrized version of $\mathcal{A}(\gamma)$ by replacing $\mathbb{G}_{u}$ with its symmetrized version $\mathbb{G}^{0}$. The equality is by the law of iterated expectation by noting that $\max_{1\leq j\leq p} ||X_{j}||_{F}^{2}$ in the event $\Omega_{1}$ is a function of $W$. The first inequality is by Lemma 2.3.7 in \cite{van1996weak} and Chebyshev's inequality, and by the bound on the conditional variance \eqref{variance} and the definition of $\Omega_{1}$. The last inequality holds because $\gamma\sqrt{\left(pC_{X}\lor 1\right)}/s\to 0$ by the definition of $s$ for any fixed $C_{sup}$. Next we show that $\mathbb{P}(\mathcal{A}^{0}(\gamma)>s/4|\Omega_{1})\to 0$ under our choice of $s$ for some $C_{sup}$.

Consider the random variable $\rho_{u} \left(V_{it}(u)-X_{it}' \Delta_{\beta}-\Delta_{L,it}\right)-\rho_{u} (V_{it}(u))$:
\begin{equation*}
\rho_{u} \left(V_{it}(u)-X_{it}' \Delta_{\beta}-\Delta_{L,it}\right)-\rho_{u} (V_{it}(u))
=-u\cdot(X_{it}' \Delta_{\beta}+\Delta_{L,it})+\delta_{it}\left(X_{it}' \Delta_{\beta}+\Delta_{L,it},u\right)
\end{equation*}
where $\delta_{it}\left(X_{it}' \Delta_{\beta}+\Delta_{L,it},u\right)=\left(V_{it}(u)-X_{it}' \Delta_{\beta}-\Delta_{L,it}\right)_{-}-\left(V_{it}(u)\right)_{-}$. Let
\[
\mathcal{B}^{0}_{1}(\gamma)\coloneqq \sup_{\substack{||\Delta_{
\beta}||_{F}^{2}\leq\gamma^{2}}} \big|\mathbb{G}^{0}\big(X_{it}' \Delta_{\beta}\big)\big|,
\]
\[
\mathcal{B}^{0}_{2}(\gamma)\coloneqq \sup_{\substack{(\Delta_{\beta},\Delta_{L})\in\mathcal{R}_{u}\\
||\Delta_{\beta}||_{F}^{2}+\frac{1}{NT}||\Delta_{L}||_{F}^{2}\leq\gamma^{2}}} \big|\mathbb{G}^{0}\big(\Delta_{L,it}\big)\big|,
\]
and
\[
\mathcal{C}^{0}(\gamma)\coloneqq\sup_{\substack{u\in\mathcal{U}\\(\Delta_{\beta},\Delta_{L})\in\mathcal{R}_{u}\\||\Delta_{
\beta}||_{F}^{2}+\frac{1}{NT}||\Delta_{L}||_{F}^{2}\leq\gamma^{2}}}\big|\mathbb{G}^{0}\big(\delta_{it}\big(X_{it}' \Delta_{\beta}+\Delta_{L,it},u\big)\big)\big|,
\]
then $\mathcal{A}^{0}(\gamma)\leq \mathcal{B}^{0}_{1}(\gamma)+\mathcal{B}^{0}_{2}(\gamma)+\mathcal{C}^{0}(\gamma)$ with probability one. Hence,
\begin{align}
&\mathbb{P}\left(\mathcal{A}^{0}(\gamma)\geq \frac{s}{4}\big|\Omega_{1}\right)\notag\\
\leq& \mathbb{P}\left(\mathcal{B}^{0}_{1}(\gamma)+\mathcal{B}^{0}_{2}(\gamma)+\mathcal{C}^{0}(\gamma)\geq \frac{s}{4}\big|\Omega_{1}\right)\notag\\
\leq &3\max\left\lbrace \mathbb{P}\left(\mathcal{B}^{0}_{1}(\gamma)\geq \frac{s}{4}\big|\Omega_{1}\right),\mathbb{P}\left(\mathcal{B}^{0}_{2}(\gamma)\geq \frac{s}{4}\big|\Omega_{1}\right),\mathbb{P}\left(\mathcal{C}^{0}(\gamma)\geq \frac{s}{4}\big|\Omega_{1}\right)\right\rbrace.\label{epA0}
\end{align}
We now derive upper bounds on $\mathcal{B}^{0}_{1}(\gamma)$, $\mathcal{B}^{0}_{2}(\gamma)$ and $\mathcal{C}^{0}(\gamma)$ respectively.

\textbf{Bound on $\mathcal{B}^{0}_{1}(\gamma)$}. The derivation of the bound on $\mathcal{B}^{0}_{1}(\gamma)$ follows \cite{belloni2011l1} closely. We present the proof here for completeness. For some $K_{1}>0$, by Markov's inequality,
\begin{align*}
&\mathbb{P}\left(\mathcal{B}^{0}_{1}(\gamma)>K_{1}\big|W,\Omega_{1}\right)\\
\leq& \min_{\tau\geq 0}e^{-\tau K_{1}}\mathbb{E}\left[\exp\left(\tau \mathcal{B}_{1}^{0}(\gamma)\right)\big|W,\Omega_{1}\right]\\
\leq &\min_{\tau\geq 0}e^{-\tau K_{1}}\mathbb{E}\left[\exp\left(\tau \sup_{||\Delta_{\beta}||_{F}^{2}\leq\gamma^{2}} ||\Delta_{\beta}||_{1}\cdot \max_{1\leq j\leq p}|\mathbb{G}^{0}(X_{j,it})|\right)\big|W,\Omega_{1}\right]\\
\leq &2p\min_{\tau\geq 0}e^{-\tau K_{1}}\max_{1\leq j\leq p}\mathbb{E}\left[\exp\left(\tau \sqrt{p}\gamma\cdot \mathbb{G}^{0}(X_{j,it})\right)\big|W,\Omega_{1}\right]\\
\leq &2p\min_{\tau\geq 0}e^{-\tau K_{1}}\exp\left(\frac{\tau^{2}p\gamma^{2}C_{X}}{2}\right)
\end{align*}
where the third inequality follows from the fact that $\mathbb{E}[\max_{1\leq j\leq p}\exp(|z_{j}|)]\leq 2p\max_{1\leq j\leq p}\mathbb{E}[\exp(z_{j})]$ for a symmetric random variable $z_{j}$ \citep{belloni2011l1}. The last inequality is by an intermediate step in the proof of Hoeffding's inequality (e.g. \cite{van1996weak} p.100) and by $\Omega_{1}$. Hence, by setting $\tau=K_{1}/(p\gamma^{2}C_{X})$ and $K_{1}=\sqrt{2pC_{X}\log(pNT)}\cdot \gamma$, we have
\begin{equation*}
\mathbb{P}\left(\mathcal{B}^{0}_{1}(\gamma)>K_{1}\big|W,\Omega_{1}\right)\leq 2p\exp\left(-\frac{K_{1}^{2}}{2p\gamma^{2}C_{X}}\right)=\frac{2}{NT}\to 0.
\end{equation*}
Therefore,
\begin{equation}\label{epB1}
\mathbb{P}\left(\mathcal{B}_{1}^{0}>K_{1}\big|\Omega_{1}\right)=\mathbb{E}\left[\mathbb{P}\left(\mathcal{B}_{1}^{0}>K_{1}|W,\Omega_{1}\right)\big|\Omega_{1}\right]\leq \mathbb{E}\left(\frac{2}{NT}\right)\to 0
\end{equation}


\textbf{Bound on $\mathcal{B}^{0}_{2}(\gamma)$}. Recall that $\{\varepsilon_{it}\}_{i,t}$ is the Rademacher sequence in the symmetrized process. Let $\bm{\varepsilon}$ be the $N\times T$ matrix $(\varepsilon_{it})_{i,t}$. Then with probability one,
\begin{align*}
\mathcal{B}^{0}_{2}(\gamma)=&\frac{1}{\sqrt{NT}}\sup_{\substack{(\Delta_{\beta},\Delta_{L})\in\mathcal{R}_{u}\\||\Delta_{
\beta}||_{F}^{2}+\frac{1}{NT}||\Delta_{L}||_{F}^{2}\leq\gamma^{2}}}|\sum_{i,t}\varepsilon_{it}\Delta_{L,it}|\\
=&\frac{1}{\sqrt{NT}}\sup_{\substack{(\Delta_{\beta},\Delta_{L})\in\mathcal{R}_{u}\\||\Delta_{
\beta}||_{F}^{2}+\frac{1}{NT}||\Delta_{L}||_{F}^{2}\leq\gamma^{2}}}|\langle\bm{\varepsilon},\Delta_{L}\rangle|\\
\leq &\frac{1}{\sqrt{NT}}||\bm{\varepsilon}||\cdot \sup_{\substack{(\Delta_{\beta},\Delta_{L})\in\mathcal{R}_{u}\\||\Delta_{
\beta}||_{F}^{2}+\frac{1}{NT}||\Delta_{L}||_{F}^{2}\leq\gamma^{2}}}||\Delta_{L}||_{*}\\
\leq &\frac{1}{\sqrt{NT}}||\bm{\varepsilon}||\cdot \sup_{||\Delta_{
\beta}||_{F}^{2}+\frac{1}{NT}||\Delta_{L}||_{F}^{2}\leq\gamma^{2}}\big(C_{Cone}\sqrt{p(N\land T)\log(pNT)}||\Delta_{\beta}||_{F}+4||\mathcal{P}_{\Phi(u)}\Delta_{L}||_{*}\big)\\
\leq &\frac{1}{\sqrt{NT}}||\bm{\varepsilon}||\cdot\sup_{||\Delta_{
\beta}||_{F}^{2}+\frac{1}{NT}||\Delta_{L}||_{F}^{2}\leq\gamma^{2}}\big(C_{Cone}\sqrt{p(N\land T)\log(pNT)}||\Delta_{\beta}||_{F}+4\sqrt{3\bar{r}}||\Delta_{L}||_{F}\big)\\
\leq &\frac{2}{\sqrt{NT}}||\bm{\varepsilon}||\cdot\left(\left(C_{Cone}\sqrt{p(N\land T)\log(pNT)}\right)\lor \left(4\sqrt{3\bar{r}NT}\right)\right)\gamma
\end{align*}
where the second inequality is by the definition of cone $\mathcal{R}_{u}$ and the third inequality is by equation \eqref{equiorder}. Since $\bm{\varepsilon}$ has i.i.d. mean $0$ entries which are uniformly bounded in magnitude by $1$,  there exists a constant $C_{Sp}>0$ such that for $K_{2}\coloneqq C_{Sp}\left(\left(C_{Cone}\sqrt{p\log(pNT)}\right)\lor \left(4\sqrt{N\lor T}\sqrt{3\bar{r}}\right)\right)\gamma$,
\begin{align}
\mathbb{P}\left(\mathcal{B}^{0}_{2}(\gamma)>K_{2}\big|\Omega_{1}\right)
\leq\mathbb{P}\left(||\bm{\varepsilon}||>C_{Sp}\sqrt{N\lor T}/2\right)\to 0 \label{epB2}
\end{align}
where the convergence is by Corollary 2.3.5 in \cite{tao2012topics}. 

\textbf{Bound on $\mathcal{C}^{0}(\gamma)$}. By $\gamma=o(1)$, it is smaller than one for sufficiently large $N$ and $T$. Then let $\mathcal{U}_{l}=\{u_{1},...,u_{l}\}$ be an $\epsilon$-net of $\mathcal{U}$ where $\epsilon=\gamma$ and $\epsilon l\leq 1$. For any $\bar{u}\in \mathcal{U}$, we have the identity 
\begin{align*}
\delta_{it}(X_{it}'\Delta_{\beta}+
\Delta_{L,it},u)
=&\delta_{it}[X_{it}'(\Delta_{\beta}+\beta_{0}(u)-\beta_{0}(\bar{u}))+
\Delta_{L,it}+L_{0,it}(u)-L_{0,it}(\bar{u}),\bar{u}]\\
&-\delta_{it}[X_{it}'(\beta_{0}(u)-\beta_{0}(\bar{u}))+L_{0,it}(u)-L_{0,it}(\bar{u}),\bar{u}].
\end{align*}
Then by the triangle inequality, with probability one we have\small
\begin{align}
\mathcal{C}^{0}(\gamma)\leq&
\sup_{\substack{u\in\mathcal{U},|u-\bar{u}|<\epsilon,\bar{u}\in\mathcal{U}_{l}\\(\Delta_{\beta},\Delta_{L})\in\mathcal{R}_{u}\\||\Delta_{
\beta}||_{F}^{2}+\frac{1}{NT}||\Delta_{L}||_{F}^{2}\leq\gamma^{2}}}\big|\mathbb{G}^{0}\big(\delta_{it}[X_{it}'(\Delta_{\beta}+\beta_{0}(u)-\beta_{0}(\bar{u}))+
\Delta_{L,it}+L_{0,it}(u)-L_{0,it}(\bar{u}),\bar{u}]\big)\big|\notag\\
& +\sup_{\substack{u\in\mathcal{U},|u-\bar{u}|<\epsilon,\bar{u}\in\mathcal{U}_{l}\\(\Delta_{\beta},\Delta_{L})\in\mathcal{R}_{u}\\||\Delta_{
\beta}||_{F}^{2}+\frac{1}{NT}||\Delta_{L}||_{F}^{2}\leq\gamma^{2}}}\big|\mathbb{G}^{0}\big(\delta_{it}[X_{it}'(\beta_{0}(u)-\beta_{0}(\bar{u}))+L_{0,it}(u)-L_{0,it}(\bar{u}),\bar{u}]\big)\big|\label{eq.Cboundpre}
\end{align}\normalsize
We will proceed by treating $\Delta_{\beta}+\beta_{0}(u)-\beta_{0}(\bar{u})$ and $\beta_{0}(u)-\beta_{0}(\bar{u})$ as new $\Delta_{\beta}$s, and $\Delta_{L}+L_{0}(u)-L_{0}(\bar{u})$ and $L_{0}(u)-L_{0}(\bar{u})$ as new $\Delta_{L}$s. However, they may no longer lie in the ball $||\Delta_{
\beta}||_{F}^{2}+\frac{1}{NT}||\Delta_{L}||_{F}^{2}\leq\gamma^{2}$ and in cone $\mathcal{R}_{u}$. So, we need to first expand these two sets:

Let the event $\Omega_{2}\coloneqq\{||L_{0}(u)'-L_{0}(u)||_{F}/\sqrt{NT}\leq \zeta_{2}|u'-u|,\forall u,u'\in\mathcal{U}\}$. By Assumption \ref{3.ass4}, $\mathbb{P}(\Omega_{2})\to 1$. For $\Delta_{\beta}$, by Assumption \ref{3.ass4}, by $\epsilon=\gamma$ and by $||\Delta_{\beta}||_{F}\leq \gamma$, we have $||\Delta_{\beta}+\beta_{0}(u)-\beta_{0}(\bar{u})||_{F}^{2}\leq 2(1+\zeta_{1}^{2})\gamma^{2}$ and $||\beta_{0}(u)-\beta_{0}(\bar{u})||_{F}^{2}\leq \zeta_{1}^{2}\gamma^{2}$ for all $|u-\bar{u}|\leq \epsilon$. Similarly, for $||\Delta_{L}||_{F}\leq\sqrt{NT}\gamma$, under $\Omega_{2}$, $||\Delta_{L}+L_{0}(u)-L_{0}(\bar{u})||_{F}^{2}/NT\leq 2(1+\zeta_{2}^{2})\gamma^{2}$ while $||L_{0}(u)-L_{0}(\bar{u})||_{F}^{2}/NT\leq \zeta_{2}^{2}\gamma^{2}$ for all $|u-\bar{u}|\leq \epsilon$. Therefore, we need to expand the ball to be $||\Delta_{\beta}||_{F}^{2}+||\Delta_{L}||_{F}^{2}/NT\leq 2(1+\zeta_{1}^{2}+\zeta_{2}^{2})\gamma^{2}$. For simplicity, let $C_{\zeta}\coloneqq 1+\zeta_{1}^{2}+\zeta_{2}^{2}$.

Next, let us expand $\mathcal{R}_{u}$. Since
$\text{rank}(L_{0}(u)-L_{0}(\bar{u}))\leq r(u)+r(\bar{u})\leq 2\bar{r}$,
we have $||L_{0}(u)-L_{0}(\bar{u})||_{*}\leq \sqrt{2\bar{r}}||L_{0}(u)-L_{0}(\bar{u})||_{F}\leq \sqrt{2\bar{r}}\zeta_{2}\sqrt{NT}\gamma$ under $\Omega_{2}$ and by $\epsilon=\gamma$ for all $|u-\bar{u}|\leq \epsilon$. Similarly, for $(\Delta_{\beta},\Delta_{L})\in\mathcal{R}_{u}$ and $||\Delta_{\beta}||_{F}^{2}+||\Delta_{L}||_{F}^{2}/NT\leq 2C_{\zeta}\gamma^{2}$, under $\Omega_{2}$, the following holds for all $u\in\mathcal{U}$ and all $|u-\bar{u}|\leq \epsilon$ with probability one,
\begin{align*}
||\Delta_{L}+L_{0}(u)-L_{0}(\bar{u})||_{*}\leq& ||\Delta_{L}||_{*}+||L_{0}(u)-L_{0}(\bar{u})||_{*}\\
\leq & 4||\mathcal{P}_{\Phi(u)}\Delta_{L}(u)||_{*}+C_{Cone}\sqrt{p\log(pNT)(N\land T)}||\Delta_{\beta}||_{F}+\sqrt{2\bar{r}}\zeta_{2}\sqrt{NT}\gamma\\
\leq & (4\sqrt{6C_{\zeta}}+\sqrt{2}\zeta_{2})\sqrt{NT\bar{r}}\gamma+C_{Cone}\sqrt{p\log(pNT)(N\land T)}\sqrt{2C_{\zeta}}\gamma
\end{align*}
where the second inequality follows from the definition of $\mathcal{R}_{u}$ and from $\Omega_{2}$. Let 
\[\bar{\mathcal{R}}=\left\lbrace\Delta_{L}\in\mathbb{R}^{N\times T}:||\Delta_{L}||_{*}\leq  \left[(4\sqrt{6C_{\zeta}}+\sqrt{2}\zeta_{2})\sqrt{NT\bar{r}}+C_{Cone}\sqrt{p\log(pNT)(N\land T)}\sqrt{2C_{\zeta}}\right]\gamma\right\rbrace.
\]
Therefore, in the intersection of the ball $||\Delta_{\beta}||_{F}^{2}+||\Delta_{L}||_{F}^{2}/NT\leq 2C_{\zeta}\gamma^{2}$ and $\mathcal{R}_{u}$ for all $u\in\mathcal{U}$, the matrices $\Delta_{L}$, $(\Delta_{L}+L_{0}(u)-L_{0}(\bar{u}))$ and $(L_{0}(u)-L_{0}(\bar{u}))$ are all in $\bar{\mathcal{R}}$ for all $u\in\mathcal{U}$ and $|u-\bar{u}|\leq\epsilon$ under $\Omega_{2}$. Hence, under $\Omega_{2}$, inequality \eqref{eq.Cboundpre} implies that,

\begin{equation}
\mathcal{C}^{0}(\gamma)\leq  2\cdot\sup_{\substack{\bar{u}\in\mathcal{U}_{l},\Delta_{L}\in\bar{\mathcal{R}}\\||\Delta_{
\beta}||_{F}^{2}+\frac{1}{NT}||\Delta_{L}||_{F}^{2}\leq 2C_{\zeta}\gamma^{2}}}\big|\mathbb{G}^{0}\big(\delta_{it}(X_{it}'\Delta_{\beta}+\Delta_{L,it},\bar{u})\big)\big|\coloneqq 2\mathcal{C}^{1}(\gamma)\label{eq.Cbound}
\end{equation}

Now consider the following events $\Omega_{3}$ and $\Omega_{4}$:
\begin{align*}
\Omega_{3}\coloneqq& \left\lbrace\sup_{||\Delta_{
\beta}||_{F}^{2}\leq 2C_{\zeta}\gamma^{2}}\big|\mathbb{G}^{0}\big(X_{it}'\Delta_{\beta}\big)\big|\leq 2\gamma\sqrt{C_{\zeta}pC_{X}\log(pNT)}\right\rbrace\\
\Omega_{4}\coloneqq& \Bigg\lbrace\sup_{\substack{\Delta_{L}\in\bar{\mathcal{R}}\\||\Delta_{\beta}||_{F}^{2}+\frac{1}{NT}||\Delta_{L}||_{F}^{2}\leq 2C_{\zeta}\gamma^{2}}}\big|\mathbb{G}^{0}\left(\Delta_{L,it}\right)\big|\\
&\leq C_{Sp}\left[\left(C_{Cone}\sqrt{p\log(pNT)}\sqrt{2C_{\zeta}}\right)\lor \left(\sqrt{N\lor T}(4\sqrt{6C_{\zeta}}+\sqrt{2}\zeta_{2})\sqrt{\bar{r}}\right)\right]\gamma\Bigg\rbrace
\end{align*}
Similar to the derivation of the bounds on $\mathcal{B}_{1}^{0}(\gamma)$ and on $\mathcal{B}^{0}_{2}(\gamma)$, we can show that $\mathbb{P}(\Omega_{3}|\Omega_{1}\cap\Omega_{2})\to 1$ and $\mathbb{P}(\Omega_{4}|\Omega_{1}\cap\Omega_{2})\to 1$. Therefore,
\begin{equation}
\mathbb{P}(\Omega_{3}\cap\Omega_{4}|\Omega_{1}\cap\Omega_{2})=\mathbb{P}(\Omega_{3}|\Omega_{1}\cap\Omega_{2})-\mathbb{P}(\Omega_{3}\cap\Omega_{4}^{c}|\Omega_{1}\cap\Omega_{2})\geq \mathbb{P}(\Omega_{3}|\Omega_{1}\cap\Omega_{2})-\mathbb{P}(\Omega_{4}^{c}|\Omega_{1}\cap\Omega_{2})\to 1\label{eq.Omega23}
\end{equation}
Now we can derive the upper bound on $\mathcal{C}^{0}(\gamma)$. For some $K_{3}>0$, by equation \eqref{eq.Cbound},
\begin{align}
&\mathbb{P}\left(\mathcal{C}^{0}(\gamma)\geq 2K_{3}\big|W,\Omega_{2}\right)\notag\\
\leq & \mathbb{P}\left(\mathcal{C}^{1}(\gamma)\geq K_{3}\big|W,\Omega_{2}\right)\notag\\
=&\mathbb{P}\left(\mathcal{C}^{1}(\gamma)\geq K_{3},\Omega_{3}\cap\Omega_{4}\big|W,\Omega_{2}\right)+\mathbb{P}\left(\mathcal{C}^{1}(\gamma)\geq K_{3},(\Omega_{3}\cap\Omega_{4})^{c}\big|W,\Omega_{2}\right)\notag\\
\leq &\mathbb{P}\left(\mathcal{C}^{1}(\gamma)\geq K_{3}\big|\Omega_{2}\cap\Omega_{3}\cap\Omega_{4},W\right)\mathbb{P}(\Omega_{3}\cap\Omega_{4}\big|W,\Omega_{2})+\mathbb{P}\left((\Omega_{3}\cap\Omega_{4})^{c}\big|W,\Omega_{2}\right)\notag\\
\leq &\mathbb{P}(\Omega_{3}\cap\Omega_{4}\big|W,\Omega_{2})e^{-\tau'  K_{3}}\mathbb{E}\left[e^{\tau' \mathcal{C}^{1}(\gamma)}\big|\Omega_{2}\cap\Omega_{3}\cap\Omega_{4},W\right]+\mathbb{P}\left((\Omega_{3}\cap\Omega_{4})^{c}\big|W,\Omega_{2}\right)\label{eq.C0C1}
\end{align}
where the last line is by Markov's inequality for some $\tau'>0$. For $\mathbb{E}\left[e^{\tau \mathcal{C}^{1}(\gamma)}\big|\Omega_{2}\cap\Omega_{3}\cap\Omega_{4},W\right]$, by $l\leq 1/\epsilon=1/\gamma$, we have \small
\begin{align}
&\mathbb{E}\left[e^{\tau \mathcal{C}^{1}(\gamma)}|\Omega_{2}\cap\Omega_{3}\cap\Omega_{4},W\right]\notag\\
\leq & \frac{1}{\gamma}\max_{\bar{u}\in\mathcal{U}_{l}}\mathbb{E}\left[\exp\left(\tau'\sup_{\substack{\Delta_{L}\in\bar{\mathcal{R}}\\||\Delta_{
\beta}||_{F}^{2}+\frac{1}{NT}||\Delta_{L}||_{F}^{2}\leq 2C_{\zeta}\gamma^{2}}}\big|\mathbb{G}^{0}\left(\delta_{it}(X_{it}'\Delta_{\beta}+\Delta_{L,it},\bar{u})\right)\big|\right)\Bigg |\Omega_{2}\cap\Omega_{3}\cap\Omega_{4},W\right]\notag\\
\leq &\frac{1}{\gamma}\mathbb{E}\left[\exp\big(2\tau'\sup_{\substack{\Delta_{L}\in\bar{\mathcal{R}}\\||\Delta_{
\beta}||_{F}^{2}+\frac{1}{NT}||\Delta_{L}||_{F}^{2}\leq 2C_{\zeta}\gamma^{2}}}\big|\mathbb{G}^{0}\left(X_{it}'\Delta_{\beta}+\Delta_{L,it}\big)\big|\right)\Bigg |\Omega_{2}\cap\Omega_{3}\cap\Omega_{4},W\right]\notag\\
\leq&\frac{1}{\gamma}\mathbb{E}\left[\exp\left(2\tau'\sup_{\substack{||\Delta_{
\beta}||_{F}^{2}\\\leq 2C_{\zeta}\gamma^{2}}}\big|\mathbb{G}^{0}\big(X_{it}'\Delta_{\beta}\big)\big|+2\tau'\sup_{\substack{\Delta_{L}\in\bar{\mathcal{R}}\\\frac{1}{NT}||\Delta_{L}||_{F}^{2}\leq 2C_{\zeta}\gamma^{2}}}\big|\mathbb{G}^{0}\big(\Delta_{L,it}\big)\big|\right)\Bigg |\Omega_{2}\cap\Omega_{3}\cap\Omega_{4},W\right]\notag\\
\leq &\frac{1}{\gamma}\exp\left(2\tau'\left[2\sqrt{2C_{\zeta}p\log(pNT)(2C_{X}\lor C_{Sp}^{2}C_{Cone}^{2})}+C_{Sp}\sqrt{N\lor T}(4\sqrt{6C_{\zeta}}+\sqrt{2}\zeta_{2})\sqrt{\bar{r}}\right]\gamma\right)\label{eq.C0exp}
\end{align}\normalsize
where the second inequality is by Theorem 4.12 of \cite{ledoux1991ineq} and by contractivity of $\delta_{it}(\cdot)$ with $\delta_{it}(0)=0$. The last inequality is by the definition of $\Omega_{3}$ and $\Omega_{4}$.
Let
\[\tau'=\frac{\sqrt{\log\left(NT+1/\gamma\right)}}{2\left[2\sqrt{2C_{\zeta}p\log(pNT)(2C_{X}\lor C_{Sp}^{2}C_{Cone}^{2})}+C_{Sp}\sqrt{N\lor T}(4\sqrt{6C_{\zeta}}+\sqrt{2}\zeta_{2})\sqrt{\bar{r}}\right]\gamma}\]
and 
\begin{align*}
&K_{3}=\\
&2\left[2\sqrt{2C_{\zeta}p\log(pNT)(2C_{X}\lor C_{Sp}^{2}C_{Cone}^{2})}+C_{Sp}\sqrt{N\lor T}(4\sqrt{6C_{\zeta}}+\sqrt{2}\zeta_{2})\sqrt{\bar{r}}\right]\gamma\sqrt{\log\left(NT+1/\gamma\right)}.
\end{align*}
Since $\exp(\sqrt{\log(NT+1/\gamma)})\leq \sqrt{\exp(\log(NT+1/\gamma))}= \sqrt{NT+1/\gamma}$ for large enough $N$ and $T$, substituting equation \eqref{eq.C0exp} into \eqref{eq.C0C1}, we obtain
\begin{align}
&\mathbb{P}\left(\mathcal{C}^{0}(\gamma)\geq 2K_{3}\big|\Omega_{1}\cap\Omega_{2}\right)\notag\\
=&\mathbb{E}\left[\mathbb{P}\left(\mathcal{C}^{0}(\gamma)\geq 2K_{3}\big |W\right)\big|\Omega_{1}\cap\Omega_{2}\right]\notag\\
\leq &\frac{1}{\sqrt{\gamma^{2}NT+\gamma}}\mathbb{E}\left[\mathbb{P}\left(\Omega_{3}\cap\Omega_{4}|W\right)\big|\Omega_{1}\cap\Omega_{2}\right]+\mathbb{E}\left[\mathbb{P}\left((\Omega_{3}\cap\Omega_{4})^{c}|W\right)\big|\Omega_{1}\cap\Omega_{2}\right]\notag\\
=&\frac{1}{\sqrt{\gamma^{2}NT+\gamma}}\mathbb{P}\left(\Omega_{3}\cap\Omega_{4}\big|\Omega_{1}\cap\Omega_{2}\right)+\mathbb{P}\left((\Omega_{3}\cap\Omega_{4})^{c}\big|\Omega_{1}\cap\Omega_{2}\right)\to 0
\end{align}
where the first equality is by the law of iterated expectation since $\max_{1\leq j\leq p}||X_{j}||_{F}^{2}$ and $||L_{0}(u')-L_{0}(u)||_{F}$ in $\Omega_{1}$ and $\Omega_{2}$ are both functions of $W$. Convergence in the last line is by equation \eqref{eq.Omega23} and by $\gamma^{2}NT+\gamma\to \infty$ following the definition of $\gamma$. Therefore, by $\mathbb{P}(\Omega_{1}\cap\Omega_{2})\to 1$,
\begin{align}
&\mathbb{P}\left(\mathcal{C}^{0}(\gamma)\geq 2K_{3}\big|\Omega_{1}\right)\leq\mathbb{P}\left(\mathcal{C}^{0}(\gamma)\geq 2K_{3}\big|\Omega_{1}\cap\Omega_{2}\right)\mathbb{P}(\Omega_{1}\cap\Omega_{2})/\mathbb{P}(\Omega_{1})+\mathbb{P}(\Omega_{2}^{c}|\Omega_{1})\to 0\label{epC}
\end{align}


Now, recall $s=C_{sup}\sqrt{\log(NT)}\left(\sqrt{p\log(pNT)}\lor\sqrt{\bar{r}(N\lor T)}\right)\gamma$ and let
\[
C_{sup}=16\sqrt{2}\left(2\sqrt{2(2C_{X}\lor C_{Sp}^{2}C_{Cone}^{2})C_{\zeta}}+ C_{Sp}\left(4\sqrt{6C_{\zeta}}+\sqrt{2}\zeta_{2}\right)\right).
\]
By $\log(NT+1/\gamma)\leq \log(2NT)\leq 2\log(NT)$, we have $s\geq 4\max\{K_{1},K_{2},2K_{3}\}$. Combining equations \eqref{epA0}, \eqref{epB1}, \eqref{epB2} and \eqref{epC}, we have
\begin{equation}\label{epA0done}
\mathbb{P}\left(\mathcal{A}^{0}(\gamma)\geq\frac{s}{4}\big |\Omega_{1} \right)\geq \mathbb{P}\left(\mathcal{A}^{0}(\gamma)\geq \max\{K_{1},K_{2},2K_{3}\}|\Omega_{1}\right)\to 0.
\end{equation} 
 By equations \eqref{eq.At}, \eqref{epA} and \eqref{epA0done},  we obtain $\mathbb{P}(\mathcal{A}(\gamma)\geq s)\to 0$.\qed


\subsection{Proofs of the Results in Section \ref{3.sec4}}\label{app.sec4}
\subsubsection{Proof of Theorem \ref{3.thm1}}
Let $\Omega_{0}$ be the event that i) $(\hat{\Delta}_{\beta}(u),\hat{\Delta}_{L}(u))\in\mathcal{R}_{u}\cap\mathcal{D}$ and $||L_{0}(u)||_{\infty}\leq \alpha_{NT}$ uniformly in $u\in\mathcal{U}$, ii) inequality \eqref{eq.lowerbound} holds, iii) the smallest singular value of $\sum_{i,t}(X_{it}X_{it}')/NT$ is greater than $C_{\sigma}\sigma_{min}^{2}$ for some constant $C_{\sigma}>0$ and inequality \eqref{eq.rsc} holds, and iv) the uniform bound on the error process $\mathbb{G}_{u}$ in Lemma \ref{3.lem3} holds. By Lemmas \ref{3.lem1} to \ref{3.lem3} and Assumptions \ref{3.ass3} and \ref{ass.id}, the event $\Omega_{0}$ occurs w.p.a.1. It is then sufficient to show that the following event has zero probability under $\Omega_{0}$:
\begin{equation}
\exists u\in\mathcal{U}:||\hat{\Delta}_{\beta}(u)||_{F}^{2}+\frac{1}{NT}||\hat{\Delta}_{L}(u)||_{F}^{2}\geq\gamma^{2}\label{eq.contra}
\end{equation}
where $\gamma=C_{error}\sqrt{\log(NT)}\left (\sqrt{p\log(pNT)/NT}\lor\sqrt{\bar{r}/(N\land T)}\right)$ for some $C_{error}>0$ about which we will be precise later.

Since $\mathcal{R}_{u}$ is a cone and zero is contained in $\mathcal{D}$ which is a convex set, for any $(\Delta_{\beta},\Delta_{L})\in \mathcal{R}_{u}\cap\mathcal{D}$ and any $\eta\in (0,1)$, $(\eta \Delta_{\beta},\eta\Delta_{L})\in \mathcal{R}_{u}\cap\mathcal{D}$. By this observation, by the definition of the estimator \eqref{3.eq3} and $L_{0}(u)\in\mathcal{L}$ for all $u\in\mathcal{U}$, and by convexity of the objective function, equation \eqref{eq.contra} implies that there exists a $u\in\mathcal{U}$ such that
\begin{align}
0> \inf_{\substack{(\Delta_{\beta},\Delta_{L})\in\mathcal{R}_{u}\cap\mathcal{D}\\||\Delta_{\beta}||_{F}^{2}+\frac{1}{NT}||\Delta_{L}||_{F}^{2}=\gamma^{2}}}&\frac{1}{NT}\left[\bm{\rho}_{u} \big(V(u)-\sum_{j=1}^{p} X_{j} \Delta_{\beta,j}-\Delta_{L}\big)-\bm{\rho}_{u} \left(V(u)\right)\right]\notag\\
& +\lambda \left[||L_{0}(u)+\Delta_{L}||_{*}-||L_{0}(u)||_{*}\right]\label{eq.forcontra}\\
=\inf_{\substack{(\Delta_{\beta},\Delta_{L})\in\mathcal{R}_{u}\cap\mathcal{D}\\||\Delta_{\beta}||_{F}^{2}+\frac{1}{NT}||\Delta_{L}||_{F}^{2}=\gamma^{2}}}& \frac{1}{NT} \mathbb{E}\left[\bm{\rho}_{u} \big(V(u)-\sum_{j=1}^{p} X_{j} \Delta_{\beta,j}-\Delta_{L}\big)-\bm{\rho}_{u} (V(u))\Big |W\right]\notag\\
&+\frac{1}{\sqrt{NT}}\mathbb{G}_{u}\left(\rho_{u} \big(V_{it}(u)-X_{it}' \Delta_{\beta}-\Delta_{L,it}\big)-\rho_{u} (V_{it}(u))\right)\notag\\
&+\lambda \left[||L_{0}(u)+\Delta_{L}||_{*}-||L_{0}(u)||_{*}\right]\label{eq.baseproof}
\end{align}

For the conditional expectation, equations \eqref{eq.lowerbound} and \eqref{eq.rsc} and iii) in $\Omega_{0}$ imply that under $\Omega_{0}$, with probability one,
\begin{align}
&\inf_{\substack{(\Delta_{\beta},\Delta_{L})\in\mathcal{R}_{u}\cap\mathcal{D}\\||\Delta_{\beta}||_{F}^{2}+\frac{1}{NT}||\Delta_{L}||_{F}^{2}=\gamma^{2}}}\frac{1}{NT} \mathbb{E}\left[\bm{\rho}_{u} \big(V(u)-\sum_{j=1}^{p} X_{j} \Delta_{\beta,j}-\Delta_{L}\big)-\bm{\rho}_{u} (V(u))\Big|W\right]\notag\\
\geq& \inf_{\substack{(\Delta_{\beta},\Delta_{L})\in\mathcal{R}_{u}\\||\Delta_{\beta}||_{F}^{2}+\frac{1}{NT}||\Delta_{L}||_{F}^{2}=\gamma^{2}}}\frac{C_{min}}{\alpha_{NT}^{2}NT}||\sum_{j=1}^{p}X_{j}\Delta_{\beta,j}+\Delta_{L}||_{F}^{2}\notag\\
\geq& \inf_{||\Delta_{\beta}||_{F}^{2}+\frac{1}{NT}||\Delta_{L}||_{F}^{2}=\gamma^{2}} \frac{C_{min}C_{RSC}}{\alpha_{NT}^{2}NT}\left(\sum_{i,t}\left(X_{it}'\Delta_{\beta}\right)^{2}+||\Delta_{L}||_{F}^{2}\right)\notag\\
\geq& \inf_{||\Delta_{\beta}||_{F}^{2}+\frac{1}{NT}||\Delta_{L}||_{F}^{2}=\gamma^{2}} \frac{C_{min}C_{RSC}}{\alpha_{NT}^{2}}\left(C_{\sigma}\sigma_{min}^{2}||\Delta_{\beta}||_{F}^{2}+\frac{||\Delta_{L}||_{F}^{2}}{NT}\right)\notag\\
\geq& \frac{C_{min}C_{RSC}}{\alpha_{NT}^{2}}\left(C_{\sigma}\sigma_{min}^{2}\land 1\right)\gamma^{2}\label{eq.minor}
\end{align}

For the error process $\mathbb{G}_{u}/\sqrt{NT}$, by iv) in $\Omega_{0}$,
\begin{align}
&\frac{1}{\sqrt{NT}}\sup_{\substack{(\Delta_{\beta},\Delta_{L})\in\mathcal{R}_{u}\\||\Delta_{\beta}||_{F}^{2}+\frac{1}{NT}||\Delta_{L}||_{F}^{2}=\gamma^{2}}}\big|\mathbb{G}_{u}\left(\rho_{u} \big(V_{it}(u)-X_{it}' \Delta_{\beta}-\Delta_{L,it}\big)-\rho_{u} (V_{it}(u))\right)\big|\notag\\
\leq &C_{sup}\sqrt{\log(NT)}\left (\sqrt{\frac{p\log(pNT)}{NT}}\lor\sqrt{\frac{\bar{r}}{N\land T}}\right)\gamma\label{eq.sup}
\end{align}
with probability one under $\Omega_{0}$.

Finally, by the choice of $\lambda$ and by the definition of $\mathcal{R}_{u}$,
\begin{align}
&\sup_{\substack{(\Delta_{\beta},\Delta_{L})\in\mathcal{R}_{u}\\||\Delta_{\beta}||_{F}^{2}+\frac{1}{NT}||\Delta_{L}||_{F}^{2}=\gamma^{2}}}\lambda \big|\left[||L_{0}(u)+\Delta_{L}||_{*}-||L_{0}(u)||_{*}\right]\big|\notag\\
\leq &\sup_{\substack{(\Delta_{\beta},\Delta_{L})\in\mathcal{R}_{u}\\||\Delta_{\beta}||_{F}^{2}+\frac{1}{NT}||\Delta_{L}||_{F}^{2}=\gamma^{2}}}\lambda ||\Delta_{L}||_{*}\notag\\
\leq &\lambda\sup_{||\Delta_{L}||_{F}^{2}\leq NT \gamma^{2}}\big(C_{Cone}\sqrt{p(N\land T)\log(NT)}\gamma+4||\mathcal{P}_{\Phi(u)}\Delta_{L}||_{*}\big)\notag\\
\leq &\lambda\sup_{||\Delta_{L}||_{F}^{2}\leq NT \gamma^{2}}\big(C_{Cone}\sqrt{p(N\land T)\log(NT)}\gamma+4\sqrt{3\bar{r}}||\Delta_{L}||_{F}\big)\notag\\
\leq &\frac{2\sqrt{2C_{X}}}{C_{Cone}}\left(C_{Cone}+ 4\sqrt{3}\right)\left(\sqrt{\frac{p\log(NT)}{NT}}\lor \sqrt{\frac{\bar{r}}{N\land T}}\right)\gamma\notag\\
\leq & C_{sup}\sqrt{\log(NT)}\left (\sqrt{\frac{p\log(pNT)}{NT}}\lor\sqrt{\frac{\bar{r}}{N\land T}}\right)\gamma\label{eq.penalty}
\end{align}
with probability one under $\Omega_{0}$. The last inequality holds for large enough $N$ and $T$ regardless of the exact values of the constants $C_{X}$, $C_{Cone}$ and $C_{sup}$.

Let $C_{error}=2C_{sup}/(C_{min}C_{RSC}(C_{\sigma}\sigma_{min}^{2}\land 1))$. By equations \eqref{eq.minor}, \eqref{eq.sup} and \eqref{eq.penalty}, the right side of equation \eqref{eq.baseproof} is equal to 
\begin{align*}
\gamma\cdot\left( \frac{C_{min}C_{RSC}}{\alpha_{NT}^{2}}\left(C_{\sigma}\sigma_{min}^{2}\land 1\right)\gamma-2C_{sup}\sqrt{\log(NT)}\left (\sqrt{\frac{p\log(pNT)}{NT}}\lor\sqrt{\frac{\bar{r}}{N\land T}}\right)\right)=0,
\end{align*}
with probability one. Hence, under $\Omega_{0}$, the event that inequality \eqref{eq.forcontra} holds and thus the event that inequality \eqref{eq.contra} holds have zero probability. \qed

\subsubsection{Proof of Corollary \ref{cor.sv}}
By Weyl's inequality for singular values, 
\begin{align*}
\max_{k\in\{1,...,N\land T\}}\left\lbrace |\hat{\sigma}_{k}(u)-\sigma_{k}(u)|\right\rbrace\leq ||\hat{\Delta}_{L}(u)||\leq ||\hat{\Delta}_{L}(u)||_{F}.
\end{align*}
with probability one. Hence, $\sup_{u\in\mathcal{U}}\max_{k}\left\lbrace |\hat{\sigma}_{k}(u)-\sigma_{k}(u)|\right\rbrace\leq \sup_{u\in\mathcal{U}} ||\hat{\Delta}_{L}(u)||_{F}$. Equation \eqref{eq.cor} thus follows by plugging in the uniform rate of $||\hat{\Delta}_{L}(u)||_{F}$ obtained in Theorem \ref{3.thm1} and by $\sigma_{r(u)+1}(u)=\cdots=\sigma_{N\land T}(u)=0$.\qed

\subsubsection{Proof of Corollary \ref{cor.rank}}
By the definition of $\hat{r}(u)$, the event $\{\hat{r}(u)=r(u)\}$ is equivalent to $\{\hat{\sigma}_{r(u)}\geq C_{r}\}\cap\{\hat{\sigma}_{r(u)+1}< C_{r}\}$. The latter event,  under the event $\Omega_{sv2}=\{\sigma_{r(u)}\text{ is of the order of } \sqrt{NT}\}$, can be implied by $\Omega_{sv1}\coloneqq\{|\hat{\sigma}_{r(u)}-\sigma_{r(u)}|\leq \sqrt{NT}\gamma,|\hat{\sigma}_{r(u)+1}-0|\leq \sqrt{NT}\gamma\}$  for sufficiently large $N$ and $T$ by the choice of $C_{r}$. The desired result is thus obtained since w.p.a.1, $\Omega_{sv1}$ is true by Corollary \ref{cor.sv} and  $\Omega_{sv2}$ is true by assumption.

\subsection{Proof of Theorem \ref{appx.thmmain} in Appendix \ref{app.second}}\label{app.techsecond}

Since $\mathcal{D}^{(2)}$ is a cone, the main argument in the proof of Theorem \ref{3.thm1} still holds. Meanwhile, Lemmas \ref{3.lem1} and \ref{3.lem3} and the upper bound on the penalty difference in equation \eqref{eq.base} does not depend on $\mathcal{D}$ nor on the lower bound obtained in \ref{step1}. The desired error bound in Theorem \ref{appx.thmmain} thus follows once we establish an inequality similar to equation \eqref{eq.minor} in the proof of Theorem \ref{3.thm1}.

With the covariates, recall $\mathcal{D}^{(2)}=\mathbb{R}^{p}\times \left\lbrace\Delta_{L}\in\mathbb{R}^{N\times T}:||\mathcal{P}_{\Omega}\Delta_{L}||_{F}^{2}\geq C_{sm} ||\Delta_{L}||_{F}^{2}\right\rbrace$. Let $\gamma=C_{error,2}\alpha^{2}_{NT}\sqrt{\log(NT)}((p\log(pNT)/NT)\lor (\bar{r}/(N\land T)))$ for $C_{error,2}=2.5C_{error}$ where $C_{error}$ is the constant in Theorem \ref{3.thm1}. Let $\Omega_{0}'$ be the event that i) $(\hat{\Delta}_{\beta}(u),\hat{\Delta}_{L}(u))\in\mathcal{R}_{u}\cap\mathcal{D}^{(2)} $ uniformly in $u\in\mathcal{U}$, ii) $\sum_{j=1}^{p}||X_{j}||_{\infty}\gamma\leq \alpha_{NT}$ and iii) the smallest singular value of $\sum_{i,t}(X_{it}X_{it}')/NT$ is greater than $C_{\sigma}\sigma_{min}^{2}$ for some constant $C_{\sigma}>0$ and inequality \eqref{eq.rsc} holds. By Lemma \ref{3.lem1}, Assumptions \ref{3.ass3} and \ref{ass.id} and the condition in the theorem, the event $\Omega_{0}'$ occurs w.p.a.1. 

 Under $\Omega_{0}'$, with probability one we have
\begin{align}
&\inf_{\substack{(\Delta_{\beta},\Delta_{L})\in\mathcal{R}_{u}\cap\mathcal{D}^{(2)}\\||\Delta_{\beta}||_{F}^{2}+\frac{1}{NT}||\Delta_{L}||_{F}^{2}=\gamma^{2}}}\frac{1}{NT} \mathbb{E}\left[\bm{\rho}_{u} \big(V(u)-\sum_{j=1}^{p} X_{j} \Delta_{\beta,j}-\Delta_{L}\big)-\bm{\rho}_{u} (V(u))\Big|W\right]\notag\\
=&\inf_{\substack{(\Delta_{\beta},\Delta_{L})\in\mathcal{R}_{u}\cap\mathcal{D}^{(2)}\\||\Delta_{\beta}||_{F}^{2}+\frac{1}{NT}||\Delta_{L}||_{F}^{2}=\gamma^{2}}}\frac{1}{NT} \sum_{i,t}\int_{0}^{\Delta_{L,it}+X_{it}'\Delta_{\beta}}\left(F_{V_{it}(u)|W}(s)-F_{V_{it}(u)|W}(0)\right)ds\notag\\
\geq &\inf_{\substack{(\Delta_{\beta},\Delta_{L})\in\mathcal{R}_{u}\cap\mathcal{D}^{(2)}\\||\Delta_{\beta}||_{F}^{2}+\frac{1}{NT}||\Delta_{L}||_{F}^{2}=\gamma^{2}}}\frac{1}{NT}\sum_{\{i,t:|\Delta_{L,it}|\leq 2\alpha_{NT}\}}\int_{0}^{\Delta_{L,it}+X_{it}'\Delta_{\beta}}\left(F_{V_{it}(u)|W}(s)-F_{V_{it}(u)|W}(0)\right)ds\notag\\
\geq &\inf_{\substack{(\Delta_{\beta},\Delta_{L})\in\mathcal{R}_{u}\cap\mathcal{D}^{(2)}\\||\Delta_{\beta}||_{F}^{2}+\frac{1}{NT}||\Delta_{L}||_{F}^{2}=\gamma^{2}}}\frac{(1\land \delta)^{2}\underline{f}}{18\alpha_{NT}^{2}NT}||\sum_{j=1}^{p}X_{j}\Delta_{\beta,j}+\mathcal{P}_{\Omega}\Delta_{L}||_{F}^{2}\notag\\
\geq & \inf_{\substack{(\Delta_{\beta},\Delta_{L})\in\mathcal{R}_{u}\cap\mathcal{D}^{(2)}\\||\Delta_{\beta}||_{F}^{2}+\frac{1}{NT}||\Delta_{L}||_{F}^{2}=\gamma^{2}}} \frac{(1\land \delta)^{2}\underline{f}}{18\alpha_{NT}^{2}	NT}\left(0.5||\sum_{j=1}^{p}X_{j}\Delta_{\beta,j}+\Delta_{L}||_{F}^{2}-||\mathcal{P}_{\Omega^{\perp}}\Delta_{L}||_{F}^{2}\right)\notag\\
\geq &\inf_{||\Delta_{\beta}||_{F}^{2}+\frac{1}{NT}||\Delta_{L}||_{F}^{2}=\gamma^{2}} \frac{(1\land \delta)^{2}\underline{f}}{18\alpha_{NT}^{2}NT}\left(0.5C_{RSC}||\sum_{j=1}^{p}X_{j}\Delta_{\beta,j}||_{F}^{2}+\left(0.5C_{RSC}+C_{sm}-1\right)||\Delta_{L}||_{F}^{2}\right)\notag\\
\geq &\inf_{||\Delta_{\beta}||_{F}^{2}+\frac{1}{NT}||\Delta_{L}||_{F}^{2}=\gamma^{2}} \frac{0.4C_{RSC}(1\land \delta)^{2}\underline{f}}{18\alpha_{NT}^{2}NT}\left(||\sum_{j=1}^{p}X_{j}\Delta_{\beta,j}||_{F}^{2}+||\Delta_{L}||_{F}^{2}\right)\notag\\
\geq &\frac{0.4C_{RSC}C_{min}}{\alpha_{NT}^{2}}\left(C_{\sigma}\sigma_{min}^{2}\land 1\right)\gamma^{2}
\end{align} 
 where the equality follows the proof of Lemma \ref{lowerbound}. The first inequality holds because every integral in the summation is nonnegative. The second inequality again follows the proof of Lemma \ref{lowerbound}. The third inequality is elementary. The fourth inequality is by the definition of $\mathcal{D}^{(2)}$ in equation \eqref{eq.secondkey} and by the restricted strong convexity condition \eqref{eq.rsc}. The penultimate inequality is by $C_{sm}\to 1$. The last inequality is because the smallest singular value of $\sum_{i,t}(X_{it}X_{it}')/NT$ is greater than $C_{\sigma}\sigma_{min}^{2}$ and by $C_{min}=(1\land \delta)^{2}\underline{f}/18$. The proof is complete by substituting this lower bound into the proof of Theorem \ref{3.thm1}. \qed

\subsection{Proofs of the Lemmas in Appendices \ref{app.sec3} and \ref{app.sec4}}\label{app.techlem}
\subsubsection{Proof of Lemma \ref{3.lemA1}}
Recall that $\Omega_{1}$ is the event that $\max_{1\leq j\leq p}||X_{j}||_{F}^{2}\leq C_{X}NT$. Under Assumption \ref{3.ass1}, $\mathbb{P}(\Omega_{1})\to 1$. Recall that the $(i,t)$-th element in subgradient $\nabla\bm{\rho}_{u}(V(u))$ is
\begin{equation*}
\left(\nabla\bm{\rho}_{u}\left(V(u)\right)\right)_{it}=u\mathbbm{1}(V_{it}(u)>0)+(u-1)\mathbbm{1}(V_{it}(u)<0).
\end{equation*}
with probability one. By Assumption \ref{3.ass1} and by $V(u)=Y-q_{Y|W}(u)$, the elements in $\nabla\bm{\rho}_{u}\left(V(u)\right)$ are independent with mean $0$ conditional on $W$, and are uniformly bounded within $[-1,1]$. We start by proving equation \eqref{3.eqA1}.

\noindent\textbf{Proof of Equation \eqref{3.eqA1}}. Let $M=\sqrt{2C_{X}NT\log(pNT)}$. Note that
\begin{equation}\label{eqA1.1}
\mathbb{P}\left(\sup_{\substack{u\in\mathcal{U}\\1\leq j\leq p}}|\langle \nabla\bm{\rho}_{u}(V(u)),X_{j}\rangle|>M\right)\leq \mathbb{P}\left(\sup_{\substack{u\in\mathcal{U}\\1\leq j\leq p}}|\langle \nabla\bm{\rho}_{u}(V(u)),X_{j}\rangle|>M\bigg|\Omega_{1}\right)\mathbb{P}(\Omega_{1})+\mathbb{P}(\Omega_{1}^{c}).
\end{equation}
Since $\mathbb{P}(\Omega_{1}^{c})\to 0$, it is sufficient to show the conditional probability in equation \eqref{eqA1.1} converges to zero.

Let $\mathcal{U}_{K}=(u_{1},u_{2},...,u_{K})$ be an $\varepsilon$-net of $\mathcal{U}$. Let $\varepsilon= \frac{1}{\sqrt{NT}}$ and $K\varepsilon\leq 1$. By the triangle inequality,
\begin{align*}
&\sup_{\substack{u\in\mathcal{U}\\1\leq j\leq p}}\big|\big\langle \nabla\bm{\rho}_{u}(V(u)),X_{j}\big\rangle\big|\\
\leq & \max_{\substack{u_{k}\in\mathcal{U}_{K}\\1\leq j\leq p}}\big|\big\langle \nabla\bm{\rho}_{u_{k}}(V(u_{k})),X_{j}\big\rangle\big|+\sup_{\substack{|u-u_{k}|\leq \varepsilon,u_{k}\in\mathcal{U}_{K}\\1\leq j\leq p}}\big|\big\langle \nabla\bm{\rho}_{u}(V(u))-\nabla\bm{\rho}_{u_{k}}(V(u_{k})),X_{j}\big\rangle\big|\\
\eqqcolon& \mathcal{E}_{1}+ \mathcal{E}_{2}.
\end{align*}
Then,
\begin{equation}\label{eqA1.2}
\mathbb{P}\left(\sup_{\substack{u\in\mathcal{U}\\1\leq j\leq p}}|\langle \nabla\bm{\rho}_{u}(V(u)),X_{j}\rangle|>M\bigg|\Omega_{1}\right)\leq 2\max\left\lbrace \mathbb{P}\left(\mathcal{E}_{1}>M|\Omega_{1}\right),\mathbb{P}\left(\mathcal{E}_{2}>M|\Omega_{1}\right)\right\rbrace.
\end{equation}

\textbf{Bound on $\mathcal{E}_{1}$}. By $K\leq 1/\varepsilon$, we have
\begin{align}
\mathbb{P}\left(\mathcal{E}_{1}>M|\Omega_{1}\right)
 \leq &\frac{p}{\varepsilon}\max_{\substack{u_{k}\in\mathcal{U}_{K}\\1\leq j\leq p}}\mathbb{P}\left(\big|\big\langle \nabla\bm{\rho}_{u_{k}}(V(u_{k})),X_{j}\big\rangle\big|\geq M\big|\Omega_{1}\right)\notag\\
 =&\frac{p}{\varepsilon}\max_{\substack{u_{k}\in\mathcal{U}_{K}\\1\leq j\leq p}}\mathbb{E}\left[\mathbb{P}\left(\big|\big\langle \nabla\bm{\rho}_{u_{k}}(V(u_{k})),X_{j}\big\rangle\big|\geq M\big|\Omega_{1},W\right)\big|\Omega_{1}\right]\notag\\
\leq & 2p\sqrt{NT}\exp\left( -\frac{M^{2}}{2C_{X}NT}\right)\notag\\
=&\frac{2}{\sqrt{NT}}\to 0\label{eqA1.3}
\end{align}
where the first inequality is by the union bound and the following equality is due to the law of iterative expectation. The penultimate inequality is by Hoeffding's inequality and by $\Omega_{1}$ under $\varepsilon=1/\sqrt{NT}$.

\textbf{Bound on $\mathcal{E}_{2}$}. By definition, the $(i,t)$-th element in $\nabla\bm{\rho}_{u}(V(u))-\nabla\bm{\rho}_{u_{k}}(V(u_{k}))$ is almost surely
\begin{align}
&u\mathbbm{1}(V_{it}(u)>0)+(u-1)\mathbbm{1}(V_{it}(u)<0)-\left[u_{k}\mathbbm{1}(V_{it}(u_{k})>0)+(u_{k}-1)\mathbbm{1}(V_{it}(u_{k})<0)\right]\notag\\
=&(u-u_{k})+\mathbbm{1}(V_{it}(u_{k})<0)-\mathbbm{1}(V_{it}(u)<0)\label{eqA1.decom}
\end{align}
because by continuity of $V_{it}(u)$ at zero (implied by the existence of a positive density function around zero under Assumption \ref{3.ass2}), it equals zero with probability zero. Let $\Xi_{1}$ and $\Xi_{2}$ be two $N\times T$ matrices whose $(i,t)$-th elements are
\begin{align}
&\Xi_{1,it}\coloneqq u-u_{k}\label{Xi1}\\
&\Xi_{2,it}\coloneqq \mathbbm{1}(V_{it}(u_{k})<0)-\mathbbm{1}(V_{it}(u)<0).\label{Xi2}
\end{align} 
Then by equation \eqref{eqA1.decom}, $\nabla\bm{\rho}_{u}(V(u))-\nabla\bm{\rho}_{u_{k}}(V(u_{k}))=\Xi_{1}+\Xi_{2}$. Therefore,
\begin{align}
\mathbb{P}(\mathcal{E}_{2}>M|\Omega_{1})\leq &\mathbb{P}\left(\sup_{\substack{|u-u_{k}|\leq  \varepsilon\\u_{k}\in\mathcal{U}_{K}\\1\leq j\leq p}}|\left\langle \Xi_{1},X_{j}\right\rangle|>M\Bigg|\Omega_{1}\right)+\mathbb{P}\left(\sup_{\substack{|u-u_{k}|\leq \varepsilon\\u_{k}\in\mathcal{U}_{K}\\1\leq j\leq p}}|\left\langle \Xi_{2},X_{j}\right\rangle|>M\Bigg|\Omega_{1}\right)\notag\\
\eqqcolon&\mathbb{P}_{1}+\mathbb{P}_{2}\label{eqA1.4}
\end{align}

We first show $\mathbb{P}_{1}=0$. By the Cauchy-Schwartz inequality,
\begin{align}
\mathbb{P}\left(\sup_{\substack{|u-u_{k}|\leq \varepsilon,u_{k}\in\mathcal{U}_{K}\\1\leq j\leq p}}|\left\langle \Xi_{1},X_{j}\right\rangle|>M\Bigg|\Omega_{1}\right)\leq&\mathbb{P}\left(\sup_{\substack{|u-u_{k}|\leq \varepsilon,u_{k}\in\mathcal{U}_{K}\\1\leq j\leq p}}||\Xi_{1}||_{F}||X_{j}||_{F}>M\Bigg|\Omega_{1}\right)\notag\\ 
\leq & \mathbb{P}\left(\varepsilon \sqrt{C_{X}}NT>M\right)\notag\\
=&  \mathbb{P}\left(\sqrt{C_{X}NT}>\sqrt{2C_{X}\log(pNT)NT}\right)\notag\\
=&0\label{eqA1.5}
\end{align}
for large enough $N$ and $T$. The second inequality is by the definition of $\Xi_{1}$ and $\Omega_{1}$. 	The penultimate equality is by $\varepsilon=1/\sqrt{NT}$. 

Now we show that $\mathbb{P}_{2}$ converges to zero. Let $\Xi_{2}^{(1)}$ and $\Xi_{2}^{(2)}$ be two $N\times T$ matrices whose $(i,t)$-th elements are
\begin{align}
&\Xi_{2,it}^{(1)}(u_{k}) \coloneqq \mathbbm{1}(V_{it}(u_{k})<0)-\mathbbm{1}(V_{it}(u_{k}-\varepsilon)<0)\label{Xi21}\\
&\Xi_{2,it}^{(2)}(u_{k})\coloneqq \mathbbm{1}(V_{it}(u_{k})<0)-\mathbbm{1}(V_{it}(u_{k}+\varepsilon)<0)\label{Xi22}
\end{align}
Consider a generic element in $\Xi_{2}$: $\Xi_{2,it}\coloneqq\mathbbm{1}(V_{it}(u_{k})<0)-\mathbbm{1}(V_{it}(u)<0)$. By $V_{it}(u)=Y_{it}-q_{Y_{it}|W}(u)$, $V_{it}(u)$ is strictly decreasing in $u$ almost surely. Hence, $\mathbbm{1}(V_{it}(u)<0)$ is weakly increasing in $u$ almost surely. Consequently, if $u_{k}-\varepsilon\leq u\leq u_{k}$, then $0\leq \Xi_{2,it}\leq \Xi_{2,it}^{(1)}(u_{k})\leq 1$. Similarly, if $u_{k}+\varepsilon\geq u\geq u_{k}$, then $0\geq \Xi_{2,it}\geq \Xi_{2,it}^{(2)}(u_{k})\geq -1$. The following inequalities thus hold with probability one.
\begin{align*}
\sup_{\substack{|u-u_{k}|\leq \varepsilon,u_{k}\in\mathcal{U}_{K}\\1\leq j\leq p}}\big|\langle \Xi_{2},X_{j}\rangle\big|\leq& \sup_{\substack{u_{k}-\varepsilon\leq u\leq u_{k},u_{k}\in\mathcal{U}_{K}\\1\leq j\leq p}}\big|\langle \Xi_{2},X_{j}\rangle\big|+\sup_{\substack{u_{k}\leq u\leq u_{k}+\varepsilon,u_{k}\in\mathcal{U}_{k}\\1\leq j\leq p}}\big|\langle \Xi_{2},X_{j}\rangle\big|\\
\leq &  \sup_{\substack{u_{k}-\varepsilon\leq u\leq u_{k},u_{k}\in\mathcal{U}_{K}\\1\leq j\leq p}}\big|\langle \Xi_{2},|X_{j}|\rangle\big|+\sup_{\substack{u_{k}\leq u\leq u_{k}+\varepsilon,u_{k}\in\mathcal{U}_{k}\\1\leq j\leq p}}\big|\langle \Xi_{2},|X_{j}|\rangle\big|\\
\leq & \max_{\substack{u_{k}\in\mathcal{U}_{K}\\1\leq j\leq p}}\big|\langle \Xi_{2}^{(1)},|X_{j}|\rangle\big|+ \max_{\substack{u_{k}\in\mathcal{U}_{K}\\1\leq j\leq p}}\big|\langle \Xi_{2}^{(2)},|X_{j}|\rangle\big|
\end{align*}
The first inequality is elementary. To see why the second inequality holds, note that the elements in $\Xi_{2}$ are all nonnegative when $u\in [u_{k}-\varepsilon,u_{k}]$ and are all nonpositive when $u\in [u_{k},u_{k}+\varepsilon]$. So, for a given $\Xi_{2}$ and a given $X_{j}$, the two absolute inner products on the right side of the first inequality increase if we flip the signs of the $X_{j,it}$s so that they also have the same sign. The third inequality then follows because now that elements in both $|X_{j}|$ and in $\Xi_{2}$ have the same signs, the two absolute inner products in the second line increase as the magnitude of any of the elements in $\Xi_{2}$ increases. 

Let us first bound $\max_{u_{k}\in\mathcal{U}_{k},1\leq j\leq p}\big|\langle \Xi_{2}^{(1)},|X_{j}|\rangle\big|$. The expectation of a generic element $\Xi_{2,it}^{(1)}$ in $\Xi_{2}^{(1)}$ satisfies
\begin{align*}
&\mathbb{E}\left(\Xi_{2,it}^{(1)}(u_{k})\big |W\right)\\
=&\mathbb{P}\left(V_{it}(u_{k})< 0\leq V_{it}(u_{k}-\varepsilon)\big |W\right)\\
=&\mathbb{P}\left(q_{Y_{it}|W}(u_{k}-\varepsilon)\leq Y_{it}< q_{Y_{it}|W}(u_{k})\big |W\right)\\
= &\varepsilon
\end{align*}
where the second equality is by the definition of $V_{it}(u)$. Let $\bar{\Xi}_{2}^{(1)}=\mathbb{E}\left(\Xi_{2}^{(1)}(u_{k})\big |W\right)$ be an $N\times T$ matrix whose elements are all equal to $\varepsilon$. Under $\Omega_{1}$, by the Cauchy-Schwartz inequality and by $\varepsilon=1/\sqrt{NT}$, we have $\max_{u_{k}\in\mathcal{U}_{k}, 1\leq j\leq p}\big|\langle\bar{\Xi}_{2}^{(1)},|X_{j}|\rangle\big|\leq \varepsilon  \sqrt{C_{X}}NT=\sqrt{C_{X}NT}$ with probability one. Therefore,
\begin{align}
&\mathbb{P}\left(\max_{\substack{u_{k}\in\mathcal{U}_{k}\\ 1\leq j\leq p}}\big|\langle \Xi_{2}^{(1)},|X_{j}|\rangle\big|>M\Bigg |\Omega_{1}\right)\notag\\
&\leq \mathbb{P}\left(\max_{\substack{u_{k}\in\mathcal{U}_{k}\\ 1\leq j\leq p}}\big|\langle \Xi_{2}^{(1)}-\bar{\Xi}_{2}^{(1)},|X_{j}|\rangle\big|+\max_{\substack{u_{k}\in\mathcal{U}_{k}\\ 1\leq j\leq p}}\big|\langle\bar{\Xi}_{2}^{(1)},|X_{j}|\rangle\big|>M\Bigg |\Omega_{1}\right)\notag\\
&\leq \mathbb{P}\left(\max_{\substack{u_{k}\in\mathcal{U}_{k}\\ 1\leq j\leq p}}\big|\langle \Xi_{2}^{(1)}-\bar{\Xi}_{2}^{(1)},|X_{j}|\rangle\big|>M-\sqrt{C_{X}NT}\Bigg |\Omega_{1}\right)\notag\\
&=\frac{p}{\varepsilon}\max_{\substack{u_{k}\in\mathcal{U}_{k}\\ 1\leq j\leq p}}\mathbb{E}\left[\mathbb{P}\left(\big|\langle \Xi_{2}^{(1)}-\bar{\Xi}_{2}^{(1)},|X_{j}|\rangle\big|>M-\sqrt{C_{X}NT}\Bigg | W\right)\Bigg|\Omega_{1}\right]\notag\\
&\leq 2p\sqrt{NT}\exp\left(-\frac{(M-\sqrt{C_{X}NT})^{2}}{2C_{X}NT}\right)\to 0\label{eqA1.6}
\end{align}
The penultimate equality is by the law of iterated expectation since $\max_{1\leq j\leq p}||X_{j}||_{F}^{2}$ in $\Omega_{1}$ is a function of $W$. The last inequality is by Hoeffding's inequality since conditional on $W$, elements in $(\Xi_{2}^{(1)}-\bar{\Xi}_{2}^{(1)})$ are independent with zero mean and are bounded within $[-1,1]$. 

Finally, we can show that $\mathbb{P}\left(\max_{u_{k}\in\mathcal{U}_{k},1\leq j\leq p}|\langle \Xi_{2}^{(2)},|X_{j}|\rangle|>M |\Omega_{1}\right)\to 0$ as well following exactly the same argument. Combining it with equations \eqref{eqA1.1}, \eqref{eqA1.2}, \eqref{eqA1.3}, \eqref{eqA1.4}, \eqref{eqA1.5} and \eqref{eqA1.6}, we obtain the desired result.

\noindent \textbf{Proof of Equation \eqref{3.eqA2}}.
Again, let $\mathcal{U}_{K}=(u_{1},u_{2},...,u_{K})$ be an $\varepsilon$-net of $\mathcal{U}$ with $\varepsilon K\leq 1$. This time let $\varepsilon= 1/\sqrt{N\lor T}$. By the triangle inequality, we have
\begin{align}
\sup_{u\in\mathcal{U}}||\nabla\bm{\rho}_{u}(V(u))||
\leq&  \max_{u_{k}\in\mathcal{U}_{K}}||\nabla\bm{\rho}_{u_{k}}(V(u_{k}))||+\sup_{|u-u_{k}|\leq \varepsilon,u_{k}\in\mathcal{U}_{K}}|| \nabla\bm{\rho}_{u}(V(u))-\nabla\bm{\rho}_{u_{k}}(V(u_{k}))||\notag\\
\eqqcolon &\mathcal{F}_{1}+\mathcal{F}_{2}\label{eqA2.1}
\end{align}
Let $M=C_{op}\sqrt{N\lor T}$ for some $C_{op}>2$. By equation \eqref{eqA2.1},
\begin{equation}\label{eqA2.2}
\mathbb{P}\left(\sup_{u\in\mathcal{U}}||\nabla\bm{\rho}_{u}(V(u))||>M\right)\leq 2\max\left\lbrace\mathbb{P}\left(\mathcal{F}_{1}>M\right),\mathbb{P}\left(\mathcal{F}_{2}>M\right)\right\rbrace
\end{equation}

\textbf{Bound on $\mathcal{F}_{1}$}. Recall that conditional on $W$, the elements in $\nabla\bm{\rho}_{u}(V(u))$ are independent with mean zero and bounded within $[-1,1]$ for all $u\in\mathcal{U}$ with probability one. Therefore, by $\varepsilon K\leq 1$ and the union bound, there exist universal constants $0<c\leq C_{op}$ and $c'>0$ such that 
\begin{align}
\mathbb{P}\left(\mathcal{F}_{1}>M\right)&\leq \frac{1}{\varepsilon}\max_{u_{k}\in\mathcal{U}_{K}}\mathbb{P}\left(||\nabla\bm{\rho}_{u_{k}}(V(u_{k}))||>M\right)\notag\\
&=\frac{1}{\varepsilon}\max_{u_{k}\in\mathcal{U}_{K}}\mathbb{E}\left[\mathbb{P}\left(||\nabla\bm{\rho}_{u_{k}}(V(u_{k}))||>M|W\right)\right]\notag\\
&\leq c\sqrt{N\lor T}\exp(-C_{op}c'(N\lor T))\to 0 \label{eqA2.3}
\end{align}
where the second line is by the law of iterated expectation. The last inequality follows from Corollary 2.3.5 in \cite{tao2012topics} (p.129) that bounds the spectral norm of a matrix with independent mean zero entries that are bounded in magnitude by 1. 

\textbf{Bound on $\mathcal{F}_{2}$}. 	Similar to the proof of equation \eqref{3.eqA1},
\begin{align}
\sup_{|u-u_{k}|\leq \varepsilon,u_{k}\in\mathcal{U}_{K}}|| \nabla\bm{\rho}_{u}(V(u))-\nabla\bm{\rho}_{u_{k}}(V(u_{k}))||\leq &  \sup_{|u-u_{k}|\leq \varepsilon,u_{k}\in\mathcal{U}_{K}} ||\Xi_{1}||+ \sup_{|u-u_{k}|\leq \varepsilon,u_{k}\in\mathcal{U}_{K}}||\Xi_{2}||\notag\\
\leq & \varepsilon\sqrt{NT}+ \sup_{|u-u_{k}|\leq \varepsilon,u_{k}\in\mathcal{U}_{K}}||\Xi_{2}||\notag\\
=& \sqrt{N\land T}+ \sup_{|u-u_{k}|\leq \varepsilon,u_{k}\in\mathcal{U}_{K}}||\Xi_{2}||\notag
\end{align}
where $\Xi_{1}$ and $\Xi_{2}$ are defined in equations \eqref{Xi1} and \eqref{Xi2} in the proof of equation \eqref{3.eqA1}. The second inequality holds because all the elements in $\Xi_{1}$ are equal to $u-u_{k}$, whose magnitude is bounded by $\varepsilon$ and the spectral norm of a matrix of all ones is equal to $\sqrt{NT}$. The last equality is by $\varepsilon=1/\sqrt{N\lor T}$. Hence,
\begin{equation}\label{eqA2.5}
\mathbb{P}(\mathcal{F}_{2}>M)\leq \mathbb{P}\left(\sup_{|u-u_{k}|\leq \varepsilon,u_{k}\in\mathcal{U}_{K}}||\Xi_{2}||>M-\sqrt{N\land T}\right)
\end{equation}

Now we bound $\sup_{|u-u_{k}|\leq \varepsilon,u_{k}\in\mathcal{U}_{K}}||\Xi_{2}||$. By definition, for a generic matrix $N\times T$ matrix $A$, $||A||\coloneqq\sup_{||x||_{F}=1}||Ax||_{F}$ where $x$ is a $T\times 1$ vector. Suppose all the elements in $A$ have the same sign. Then, the supremum is achieved only if all the elements in $x$ also have the same sign and thus $\sup_{||x||_{F}=1}||Ax||_{F}= \sup_{||x||_{F}=1}\big|\big|A\cdot |x|\big|\big|_{F}$. Meanwhile, for a matrix $B$ whose elements also have the same sign with $|B_{it}|\geq |A_{it}|$ for all $i$ and $t$, we have $\big|\big|A\cdot |x|\big|\big|_{F}\leq\big|\big|B\cdot |x|\big|\big|_{F}$. Therefore,
\begin{equation}
||A||=\sup_{||x||_{F}=1}\big|\big|A\cdot |x|\big|\big|_{F}\leq \sup_{||x||_{F}=1}\big|\big|B\cdot |x|\big|\big|_{F}= \sup_{||x||_{F}=1}||Bx||_{F}=||B||\label{eqA2.key}
\end{equation}
Hence,
\begin{align}
\sup_{|u-u_{k}|\leq \varepsilon,u_{k}\in\mathcal{U}_{K}} ||\Xi_{2}||\leq 
& \sup_{\substack{u_{k}-\varepsilon\leq u\leq u_{k}\\u_{k}\in\mathcal{U}_{k}}}||\Xi_{2}||+\sup_{\substack{u_{k}\leq u\leq u_{k}+\varepsilon\\u_{k}\in\mathcal{U}_{k}}}||\Xi_{2}||\notag\\
\leq &  \max_{u_{k}\in\mathcal{U}_{k}}||\Xi_{2}^{(1)}||+\max_{u_{k}\in\mathcal{U}_{k}}||\Xi_{2}^{(2)}||\label{eqA2.6}
\end{align}
where $\Xi_{2}^{(1)}$ and $\Xi_{2}^{(2)}$ are defined in equations \eqref{Xi21} and \eqref{Xi22} in the proof of equation \eqref{3.eqA1} and do not depend on $u$. To see why the second inequality holds, recall that the elements in $\Xi_{2}$ are all nonnegative when $u<u_{k}$ and all nonpositive when $u>u_{k}$, and in either case, we have $|\Xi_{2,it}|\leq |\Xi_{2,it}^{(\iota)}|,\iota=1,2$ for all $i,t$. Inequality \eqref{eqA2.6} is thus implied by inequality \eqref{eqA2.key}.

Again, let us only derive the bound on $\max_{u_{k}\in\mathcal{U}_{k}}||\Xi_{2}^{(1)}||$ because the bound on $\max_{u_{k}\in\mathcal{U}_{k}}||\Xi_{2}^{(2)}||$ follows the same argument. Recall that matrix $\bar{\Xi}_{2}^{(1)}\coloneqq (\varepsilon)_{i,t}$ is the conditional mean of $\Xi_{2}^{(1)}$ given $W$. Note that $\max_{u_{k}\in\mathcal{U}_{K}}||\bar{\Xi}_{2}^{(1)}||= \varepsilon\sqrt{NT}=\sqrt{N\land T}$ by $\varepsilon=1/\sqrt{N\lor T}$. Recall that $C_{op}>2$. Therefore, there exists universal constant $0<c''\leq (C_{op}-2)$ and $c'''>0$ such that
\begin{align}
&\mathbb{P}\left(\max_{u_{k}\in\mathcal{U}_{k}}||\Xi_{2}^{(1)}||>M-\sqrt{N\land T}\right)\notag\\
\leq& \mathbb{P}\left(\max_{u_{k}\in\mathcal{U}_{k}}||\Xi_{2}^{(1)}-\bar{\Xi}_{2}^{(1)}||+ \max_{u_{k}\in\mathcal{U}_{k}}||\bar{\Xi}_{2}^{(1)}||>C_{op}\sqrt{N\lor T}-\sqrt{N\land T}\right)\notag\\
\leq &\mathbb{P}\left(\max_{u_{k}\in\mathcal{U}_{k}}||\Xi_{2}^{(1)}-\bar{\Xi}_{2}^{(1)}||>C_{op}\sqrt{N\lor T}-2\sqrt{N\land T}\right)\notag\\
 \leq &\sqrt{N\lor T}\max_{u_{k}\in\mathcal{U}_{K}}\mathbb{P}\left(||\Xi_{2}^{(1)}-\bar{\Xi}_{2}^{(1)}||>(C_{op}-2)\sqrt{N\lor T}\right)\notag\\
 =&\sqrt{N\lor T}\max_{u_{k}\in\mathcal{U}_{K}}\mathbb{E}\left[\mathbb{P}\left(||\Xi_{2}^{(1)}-\bar{\Xi}_{2}^{(1)}||>(C_{op}-2)\sqrt{N\lor T}\big |W\right)\right]\notag\\
 \leq & \frac{c''\sqrt{N\lor T}}{\exp\left[c'''(C_{op}-2)(N\lor T)\right]}\to 0\label{eqA2.7}
\end{align}
where the last inequality follows the same argument for equation \eqref{eqA2.3} since the elements in $(\Xi_{2}^{(1)}(u_{k})-\bar{\Xi}_{2}^{(1)}(u_{k}))$ are independent with zero mean and are bounded in $[-1,1]$ conditional on $W$. Similarly we have $\mathbb{P}\left(\max_{u_{k}\in\mathcal{U}_{k}}||\Xi_{2}^{(2)}||>M-\sqrt{N\land T}\right)\to 0$ as well. Combining it with equations \eqref{eqA2.2}, \eqref{eqA2.3}, \eqref{eqA2.5}, \eqref{eqA2.6} and \eqref{eqA2.7}, we have the desired result.\qed

\subsubsection{Proof of Lemma \ref{3.lemA3}}
For any fixed $w_{1}\in \mathbb{R}$, $\mathbbm{1}(w_{1}\leq z)$ is weakly increasing in $z$. So, if $z\geq 0$, we have $\mathbbm{1}(w_{1}\leq z)-\mathbbm{1}(w_{1}\leq 0)\geq 0$. Then, when $w_{2}>0$, for any $\kappa\in (0,1]$, 
\begin{align*}
&\int_{0}^{w_{2}}\big(\mathbbm{1}(w_{1}\leq z)-\mathbbm{1}(w_{1}\leq 0)\big)dz- \int_{0}^{\kappa w_{2}}\big(\mathbbm{1}(w_{1}\leq z)-\mathbbm{1}(w_{1}\leq 0)\big)dz\\
=& \int_{\kappa w_{2}}^{w_{2}}\big(\mathbbm{1}(w_{1}\leq z)-\mathbbm{1}(w_{1}\leq 0)\big)dz\geq 0
\end{align*}

When $w_{2}\leq 0$, note that $\int_{0}^{w_{2}}\big(\mathbbm{1}(w_{1}\leq z)-\mathbbm{1}(w_{1}\leq 0)\big)dz=\int_{w_{2}}^{0}\big(\mathbbm{1}(w_{1}\leq 0)-\mathbbm{1}(w_{1}\leq z)\big)dz$ and $\int_{0}^{\kappa w_{2}}\big(\mathbbm{1}(w_{1}\leq z)-\mathbbm{1}(w_{1}\leq 0)\big)dz=\int_{\kappa w_{2}}^{0}\big(\mathbbm{1}(w_{1}\leq 0)-\mathbbm{1}(w_{1}\leq z)\big)dz$. Now that $z\leq 0$, $\mathbbm{1}(w_{1}\leq 0)-\mathbbm{1}(w_{1}\leq z)\geq 0$. Therefore, 
\begin{align*}
&\int_{0}^{w_{2}}\big(\mathbbm{1}(w_{1}\leq z)-\mathbbm{1}(w_{1}\leq 0)\big)dz- \int_{0}^{\kappa w_{2}}\big(\mathbbm{1}(w_{1}\leq z)-\mathbbm{1}(w_{1}\leq 0)\big)dz\\
=& \int^{\kappa w_{2}}_{w_{2}}\big(\mathbbm{1}(w_{1}\leq 0)-\mathbbm{1}(w_{1}\leq z)\big)dz\geq 0
\end{align*}\qed

\section{Additional Simulation Results}\label{app.simulation}

We present the simulation results under $\phi=0.1$ and $0.3$ in this section. Recall that $\phi$ governs the correlation between the covariates and the common component. From the results, we can see that the patterns we find in Table \ref{tab.n.0.2} and Table \ref{tab.t.0.2} are preserved under these values of $\phi$. The penalized estimator of both the coefficients and the low-rank component converges to the true parameter values as $N$ and $T$ increase. In many cases, the performance of the penalized estimator and the iterative estimator are comparable, and the former sometimes performs better than the latter; for instance when $\phi=0.1,u=0.8$ and $(N,T)=(500,500)$, the penalized estimator has smaller bias in the coefficient estimator under both types of error distributions. Meanwhile, the penalized estimator runs much faster in all specifications. 

Now we compare the differences between $\text{MSE}_{L}$ and $\text{MSE}_{q}$ across different $\phi$s. This difference reflects how well the estimation error for the low-rank matrix can be separated from linear combinations of the covariates. Higher correlation between the covariates and the low-rank matrix results in larger difference between $\text{MSE}_{L}$ and $\text{MSE}_{q}$ because the restricted strong convexity constant $C_{RSC}$ tends to be smaller. From the results, when $\phi=0.1$ (Tables \ref{tab.n.0.1} and \ref{tab.t.0.1}), being the smallest among the three values of $\phi$, we can see that the differences between $\text{MSE}_{L}$ and $\text{MSE}_{q}$ for the penalized estimator are almost negligible. This is because a small $\phi$ results in a larger $C_{RSC}$. When $\phi=0.3$ (Tables \ref{tab.n.0.3} and \ref{tab.t.0.3}), $\text{MSE}_{L}$ becomes much larger than $\text{MSE}_{q}$, although both still shrink towards zero as $N$ and $T$ increase.
    \begin{table}[H]
\centering
\caption{Standard Normal Error. Parameter $\phi=0.1$}\label{tab.n.0.1}
\begin{tabular}{ccccccccccccc}
\hline
\hline
&&\multicolumn{3}{c}{$\text{Bias}^{2}_{\beta}\times 10^{2}$}&  \multicolumn{2}{c}{$\text{Var}_{\beta}\times 10^{4}$} &\multicolumn{2}{c}{$\text{MSE}_{L}$}&\multicolumn{2}{c}{$\text{MSE}_{q}$}&\multicolumn{2}{c}{\begin{tabular}{@{}c@{}}Time\\(second)
\end{tabular}}\\
\cline{2-13}
$u$& $(N,T)$ & Nu & It & Po  & Nu & It &Nu & It   & Nu & It & Nu & It \\
\hline
\multirow{10}{*}{0.2}&$(200,200)$&0.66&0.02&4.7&3.30&2.08&0.17&0.04&0.17&0.04&1&48\tabularnewline
&$(300,300)$&0.49&0.01&4.8&1.55&1.09&0.11&0.03&0.10&0.03&3&142\tabularnewline
&$(400,400)$&0.38&0.007&4.9&0.68&0.48&0.09&0.02&0.07&0.02&9&367\tabularnewline
&$(500,500)$&0.30&0.005&4.9&0.50&0.36&0.07&0.01&0.06&0.01&17&658\tabularnewline
\cline{2-13}
&$(200,300)$&0.68&0.02&4.9&2.07&1.38&0.15&0.03&0.13&0.03&2&78\tabularnewline
&$(300,200)$&0.66&0.02&4.8&2.17&1.24&0.15&0.03&0.13&0.03&2&79\tabularnewline
\cline{2-13}
&$(200,400)$&0.64&0.01&4.8&1.74&0.95&0.14&0.03&0.12&0.03&3&118\tabularnewline
&$(400,200)$&0.63&0.01&4.9&1.55&1.15&0.14&0.03&0.12&0.03&3&120\tabularnewline
\cline{2-13}
&$(200,500)$&0.63&0.01&5.0&1.21&0.79&0.14&0.03&0.11&0.03&4&166\tabularnewline
&$(500,200)$&0.64&0.01&5.0&1.34&0.76&0.14&0.03&0.11&0.03&4&165\tabularnewline
\hline
\multirow{10}{*}{0.5}&$(200,200)$&0.05&0.04&8.3&2.52&2.53&0.27&0.16&0.25&0.16&1&34\tabularnewline
&$(300,300)$&0.03&0.02&8.6&0.95&1.09&0.14&0.14&0.13&0.14&3&88\tabularnewline
&$(400,400)$&0.02&0.007&8.3&0.45&0.60&0.09&0.13&0.08&0.13&7&207\tabularnewline
&$(500,500)$&0.02&0.006&8.4&0.33&0.49&0.07&0.12&0.06&0.12&12&347\tabularnewline
\cline{2-13}
&$(200,300)$&0.04&0.03&8.4&1.30&1.45&0.14&0.15&0.13&0.14&1&59\tabularnewline
&$(300,200)$&0.05&0.03&8.5&1.28&1.54&0.15&0.15&0.13&0.15&1&55\tabularnewline
\cline{2-13}
&$(200,400)$&0.04&0.03&8.7&0.93&1.21&0.11&0.15&0.09&0.14&2&80\tabularnewline
&$(400,200)$&0.04&0.02&8.5&1.03&1.28&0.11&0.15&0.09&0.15&2&72\tabularnewline
\cline{2-13}
&$(200,500)$&0.03&0.02&8.3&0.86&1.08&0.09&0.13&0.08&0.13&3&118\tabularnewline
&$(500,200)$&0.04&0.02&8.3&0.79&1.14&0.10&0.14&0.08&0.14&3&106\tabularnewline
\hline
\multirow{10}{*}{0.8}&$(200,200)$&0.05&0.21&14.0&5.92&5.07&0.82&0.52&0.72&0.46&2&73\tabularnewline
&$(300,300)$&0.02&0.06&13.8&2.36&1.97&0.38&0.26&0.34&0.24&6&167\tabularnewline
&$(400,400)$&0.01&0.03&13.9&1.08&1.17&0.24&0.19&0.22&0.19&17&448\tabularnewline
&$(500,500)$&0.008&0.02&14.1&0.50&0.60&0.17&0.17&0.16&0.16&32&880\tabularnewline
\cline{2-13}
&$(200,300)$&0.04&0.14&14.0&3.90&3.12&0.53&0.37&0.47&0.33&3&111\tabularnewline
&$(300,200)$&0.03&0.09&13.6&3.53&3.22&0.53&0.35&0.47&0.32&3&104\tabularnewline
\cline{2-13}
&$(200,400)$&0.03&0.12&14.0&2.78&2.59&0.45&0.31&0.40&0.28&5&154\tabularnewline
&$(400,200)$&0.03&0.06&13.8&2.68&2.08&0.45&0.30&0.40&0.28&5&152\tabularnewline
\cline{2-13}
&$(200,500)$&0.03&0.10&14.0&1.99&1.93&0.42&0.28&0.38&0.26&7&203\tabularnewline
&$(500,200)$&0.03&0.05&14.0&2.28&1.88&0.42&0.27&0.37&0.26&7&200\tabularnewline
\hline
\multicolumn{13}{l}{\footnotesize \textit{Note}: Columns Nu, It and Po report the results of the nuclear norm penalized estimator proposed in}\tabularnewline
\multicolumn{13}{l}{\footnotesize this paper, the iterative estimator \eqref{estimatorIt} and the pooled estimator \eqref{estimatorPo}, averaged over 100 simulations.}
     \end{tabular}
    \end{table}
    
  \begin{table}[H]
\centering
\caption{Student $t$-Distributed Error (Degree of Freedom$=2$). Parameter $\phi=0.1$}\label{tab.t.0.1}
\begin{tabular}{ccccccccccccc}
\hline
\hline
&&\multicolumn{3}{c}{$\text{Bias}^{2}_{\beta}\times 10^{2}$}&  \multicolumn{2}{c}{$\text{Var}_{\beta}\times 10^{4}$} &\multicolumn{2}{c}{$\text{MSE}_{L}$}&\multicolumn{2}{c}{$\text{MSE}_{q}$}&\multicolumn{2}{c}{\begin{tabular}{@{}c@{}}Time\\(second)
\end{tabular}}\\
\cline{2-13}
$u$& $(N,T)$ & Nu & It & Po  & Nu & It &Nu & It   & Nu & It & Nu & It \\
\hline
\multirow{10}{*}{0.2}&$(200,200)$&1.15&0.04&4.0&6.96&4.57&0.34&0.09&0.35&0.09&1&42\tabularnewline
&$(300,300)$&0.88&0.03&4.3&2.72&1.90&0.22&0.05&0.21&0.05&5&130\tabularnewline
&$(400,400)$&0.65&0.02&4.2&1.39&1.13&0.16&0.04&0.15&0.04&12&363\tabularnewline
&$(500,500)$&0.52&0.01&4.2&0.77&0.74&0.13&0.03&0.12&0.03&24&659\tabularnewline
\cline{2-13}
&$(200,300)$&1.11&0.03&4.2&4.82&3.03&0.27&0.07&0.26&0.07&2&75\tabularnewline
&$(300,200)$&1.11&0.03&4.3&4.44&2.88&0.27&0.07&0.26&0.07&2&72\tabularnewline
\cline{2-13}
&$(200,400)$&1.09&0.03&4.4&3.43&2.35&0.25&0.06&0.23&0.06&3&114\tabularnewline
&$(400,200)$&1.05&0.03&4.2&3.42&2.21&0.25&0.06&0.23&0.06&4&109\tabularnewline
\cline{2-13}
&$(200,500)$&1.01&0.03&4.3&2.47&1.58&0.25&0.05&0.22&0.05&5&158\tabularnewline
&$(500,200)$&1.03&0.03&4.2&2.47&1.52&0.25&0.06&0.22&0.05&4&150\tabularnewline
\hline
\multirow{10}{*}{0.5}&$(200,200)$&0.07&0.05&8.8&2.97&2.42&0.33&0.17&0.30&0.17&1&36\tabularnewline
&$(300,300)$&0.04&0.02&8.8&1.13&1.27&0.17&0.15&0.15&0.14&4&89\tabularnewline
&$(400,400)$&0.03&0.01&8.9&0.57&0.74&0.12&0.14&0.10&0.14&10&222\tabularnewline
&$(500,500)$&0.02&0.006&8.9&0.36&0.49&0.09&0.13&0.08&0.13&20&338\tabularnewline
\cline{2-13}
&$(200,300)$&0.06&0.04&9.0&1.58&1.77&0.18&0.16&0.16&0.16&2&56\tabularnewline
&$(300,200)$&0.06&0.03&9.0&1.69&1.76&0.18&0.17&0.16&0.16&2&54\tabularnewline
\cline{2-13}
&$(200,400)$&0.05&0.03&8.8&1.29&1.40&0.14&0.16&0.11&0.15&3&71\tabularnewline
&$(400,200)$&0.05&0.03&8.8&1.25&1.47&0.14&0.16&0.11&0.15&3&80\tabularnewline
\cline{2-13}
&$(200,500)$&0.05&0.03&8.9&0.98&1.26&0.12&0.15&0.09&0.15&4& 100\tabularnewline
&$(500,200)$&0.05&0.02&8.9&0.90&1.17&0.12&0.16&0.09&0.15&4&104\tabularnewline
\hline
\multirow{10}{*}{0.8}&$(200,200)$&0.07&0.46&14.3&11&11&1.06&0.90&0.93&0.81&2&73\tabularnewline
&$(300,300)$&0.03&0.13&14.3&3.73&3.94&0.51&0.46&0.45&0.44&7&249\tabularnewline
&$(400,400)$&0.02&0.05&13.9&2.00&2.28&0.33&0.29&0.29&0.28&17&514\tabularnewline
&$(500,500)$&0.01&0.03&14.2&0.86&1.07&0.24&0.23&0.22&0.23&33&830\tabularnewline
\cline{2-13}
&$(200,300)$&0.06&0.31&14.2&6.33&6.44&0.68&0.70&0.59&0.63&3&143\tabularnewline
&$(300,200)$&0.06&0.24&14.2&6.11&6.53&0.70&0.60&0.60&0.64&3&136\tabularnewline
\cline{2-13}
&$(200,400)$&0.05&0.25&14.5&4.47&4.65&0.57&0.57&0.49&0.52&5&226\tabularnewline
&$(400,200)$&0.05&0.16&14.0&4.42&4.40&0.56&0.55&0.49&0.51&5&205\tabularnewline
\cline{2-13}
&$(200,500)$&0.05&0.21&14.4&3.15&3.50&0.53&0.50&0.46&0.46&7&291\tabularnewline
&$(500,200)$&0.04&0.11&13.5&3.43&3.50&0.53&0.46&0.46&0.44&7&292\tabularnewline
\hline
\multicolumn{13}{l}{\footnotesize \textit{Note}: Columns Nu, It and Po report the results of the nuclear norm penalized estimator proposed in}\tabularnewline
\multicolumn{13}{l}{\footnotesize this paper, the iterative estimator \eqref{estimatorIt} and the pooled estimator \eqref{estimatorPo}, averaged over 100 simulations.}
     \end{tabular}
    \end{table}

      \begin{table}[H]
\centering
\caption{Standard Normal Error. Parameter $\phi=0.3$}\label{tab.n.0.3}
\begin{tabular}{ccccccccccccc}
\hline
\hline
&&\multicolumn{3}{c}{$\text{Bias}^{2}_{\beta}\times 10^{2}$}&  \multicolumn{2}{c}{$\text{Var}_{\beta}\times 10^{4}$} &\multicolumn{2}{c}{$\text{MSE}_{L}$}&\multicolumn{2}{c}{$\text{MSE}_{q}$}&\multicolumn{2}{c}{\begin{tabular}{@{}c@{}}Time\\(second)
\end{tabular}}\\
\cline{2-13}
$u$& $(N,T)$ & Nu & It & Po  & Nu & It &Nu & It   & Nu & It & Nu & It \\
\hline
\multirow{10}{*}{0.2}&$(200,200)$&4.8&0.19&14.6&5.08&2.03&0.55&0.06&0.31&0.05&1&61\tabularnewline
&$(300,300)$&3.75&0.10&14.8&2.40&0.97&0.42&0.03&0.18&0.03&4&177\tabularnewline
&$(400,400)$&2.99&0.06&14.9&1.51&0.47&0.35&0.02&0.12&0.02&10&474\tabularnewline
&$(500,500)$&2.44&0.04&15.0&1.09&0.31&0.30&0.02&0.09&0.01&20&872\tabularnewline
\cline{2-13}
&$(200,300)$&4.71&0.15&14.9&4.02&1.46&0.52&0.05&0.22&0.03&2&102\tabularnewline
&$(300,200)$&4.73&0.15&14.8&3.84&1.41&0.52&0.05&0.22&0.04&2&102\tabularnewline
\cline{2-13}
&$(200,400)$&4.53&0.13&14.9&3.38&0.99&0.51&0.04&0.20&0.03&3&154\tabularnewline
&$(400,200)$&4.54&0.13&14.8&3.64&1.08&0.51&0.04&0.20&0.03&3&152\tabularnewline
\cline{2-13}
&$(200,500)$&4.52&0.11&15.0&2.86&0.80&0.52&0.04&0.19&0.03&4&219\tabularnewline
&$(500,200)$&4.50&0.11&14.9&2.83&0.79&0.51&0.04&0.19&0.03&4&214\tabularnewline
\hline
\multirow{10}{*}{0.5}&$(200,200)$&0.92&0.34&29.0&5.15&5.05&0.59&0.19&0.28&0.12&1&70\tabularnewline
&$(300,300)$&0.48&0.14&28.7&1.52&3.60&0.31&0.13&0.14&0.10&3&183\tabularnewline
&$(400,400)$&0.35&0.08&28.7&0.73&1.89&0.22&0.10&0.09&0.09&9&473\tabularnewline
&$(500,500)$&0.29&0.05&28.7&0.62&1.75&0.17&0.09&0.07&0.08&16&752\tabularnewline
\cline{2-13}
&$(200,300)$&0.66&0.23&28.8&1.88&4.09&0.37&0.16&0.15&0.11&2&110\tabularnewline
&$(300,200)$&0.65&0.20&28.8&2.18&4.32&0.37&0.16&0.15&0.12&2&101\tabularnewline
\cline{2-13}
&$(200,400)$&0.62&0.20&28.9&1.78&3.95&0.32&0.15&0.11&0.10&3&163\tabularnewline
&$(400,200)$&0.61&0.21&28.7&1.78&3.87&0.32&0.14&0.11&0.10&3&173\tabularnewline
\cline{2-13}
&$(200,500)$&0.59&0.17&29.0&1.27&3.28&0.30&0.14&0.09&0.10&4&232\tabularnewline
&$(500,200)$&0.58&0.18&28.8&1.43&3.26&0.29&0.13&0.09&0.09&4&237\tabularnewline
\hline
\multirow{10}{*}{0.8}&$(200,200)$&0.55&0.58&44.7&6.23&5.52&1.33&0.77&0.73&0.50&2&74\tabularnewline
&$(300,300)$&0.28&0.18&44.0&2.29&2.36&0.64&0.34&0.35&0.26&7&184\tabularnewline
&$(400,400)$&0.17&0.09&43.7&1.11&1.36&0.38&0.23&0.22&0.19&19&440\tabularnewline
&$(500,500)$&0.10&0.05&43.4&0.63&0.88&0.26&0.19&0.16&0.17&36&1000\tabularnewline
\cline{2-13}
&$(200,300)$&0.41&0.33&44.1&4.93&4.57&0.90&0.51&0.48&0.36&4&114\tabularnewline
&$(300,200)$&0.47&0.32&44.3&3.76&3.43&0.93&0.51&0.48&0.36&4&115\tabularnewline
\cline{2-13}
&$(200,400)$&0.40&0.29&44.0&2.77&3.68&0.77&0.42&0.40&0.30&6&177\tabularnewline
&$(400,200)$&0.38&0.22&44.0&2.83&3.26&0.77&0.40&0.40&0.30&6&159\tabularnewline
\cline{2-13}
&$(200,500)$&0.37&0.23&43.6&2.39&2.46&0.72&0.36&0.38&0.27&8&232\tabularnewline
&$(500,200)$&0.37&0.20&43.3&2.37&2.45&0.72&0.36&0.38&0.27&8&202\tabularnewline
\hline
\multicolumn{13}{l}{\footnotesize \textit{Note}: Columns Nu, It and Po report the results of the nuclear norm penalized estimator proposed in}\tabularnewline
\multicolumn{13}{l}{\footnotesize this paper, the iterative estimator \eqref{estimatorIt} and the pooled estimator \eqref{estimatorPo}, averaged over 100 simulations.}
     \end{tabular}
    \end{table}
    
      \begin{table}[H]
\centering
\caption{Student $t$-Distributed Error (Degree of Freedom$=2$). Parameter $\phi=0.3$}\label{tab.t.0.3}
\begin{tabular}{ccccccccccccc}
\hline
\hline
&&\multicolumn{3}{c}{$\text{Bias}^{2}_{\beta}\times 10^{2}$}&  \multicolumn{2}{c}{$\text{Var}_{\beta}\times 10^{4}$} &\multicolumn{2}{c}{$\text{MSE}_{L}$}&\multicolumn{2}{c}{$\text{MSE}_{q}$}&\multicolumn{2}{c}{\begin{tabular}{@{}c@{}}Time\\(second)
\end{tabular}}\\
\cline{2-13}
$u$& $(N,T)$ & Nu & It & Po  & Nu & It &Nu & It   & Nu & It & Nu & It \\
\hline
\multirow{10}{*}{0.2}&$(200,200)$&5.58&0.41&13.4&7.48&4.20&0.68&0.15&0.48&0.11&1&58\tabularnewline
&$(300,300)$&4.56&0.20&13.5&3.30&1.85&0.52&0.07&0.28&0.06&4&171\tabularnewline
&$(400,400)$&3.74&0.11&13.5&2.20&0.95&0.44&0.05&0.20&0.04&13&466\tabularnewline
&$(500,500)$&3.21&0.07&13.7&1.45&0.71&0.39&0.04&0.16&0.03&27&881\tabularnewline
\cline{2-13}
&$(200,300)$&5.38&0.30&13.6&5.31&2.48&0.61&0.11&0.33&0.08&3&99\tabularnewline
&$(300,200)$&5.43&0.31&13.6&5.63&2.95&0.62&0.11&0.33&0.08&2&93\tabularnewline
\cline{2-13}
&$(200,400)$&5.24&0.24&13.7&4.86&2.16&0.60&0.09&0.29&0.07&3&153\tabularnewline
&$(400,200)$&5.19&0.24&13.6&4.51&1.95&0.59&0.09&0.29&0.07&3&147\tabularnewline
\cline{2-13}
&$(200,500)$&5.19&0.22&13.7&4.20&1.67&0.60&0.08&0.28&0.06&5&212\tabularnewline
&$(500,200)$&5.04&0.22&13.5&4.32&1.79&0.59&0.08&0.27&0.06&5&208\tabularnewline
\hline
\multirow{10}{*}{0.5}&$(200,200)$&1.21&0.38&29.5&7.58&6.24&0.75&0.22&0.35&0.15&1&68\tabularnewline
&$(300,300)$&0.66&0.17&29.5&1.76&3.57&0.40&0.15&0.18&0.11&4&176\tabularnewline
&$(400,400)$&0.44&0.09&29.4&0.95&2.66&0.27&0.12&0.12&0.10&11&441\tabularnewline
&$(500,500)$&0.35&0.05&29.5&0.62&1.85&0.21&0.11&0.09&0.10&22&666\tabularnewline
\cline{2-13}
&$(200,300)$&0.86&0.25&29.5&2.70&4.65&0.48&0.18&0.18&0.12&2&107\tabularnewline
&$(300,200)$&0.86&0.28&29.5&3.27&4.86&0.47&0.17&0.18&0.12&2&111\tabularnewline
\cline{2-13}
&$(200,400)$&0.81&0.21&29.9&2.12&4.13&0.41&0.17&0.13&0.12&3&150\tabularnewline
&$(400,200)$&0.80&0.24&29.8&2.59&4.69&0.41&0.16&0.13&0.11&3&161\tabularnewline
\cline{2-13}
&$(200,500)$&0.72&0.20&29.5&1.65&3.53&0.36&0.15&0.11&0.11&5&229\tabularnewline
&$(500,200)$&0.70&0.19&29.5&2.03&4.57&0.35&0.15&0.11&0.11&4&220\tabularnewline
\hline
\multirow{10}{*}{0.8}&$(200,200)$&0.82&1.89&45.1&11&20&1.80&1.60&0.95&0.91&2&88\tabularnewline
&$(300,300)$&0.47&0.44&46.7&4.08&7.58&0.91&0.63&0.47&0.47&7&276\tabularnewline
&$(400,400)$&0.27&0.17&45.7&2.03&2.66&0.55&0.35&0.30&0.30&20&489\tabularnewline
&$(500,500)$&0.19&0.10&45.7&1.10&1.67&0.39&0.26&0.22&0.23&39&788\tabularnewline
\cline{2-13}
&$(200,300)$&0.66&0.97&45.3&7.16&9.50&1.19&1.04&0.59&0.70&4&148\tabularnewline
&$(300,200)$&0.63&0.86&45.2&6.72&9.41&1.17&1.00&0.59&0.69&4&146\tabularnewline
\cline{2-13}
&$(200,400)$&0.58&0.66&45.3&5.10&7.45&1.02&0.82&0.50&0.58&6&252\tabularnewline
&$(400,200)$&0.60&0.59&45.5&4.22&8.92&1.03&0.79&0.50&0.57&5&207\tabularnewline
\cline{2-13}
&$(200,500)$&0.59&0.56&45.7&4.12&5.58&0.98&0.69&0.47&0.50&8&329\tabularnewline
&$(500,200)$&0.59&0.45&45.9&4.06&5.10&0.98&0.66&0.47&0.50&8&296\tabularnewline
\hline
\multicolumn{13}{l}{\footnotesize \textit{Note}: Columns Nu, It and Po report the results of the nuclear norm penalized estimator proposed in}\tabularnewline
\multicolumn{13}{l}{\footnotesize this paper, the iterative estimator \eqref{estimatorIt} and the pooled estimator \eqref{estimatorPo}, averaged over 100 simulations.}
     \end{tabular}
    \end{table}
\end{appendices}
\bibliography{/Users/Ecthelion/Documents/Econometrics/RobustPCA/paper/travel/bibqr}
\end{document}